\documentclass[12pt]{article}
\usepackage{mathtools,amssymb,mathrsfs,ytableau,microtype,tikz,slashed}
\usepackage{physics}
\usepackage{amsthm}
\usepackage{array}
\usepackage{amsmath}

\newcolumntype{C}[1]{%
  >{\centering\arraybackslash}p{#1}%
}
\usepackage{color}
\usepackage[linktocpage]{hyperref}
\hypersetup{colorlinks=true,citecolor=red,linkcolor=blue, urlcolor=blue, breaklinks=true}
\usepackage{cite}
\usepackage{moresize}
\usepackage{url}
\ytableausetup{centertableaux,boxsize=.5em}
\usepackage[numbers,sort]{natbib}

\usepackage{tabularx}

\usepackage{array}

\newcolumntype{Y}{>{\centering\arraybackslash}X}

\newcolumntype{L}{>{\raggedright\arraybackslash}X}

\usepackage{booktabs}
\newtheorem{theorem}{Theorem}[section]

\newtheorem{conjecture}{Conjecture}[section]

\makeatletter\@addtoreset{equation}{section}\makeatother

\renewcommand{\title}[1]{\vbox{\center\LARGE{#1}}\vspace{5mm}}
\renewcommand{\author}[1]{\vbox{\center\large#1}\vspace{5mm}}
\newcommand{\address}[1]{\vbox{\center\em#1}}

\begin{document}

\begin{titlepage}

\begin{center}

\vskip 0.5cm

 \title{
Quantum Cellular Automata from Kramers-Wannier Dualities and Modular Relations
}

 \author{
  Carolyn Zhang$^{1}$, Po-Shen Hsin$^{2}$,
 }

\address{${}^1$Department of Physics and Astronomy, and Quantum Matter Institute,
University of British Columbia, Vancouver, BC, Canada V6T 1Z1}

\address{${}^2$ Department of Mathematics, King’s College London, Strand, London WC2R 2LS, UK.}

\end{center}

%


\abstract{

Recent work has constructed higher-dimensional analogs of non-invertible symmetries similar to 1+1d Kramers-Wannier duality. Although their continuum descriptions often treat purely gravitational topological terms as inessential counterterms, these terms can have an essential lattice manifestation: they distinguish states prepared by finite-depth quantum circuits (FDQCs) from those entangled by nontrivial quantum cellular automata (QCAs). Motivated by this mismatch, we show that QCAs associated with gravitational topological responses arise in several related settings: (1) lattice realizations of projective $\mathrm{SL}(2,\mathbb{Z}_N)$ relations generated by topological operations on symmetries; (2) squares of dualities that generalize the relation between fermionization and Kramers–Wannier duality; (3) lattice implementations of QCAs through higher-form gauging; and (4) invertible phases protected by generalized time-reversal symmetries. We derive new projective $\mathrm{SL}(2,\mathbb{Z}_N)$ relations whose projective phases are gravitational topological responses constructed from Stiefel–Whitney classes. We furthermore give a general protocol for preparing the associated QCA-entangled states using finite-depth unitary circuits, measurements, and error correction. These results unify the study of gravitational topological responses in field theories, higher dimensional dualities, and quantum cellular automata.
}

\vfill

\today

\vfill

\end{titlepage}

\eject

\tableofcontents

\unitlength = .8mm

\setcounter{tocdepth}{3}

\section{Introduction}

Generalized symmetries have emerged as a powerful framework for studying constraints on quantum many-body systems. In addition to their applications, their more fundamental classification has also seen immense recent progress. In particular, various works developed the first examples of non-invertible Kramers-Wannier like symmetries in higher dimensions~\cite{Koide:2021zxj,Choi:2021kmx,Kaidi:2021xfk,Choi:2022rfe} (see e.g.~Refs.~\cite{Shao:2023gho,Schafer-Nameki:2023jdn} for introduction and other references to the literature). For example, in 3+1d systems with $1$-form symmetry, gauging maps the $1$-form symmetry to a dual $1$-form symmetry. Sometimes gauging maps between different theories with $1$-form symmetry and in special cases, gauging maps a theory back to itself. Such theories are symmetric under the non-invertible Kramers-Wannier-like symmetry.

Much of the work on these higher dimensional non-invertible symmetries were initiated in field theory. Recently, explicit lattice operators implementing these non-invertible transformations were obtained in 1+1d and 3+1d (see e.g.~Refs.~\cite{Aasen:2020jwb,Lootens:2021tet,Koide:2021zxj,Choi:2021kmx,Seiberg:2023cdc,Seifnashri:2024dsd,Bhardwaj:2024kvy,Hsin:2024aqb,Chatterjee:2024ych,Seiberg:2024gek,Gorantla:2024ocs,ParayilMana:2024txy,Cao:2024qjj,Hsin:2025ria,Inamura:2025cum,Wen:2026ncw}). In 1+1d, it is known that Kramers-Wannier duality for $\mathbb{Z}_2$ 0-form symmetry mixes with the translation operator on tensor product Hilbert spaces. The translation operator is an example of a nontrivial quantum cellular automaton (QCA) in 1+1d, which is a locality-preserving unitary that cannot be written as a finite-depth quantum circuit (FDQC)~\cite{CSLent_1993,Gross:2011yvb,Haah:2018jdf,Arrighi:2019uor,Haah:2019fqd,Farrelly:2019zds,Shirley:2022lhu}\footnote{Here by locality-preserving we mean strict locality-preserving, i.e. mapping strictly local operators to nearby strictly local operators. A FDQC is a trivial QCA.}. 
 QCAs are also related to periodically-driven Floquet phases (see e.g.~Refs.~\cite{Po:2016qlt,Po_2017,Farrelly:2019zds,Fidkowski_2019}), and demonstrate a rich mathematical classification structure.

\subsection{A puzzle in 3+1d for gauging 1-form symmetry}

In 3+1d, gauging 1-form symmetry has the following puzzle. For concreteness, let's discuss gauging $\mathbb{Z}_3$ $1$-form symmetry in 1-form protected topological (SPT) phase \cite{Gaiotto:2014kfa,Hsin:2018vcg,Tsui:2019ykk}. Such phases can be described by the following topological action
\begin{equation}
    \frac{2\pi p}{3}\int B\cup B~,
\end{equation}
where $p$ is an integer, and $B$ is the 2-form background $\mathbb{Z}_3$ gauge field.
We will focus on $p=2$, i.e. the topological action $\frac{4\pi }{3} \int B\cup B$.
If we gauge the 1-form symmetry in the SPT phase, there is a dual $\mathbb{Z}_3$ 1-form symmetry. Moreover, since $p=2$ is coprime with 3, the resulting theory is still an invertible topological field theory, and we can ask what is the SPT for the new $\mathbb{Z}_3$ 1-form symmetry. The SPT can be inferred from an explicit computation (see e.g.~Refs.~\cite{Choi:2021kmx,Choi:2022zal,Choi:2022rfe}):
\begin{equation}
\sum_b\exp \;i\int\left(   \frac{4\pi}{3}b\cup b+\frac{2\pi}{3} b\cup B\right)=\sum_{b'}\exp\left(i\frac{4\pi}{3}\int b'\cup b'\right)\exp\left(\frac{2\pi i}{3}\int B\cup B\right)~,
\end{equation}
where $b$ is the dynamical $\mathbb{Z}_3$ gauge field for gauging the 1-form symmetry, we have made a change of variables $b'=b+B$, and the first term on the right hand side is a gravitational term that is often ignored.
Thus we conclude that the new SPT phase from gauging the 1-form symmetry is the SPT with $p=1$, i.e. the topological term $\frac{2\pi i}{3} \int B\cup B$.

However, if we study the gauging of 1-form symmetry in the 1-form SPT on the lattice, we do not quite find the claimed SPT.
Below, in Fig.~\ref{fig:Z3example}, we illustrate the stabilizers for the SPT $p=2$ before and after gauging, as well as the stabilizers for the SPT $p=1$ (see e.g.~Ref.~\cite{Tsui:2019ykk}). Clearly, the gauged version of $p=2$ (see e.g.~Ref.~\cite{Cordova:2023bja}) does not match the SPT $p=1$ on the nose. It turns out that the stabilizer groups are not even small deformations of each other: they cannot be related by a FDQC, even if one includes translations. 

\begin{figure}[h]
   \centering
   \includegraphics[width=.8\columnwidth]{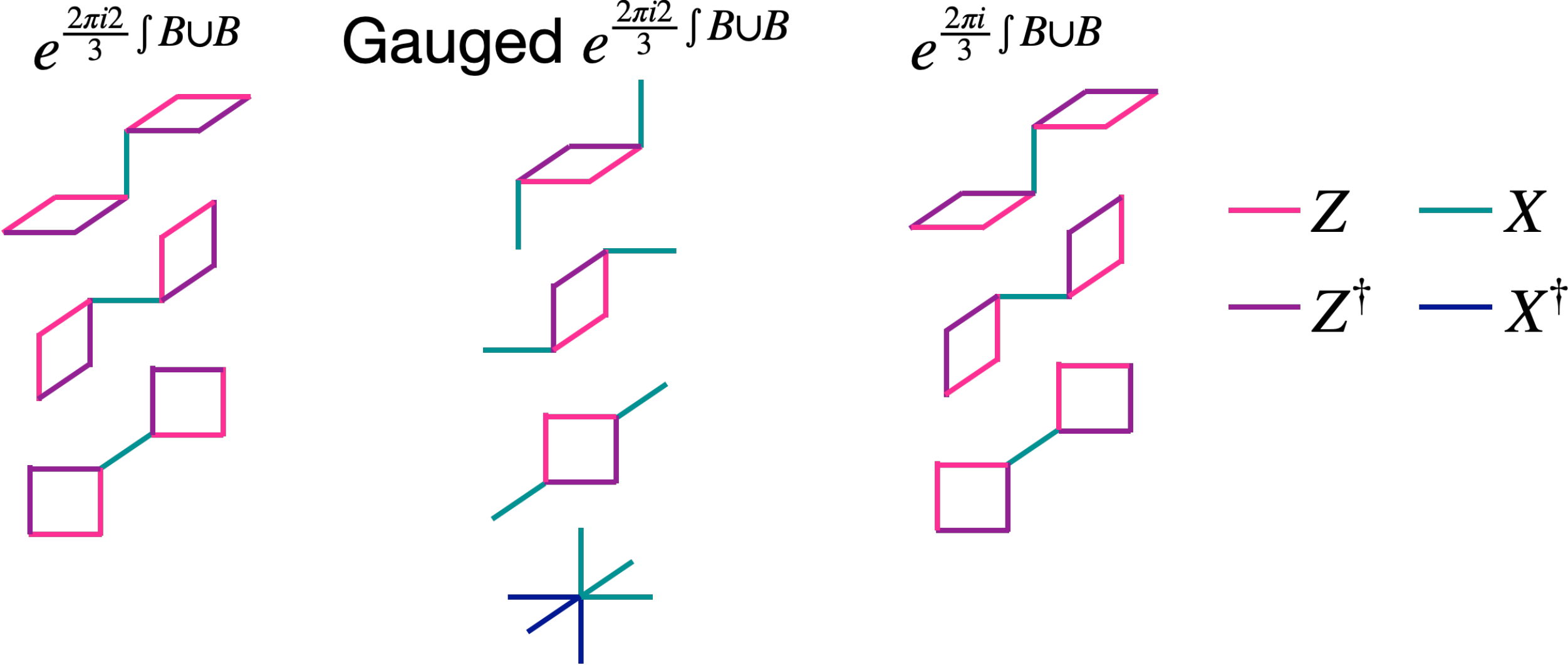} 
   \caption{Stabilizer group on a cubic lattice for the $\mathbb{Z}_3$ $1$-form SPT with $p=2$ (i.e. the partition function $e^{\frac{4\pi i}{3}\int B\cup B}$), its gauged version, and the SPT with $p=1$ (i.e. the partition function $e^{\frac{2\pi i }{3}\int B\cup B}$) on a cubic lattice. The lattice models between the middle and right column clearly mismatch. In field theory, the two models differ by a gravitational topological response.}
   \label{fig:Z3example}
   \end{figure}

The resolution for this puzzle is as follows. In field theory, we often drop the gravitational topological responses especially those that can be continuously tuned to zero, as in the gravitational instanton term. However, in lattice models, gravitational topological responses cannot be neglected. The lattice model for the SPT $p=1$ is entangled from a product state by a FDQC, and has a parent Hamiltonian which is a set of stabilizers obtained from product state stabilizers by the same FDQC. On the other hand, the gauged version of the $p=2$ model describes a state entangled by a nontrivial QCA. In the particular case at hand, it is a Walker-Wang model \cite{walker201131tqftstopologicalinsulators,von_Keyserlingk_2013} with a boundary theory described by the $SU(3)_{1}$ Chern Simons theory, or equivalently the $U(1)\times U(1)$ Chern-Simons theory with $K$ matrix $\left(\begin{array}{cc}
    2 &-1  \\
     -1&2 
\end{array}\right)$
.
The gravitational topological response allows the boundary theory to describe a chiral topological order. Under standard arguments \cite{Haah:2018jdf}, there is no FDQC that produces a stabilizer parent Hamiltonian for this state.

\subsection{Interplay between topological operations and QCAs}

In this paper, we study more generally the ways in which these higher dimensional topological operators mix with (nontrivial) QCA in lattice models. In particular, in studying the subtle appearance of gravitational topological responses on the lattice, we draw connections between:
\begin{itemize}
    \item[1. ] Projective representations of $\mathrm{SL}(2,\mathbb{Z}_N)$  generated by topological operations $S$ and $T$, which correspond to gauging abelian symmetries and stacking with SPTs respectively. Here, gravitational topological responses serve as the projective $U(1)$ phase factor in the $(ST)^3=1$ modular relation: denote by $Y$ the gravitational topological response partition function, the modular relation on the theories with the symmetries becomes
    \begin{equation}
        (ST)^3=Y~.
    \end{equation}
    We will show that a large class of gravitational topological responses can fit into the above projective representation of $\mathrm{SL}(2,\mathbb{Z}_N)$ for gauging $\mathbb{Z}_N$ symmetry. We organize these gravitational topological responses into three classes that often come up in the study of QCAs and gauging on the lattice. The modular relation for gauging abelian invertible symmetries have been discussed in e.g. Refs.~\cite{Witten:2003ya,Gaiotto:2014kfa,Bhardwaj:2020ymp,Hsin:2021qiy,Cheng:2022nji,Choi:2022zal,Apte:2022xtu}.

    \item[2. ] Higher dimensional squares of dualities similar to the Kramers-Wannier/Kitaev wire entangler map in 1+1d (see Fig~\ref{fig:generalN} for an illustration). These squares of dualities map gauging the symmetry to nontrivial QCA via analogous of fermionization or Jordan-Wigner transformation, which we call ``higher Wu gauging". The higher Wu gauging reduces to the Jordan-Wigner transformation in 1+1d. Examples of such square of dualities are discussed in Ref.~\cite{Ji:2019ugf} for 1+1d and Ref.~\cite{Hsin:2021qiy} for higher dimensions.
    
    \item[3. ] Implementation of QCA on the lattice via gauging symmetry from combinations of $S$ and $T$ operations. We will adopt the same notation $T$ to denote the operation of stacking invertible phases with symmetry, even though the $S,T$ operations might not always realize $\mathrm{SL}(2,\mathbb{Z}_N)$ (though they do for certain classes of $Y$ i.e. certain classes of QCA). This is the case when $T$ is stacking an invertible phase whose topological response is not given by quadratic forms. Gauging on the lattice provides an efficient protocol for implementing QCAs in finite time via unitary evolution, measurement, and error correction. Examples of measurement protocol for $S$ operations on the lattice are discussed in Ref.~\cite{Tantivasadakarn:2021vel}.    
    As a nontrivial consistency check, we will show explicitly on the lattice via the polynomial formalism that the QCA produced from the $(ST)^3=Y$ projective modular relation indeed reproduces a known QCA dressed by a FDQC (which does not change the class of QCA).

    \item[4. ] Generalized time-reversal symmetries, that distinguish states that differ by gravitational topological responses. Some of these time-reversal symmetries form higher groups \cite{Cordova:2018cvg,Benini:2018reh}, and others are non-invertible \cite{Choi:2022rfe}.
For an introduction to non-invertible symmetries, see e.g.~Refs.~\cite{Shao:2023gho,Schafer-Nameki:2023jdn}.

\end{itemize}

While we focus on the theoretical understanding of QCAs in quantum lattice models in general spacetime dimensions, they are also of practical relevance. In quantum simulation based on circuits and gates, whether a unitary operator is a nontrivial QCA affects whether the operator can be efficiently implemented, and whether the state created by the unitary operator can be efficiently prepared by unitary evolution.
Moreover, new architectures such as \cite{Bravyi:2023qpn,Loschnauer:2024gat,Pita-Vidal:2024ygt,Ransford:2025ksn,8x68-cx1w} allow quantum platforms including trapped ions and superconductivity chips to realize high dimension quantum systems, including 4+1d loop toric code \cite{Berthusen:2024tsd}. Thus quantum systems and QCAs in high spatial dimensions increasingly become of practical relevance.

The work is organized as follows. In section \ref{sec:modular}, we discuss how QCAs arise in projective modular relations of topological operations for abelian (higher-form) symmetries. These include QCAs related to (1) gravitational Pontryagin terms, (2) Arf-Kervaire-Brown invariants, and (3) topological responses constructed from Stiefel-Whitney classes.
In section \ref{sec:QCAfermionization}, we describe squares of dualities generalizing the well-known 1+1d Jordan-Wigner transformation which maps gauging to the Kitaev wire entangler (a nontrivial fermionic QCA). 
In section \ref{sec:lattice}, we show how general QCA corresponding to certain kinds of gravitational topological responses can be implemented on the lattice via $S$ and generalized $T$ operations, which translate concretely to combinations of unitary evolution, measurement, and error correction. Special classes of such QCA have a simplified implementation procedure related to the projective $\mathrm{SL}(2,\mathbb{Z}_N)$ representations.  
In section \ref{sec:non-invT}, we discuss the generalized time-reversal symmetries in invertible phases entangled by QCAs. Finally, in section \ref{sec:outlook}, we discuss several future directions. We include Appendix \ref{appendix:review} to summarize relevant mathematical properties for cup/higher cup products, Pontryagin square and quadratic functions. 
In Appendix \ref{sec:STST3}, we provide the details of mapping the QCA from $(ST)^3=Y$ for $\mathbb{Z}_3$ 1-form symmetry to a known QCA.

\section{QCAs from Topological Operations of Symmetries: Modular Group 
}
\label{sec:modular}

It is well-known that topological operations with abelian higher form symmetries can form a projective representation of the modular group $\mathrm{SL}(2,\mathbb{Z}_N)$. Identifying $S$ with gauging the symmetry and $T$ with stacking with an appropriate SPT, one obtains 
\begin{equation}
    S^2=C\qquad C^2=1\qquad (ST)^3=Y~,
\end{equation}
where $C$ is charge conjugation (sending background fields $B_i\to -B_i$) and $Y$ corresponds to stacking with an invertible phase that depends only on the spacetime manifold. 

Several cases of such relation are discussed in the literature:
\begin{itemize}
    \item In Refs.~\cite{Witten:2003ya,Cheng:2022nji}, the symmetry is $U(1)$ symmetry in 2+1d, and $T$ is stacking with the basic integer quantum Hall i.e. adding level-one Chern-Simons term for background $U(1)$ gauge field,
    \begin{equation}
TZ[A]=\exp\left(        \frac{i}{4\pi}AdA\right)Z[A]~,
    \end{equation}
    where $A$ is the background $U(1)$ gauge field, and $Z[A]$ is the partition function.
    The invertible phase $Y$ is the gravitational Chern-Simons term with coefficient 2, i.e. two copies of p+ip topological superconductors, which corresponds to chiral central charge 1.
    The modular group can be related to the modular transformation in 3+1d bulk $U(1)$ gauge theory \cite{Seiberg:2016gmd}.

    \item In Refs.~\cite{Gaiotto:2014kfa,Choi:2022zal,Apte:2022xtu}, the symmetry is $\mathbb{Z}_N$ 1-form symmetry, $T$ is stacking with 1-form SPT with effective action
    \begin{equation}
        TZ[B]= \left\{ 
        \begin{array}{cl}
                     \exp\left(\frac{\pi i}{N}\int {\cal P}(B)\right)Z[B]
                     &  \text{even }N \\
             \exp\left(\frac{(N+1)\pi i}{N}\int B\cup B\right)Z[B]
             &  \text{odd }N \\
        \end{array}
\right. ~,
    \end{equation}
    where $B$ is 2-form background $\mathbb{Z}_N$ gauge field, and ${\cal P}$ is the Pontryagin square operation ${\cal P}(B)=B\cup B-B\cup_1 dB$.
The invertible phase $Y$ has partition function given by the signature of the spacetime manifold, which is a gravitational instanton term given by the signature of the manifold. In the discussion we will often omit the cup products for simplicity of notation.

\item In Refs.~\cite{Bhardwaj:2020ymp,Hsin:2021qiy}, the symmetry is internal $\mathbb{Z}_2$ $n$-form symmetry in $(2n+2)$ spacetime dimension, in addition to another $\mathbb{Z}_2$ $n$-form symmetry that mixes with the spacetime symmetry to become a higher group symmetry, generalizing the fermion parity in 1+1d. The operation $T$ corresponds to stacking SPT with ordinary $\mathbb{Z}_2$ $n$-form and the higher group symmetry,
\begin{equation}
TZ[B,\rho]= \exp\left(\frac{2\pi i}{4}\int q_\rho(B)\right)Z[B]~,    
\end{equation}
where $B$ is the $\mathbb{Z}_2$ $(n+1)$-form background field for the internal $\mathbb{Z}_2$ $n$-form symmetry, and $\rho$ is the Wu$_{n+1}$ structure which is a $\mathbb{Z}_2$ $(n+1)$-form background for the higher group symmetry. The Wu structure $\rho$ are higher analogs of pin${}^-$ structures for the case $n=0$. The invertible phase $Y$ are invertible phases that depends on the Wu structure, with partition function given by the Arf-Brown-Kervaire invariant, which includes the Kitaev chain as the special case $n=0$ \cite{Bhardwaj:2020ymp,Hsin:2021qiy}.
Another special case is $n=1$, where the invariants reduces to the previous case with $N=2$, i.e. $\mathbb{Z}_2$ 1-form symmetry in 3+1d, where the Wu structure can be set to zero by a choice of orientation.

\end{itemize}

One of the main results in this work is to show that a third class of gravitational topological response $Y$ can be obtained in a projective $\mathrm{SL}(2,\mathbb{Z}_N)$ algebra. We obtain for any $Y$ that stacks a gravitational topological response built from  Stiefel-Whitney classes of the spacetime tangent bundle, a choice of $S,T$ that gives $(ST)^3=Y$. 
 We can summarize the three types of $Y$ as follows:
\begin{itemize}
    \item Class (1):  Pontryagin gravitational topological response $Y=e^{i\theta\int f(\{p_i\})}$ where the exponent is a polynomial of the Pontryagin classes. The coefficient $\theta$ in general can be real value, but the topological operations for discrete symmetries will give quantized values of $\theta$. Such invertible phases require additional time-reversal symmetry to enforce the quantized response coefficient (see section \ref{sec:non-invT} for more discussion).
     In our cases, we will focus on the gravitational term given by the signature of the spacetime manifold in space dimension divisible by 4.
    This class includes topological operations with $\mathbb{Z}_N$ higher-form symmetries for $N\geq 2$.

    \item Class (2): Arf-Brown-Kervaire invariants: $Y_{\rho}=\sum_be^{i\int q(b)}$ where $b$ is a $\mathbb{Z}_2$ higher-form gauge field and $q(b)$ is a quadratic refinement that may depend on other tangential structures (spin structure, Wu structure, etc).
The spacetime dimension is even. Here and throughout the paper we will omit the overall normalization when we sum over gauge fields.
    In this case, $S$ always corresponds to gauging a $\mathbb{Z}_2$ higher form symmetry. Some examples in this class are also lie in class (1).
    \item Class (3): Stiefel-Whitney gravitational topological responses: $Y=(-1)^{\int w_i\cup \cdots\cup w_k}$ where $w_i$ is the $i$th Stiefel-Whitney class. As in class (2), $S$ always corresponds to gauging a $\mathbb{Z}_2$ higher form symmetry.
\end{itemize}
We will show in Theorem~\ref{thm:generalST3} that, more fundamentally, we can fit any $Y$ that can be written as a Gauss sum into a $\mathrm{SL}(2,\mathbb{Z}_N)$ relation, and we will show how to write the above three classes of gravitational topological responses as Gauss sums.

Although $Y$ is often neglected in the field theory context, where it is just a gravitational topological response, the operation $Y$ cannot be ignored on the lattice. Indeed, we will see that $Y$ can indicate that the operations $S,T$ mix with nontrivial QCA. Using the conjectures in Ref.~\cite{Fidkowski:2024hpz}, we indicate in which cases the properties of $Y$ indicate that $S$ and $T$ mix with a nontrivial QCA. In particular, we will use the conjecture
\begin{conjecture}\label{conj:nontrivialqca}
    From Ref.~\cite{Fidkowski:2024hpz}: $Y$ is a nontrivial QCA if the invertible phase it entangles has a nontrivial partition function on some orientable manifold.
\end{conjecture}
This occurs in all of the class (1) and class (2) examples, as well as some of the class (3) examples such as when $Y$ stacks with $(-1)^{\int w_2w_3}$ which is an absolutely stable (not relying on any symmetries) beyond-cohomology invertible phase.

To give intuition for these $\mathrm{SL}(2,\mathbb{Z}_N)$ relations, we provide examples in each of the three classes in Table~\ref{table:examples}. Note that some of the $S$ operations, i.e. for $(-1)^{\int \nu_2^2}$, perform off-diagonal gauging. 

\begin{table}[t]\label{table:examples}
\centering
\small
\setlength{\tabcolsep}{2.5pt}
\renewcommand{\arraystretch}{1.35}

\begin{tabular}{
  |C{0.18\textwidth}
  |C{0.20\textwidth}
  |C{0.25\textwidth}
  |C{0.08\textwidth}
  |>{\centering\arraybackslash}m{0.22\textwidth}|
}
\hline
$Y$ & $S$ & $T$ & Family &\centering Comments\arraybackslash \\
\hline\hline

$\displaystyle
 \begin{gathered}
 \sum_b
 e^{\frac{\pi i}{N}\int\mathcal P(b)} \\[-5pt]=e^{\frac{\pi i}{12}\int p_1}
 \\[-5pt]
 (N\text{ even})
 \end{gathered}
$
&
$\displaystyle
 \sum_b Z[b]\,
 e^{\frac{2\pi i}{N}\int b\cup B}
$
&
$\displaystyle
 Z[B]\,
 e^{\frac{\pi i}{N}\int\mathcal P(B)}
$
& (1) &
Gravitational Pontryagin response
\\
\hline

$\displaystyle
 \begin{gathered}
 \sum_b
 e^{\frac{\pi i(1+N)}{N}\int\mathcal P(b)}\\[-5pt]=
 e^{-\frac{\pi i}{12}(N-1)\int p_1}
 \\[-5pt]
 (N\text{ odd})
 \end{gathered}
$
&
$\displaystyle
 \sum_b Z[b]\,
 e^{\frac{2\pi i}{N}\int b B}
$
&
$\displaystyle
 Z[B]\,
 e^{\frac{\pi i(1+N)}{N}\int\mathcal P(B)}
$
& (1) &
Gravitational Pontryagin response
\\
\hline

$\displaystyle
 \sum_b
 e^{\frac{\pi i}{2}\int q_{\rho_2}(b)}
$
&
$\displaystyle
 \sum_b Z[b](-1)^{\int bB}
$
&
$\displaystyle
 Z[B]\,
 e^{\frac{\pi i}{2}\int q_{\rho_2}(B)}
$
& (2) &
Chiral semion WW (CSWW)
\\
\hline

$\displaystyle
 \begin{gathered}
 (-1)^{\int\nu_2^2}
 \\[-1pt]
 =(-1)^{\int w_2^2+w_1^4}
 \end{gathered}
$
&
$\displaystyle
 \begin{gathered}
 \sum_{b_1,b_2} Z[b_1,b_2]
 \\[-1pt]
 {}\times
 (-1)^{\int(b_1B_2+b_2 B_1)}
 \end{gathered}
$
&
$\displaystyle
 \begin{gathered}
 Z[B_1,B_2]
 \\[-1pt]
 {}\times
 (-1)^{
   \int\left[
     \nu_2(B_1+B_2)+B_1 B_2
   \right]
 }
 \end{gathered}
$
& (3) &
3 fermion WW (3FWW), $\sim$ $(\text{CSWW})^4$
\\
\hline

$\displaystyle
 (-1)^{\int w_2w_3}
$
&
$\displaystyle
 \begin{gathered}
 \sum_{b,c}Z[b,c]
 \\[-1pt]
 {}\times(-1)^{\int(bC+cB)}
 \end{gathered}
$
&
$\displaystyle
 \begin{gathered}
 Z[B,C]
 \\[-1pt]
 {}\times(-1)^{\int(B C+Bw_3+C w_2)}
 \end{gathered}
$
& (3) &
Absolutely stable invertible phase
\\
\hline

$\displaystyle
 (-1)^{\int w_2^3}
$
&
$\displaystyle
 \begin{gathered}
 \sum_{b,d}Z[b,d]
 \\[-1pt]
 {}\times(-1)^{\int(bD+Bd)}
 \end{gathered}
$
&
$\displaystyle
 \begin{gathered}
 Z[B,D]
 \\[-1pt]
 {}\times
 (-1)^{\int(B D+Bw_2^2+Dw_2)}
 \end{gathered}
$
& (3) &
$T$-protected phase; conjectured trivial QCA
\\
\hline

\end{tabular}
\caption{Key examples of modular relations with $Y$ belonging in three different families: (1) Pontryagin/gravitational instanton counterterms, (2) Arf-Brown-Kervaire invariants, (3) Stiefel-Whitney counterterms. All $b,B$ gauge fields are $2$-form gauge fields, $c,C$ are $3$-form gauge fields, and $d,D$ are $4$-form gauge fields. $p_i$ are the Pontryagin classes, $\nu_i$ are the Wu classes, and $w_i$ are the Stiefel-Whitney classes. We have omitted the cup products to simplify the notation.}
\end{table}

\subsection{General proof of projective $\mathrm{SL}(2,\mathbb{Z}_N)$ relations}\label{sec:general}
Consider an invertible theory that can be written as a Gauss sum:
\begin{equation}\label{gausssum}
    Y[\{c_i\},\{\rho_i\}]=\sum_{\{b_i\}}e^{i\int q(\{b_i\},\{c_i\},\{\rho_i\})}~,
\end{equation}
where $\{c_i\}$ are characteristic classes (including both Pontryagin classes $p_i$ and Stiefel-Whitney classes $w_i$) and $\{\rho_i\}$ are Wu structures. Here, $q(\{b_i\},\{w_i\},\{\rho_i\})$ is a quadratic form: $q(\{b_i+b_i'\},\{c_i\},\{\rho_i\})=q(\{b_i\},\{c_i\},\{\rho_i\})+q(\{b_i'\},\{c_i\},\{\rho_i\})+\langle \{b_i\},\{b_i'\}\rangle$ where $\langle \{b_i\},\{b_i'\}\rangle$ is a non-degenerate bilinear pairing. Here, we omit an overall normalization in the sum over gauge fields.

For example, for $q(B)=\frac{2\pi}{2N}\mathcal{P}(B)$ where $N$ is even and $\mathcal{P}$ is the Pontryagin square, $\langle B,B'\rangle = \frac{2\pi}{N}B\cup B'$. Importantly, $Y$ depends only on the spacetime manifold, and not on background gauge fields for internal symmetries. The form (\ref{gausssum}) is actually very general-- a large class of gravitational topological responses can be represented in such way.
There may also be multiple ways to write any given $Y$ as a Gauss sum; the different ways would give different $S,T$ for which $(ST)^3=Y$. 

For any Gauss sum invertible field theory, the recipe for obtaining the corresponding $S$ and $T$ for a projective $\mathrm{SL}(2,\mathbb{Z}_N)$ representation is as follows:
\begin{itemize}
    \item Suppose that $Y$ can be written as $Y[\{c_i\},\{\rho_i\}]=\sum_{\{b_i\}}e^{i\int q(\{b_i\},\{c_i\},\{\rho_i\})}$. 
    \item We define $T Z[\{B_i\}]=Z[\{B_i\}]e^{i\int q(\{B_i\},\{c_i\},\{\rho_i\})}$. Since we are just multiplying the partition function by a phase, this corresponds to stacking with an invertible phase that can be entangled by a finite-depth circuit.
    \item We define $S$ as $S Z[\{B_i\}]=\sum_{\{b_i\}}Z[\{b_i\}]e^{i\int\langle \{b_i\},\{B_i\}\rangle}$ where the non-degenerate bilinear pairing $\langle \cdot,\cdot\rangle$ comes from the definition of $T$; it is the polarization of the quadratic form $q(\{b_i\},\{c_i\},\{\rho_i\})$.
\end{itemize}

As a small generalization of the above, we can also include spectator background gauge fields $\{B_i\}$ to get $(ST)^3=Y[\{B_i\},\{w_i\},\{\rho_i\}]$ if these spectator background gauge fields do not participate in $S$. This means that their corresponding symmetries are not gauged. For example, this can be used to obtain $Y$ corresponding to stacking with beyond-cohomology SPTs such as the 4+1d $(-1)^{\int Aw_1^4}$ $\mathbb{Z}_2$ SPT (though the QCA that entangles this phase is trivial \cite{fidkowski5d}).

\begin{theorem}\label{thm:generalST3}
    Suppose that $YZ[\{B_i\}]=Z[\{B_i\}]\sum_{\{b_i\}}e^{i\int q(\{b_i\})}$ for some quadratic form $q(\{b_i\})$ with polarization $\langle \{b_i\},\{B_i\}\rangle$. Here we are using $q(\{b_i\})$ as short-hand for $q(\{b_i\},\{c_i\},\{\rho_i\})$. Then  
\begin{equation}
S^2Z[\{B_i\}]=CZ[\{B_i\}]=Z[\{-B_i\}]\qquad (ST)^3Z[\{B_i\}]=YZ[\{B_i\}]
\end{equation}
for $TZ[\{B_i\}]=Z[\{B_i\}]e^{i\int q(\{b_i\})}$ corresponding to stacking with an SPT and $SZ[\{B_i\}]=\sum_{\{b_i\}}Z[\{b_i\}]e^{i\int \langle \{b_i\},\{B_i\}\rangle}$ corresponding to gauging the abelian symmetries associated with the background gauge fields $\{B_i\}$. In other words, $S,T$ defined in this way form a projective representation of $\mathrm{SL}(2,\mathbb{Z}_N)$ with the modified relation $(ST)^3=Y$.
\end{theorem}

\begin{proof}
By definition,
\begin{equation}
    q(\{b_i+b_i'\})=q(\{b_i\})+q(\{b_i'\})+\langle \{b_i\},\{b_i'\}\rangle~.
\end{equation}
Doing this twice gives
\begin{align}
\begin{split}
    &q(\{b_i+b_i'+b_i''\})\\
    &=q(\{b_i\})+q(\{b_i'\})+q(\{b_i''\})+\langle \{b_i\},\{b_i'\}\rangle+\langle\{b_i+b_i'\},\{b_i''\}\rangle\\
    &=q(\{b_i\})+q(\{b_i'\})+q(\{b_i''\})+\langle \{b_i\},\{b_i'\}\rangle+\langle\{b_i\},\{b_i''\}\rangle+\langle\{b_i'\},\{b_i''\}\rangle~,
\end{split}
\end{align}
where in the last line we used the fact that $\langle \cdot,\cdot\rangle$ is a bilinear pairing. With this result in hand, we can easily perform the explicit calculation:
\begin{align}
\begin{split}
    &(ST)^3Z[\{B_i\}]\\
    &=\sum_{\{b_i\},\{b_i'\},\{b_i''\}}Z[\{b_i\}]e^{i\int q(\{b_i\})}e^{i\int \langle \{b_i\},\{b_i'\}\rangle}e^{i\int q(\{b_i'\})}e^{i\int \langle\{b_i'\},\{b_i''\}\rangle}e^{i\int q(\{b_i''\})}e^{i\int \langle\{b_i''\},\{B_i\}\rangle}\\
    &=\sum_{\{b_i\},\{\tilde{b}_i\},\{b_i''\}}Z[\{b_i\}]e^{i\int q(\{\tilde{b}_i\})}e^{-i\int \langle \{b_i\},\{b_i''\}\rangle}e^{i\int \langle\{b_i''\},\{B_i\}\rangle}~,
    \end{split}
    \end{align}
    where $\tilde{b}_i=b_i+b_i'+b_i''$. Now performing the sum over $b_i''$ sets $b_i=B_i$, so we get 

\begin{equation}
        (ST)^3Z[\{B_i\}]=\sum_{\{\tilde{b}_i\}}Z[\{B_i\}]e^{i\int q(\{\tilde{b}_i\})}=YZ[\{B_i\}]
\end{equation}
as desired. To get the $S^2Z[\{B_i\}]=CZ[\{B_i\}]=Z[\{-B_i\}]$ relation, we calculate
\begin{equation}
    S^2Z[\{B_i\}]=\sum_{\{b_i\},\{b_i'\}}Z[\{b_i\}]e^{i\int\langle \{b_i\},\{b_i'\}\rangle}e^{i\int\langle \{b_i'\},\{B_i\}\rangle}=Z[\{-B_i\}]~,
\end{equation}
where the last equality follows from observing that the sum over $\{b_i'\}$ sets $b_i=-B_i$.
\end{proof}

Assuming the same form of $Y,S,T$ we also have
\begin{equation}\label{STYT}
    ST[1]=YT^{-1}[1]
\end{equation}
So to obtain the pure gravitational term we can start from the trivial state and apply $STS[1]$. To obtain (\ref{STYT}), we can calculate explicitly
\begin{align}
\begin{split}
    STZ[\{B_i\}]&=\sum_{\{b_i\}}Z[\{b_i\}]e^{i\int q(\{b_i\})}(-1)^{\int \langle\{b_i\},\{B_i\}\rangle}\\
    &=\sum_{\{b_i\}}Z[\{b_i\}]e^{i\int q(\{b_i+B_i\})}e^{-i\int q(\{B_i\})}~.
\end{split}
\end{align}
For the particular case $Z[\{B_i\}]=1$, we can shift the sum over $\{b_i\}$ to a sum over $\{b_i+B_i\}$ to get (\ref{STYT}). This cannot be done for general input $Z[\{B_i\}]$. Note that when $T$ is order two, which is always the case in for the class (3) examples, we have $ST[1]=YT[1]$. This relation will be important for when we discuss QCA in Sec.~\ref{sec:lattice}. 

\subsection{Three classes of $\mathrm{SL}(2,\mathbb{Z}_N)$ relations}
In the following, we will show that the three classes of gravitational topological responses described above can be written as Gauss sums, and obtain the corresponding $S$ and $T$ operations. Using Conjecture~\ref{conj:nontrivialqca}, we will indicate which $\mathrm{SL}(2,\mathbb{Z}_N)$ relations involve $S$ and $T$ mixing with nontrivial QCA via $Y$.

\subsubsection{Class (1): $\mathbb{Z}_N$ symmetry and gravitational Pontryagin topological responses}

Here we first review the class of $\mathrm{SL}(2,\mathbb{Z}_N)$ relations with $Y$ that performs stacking with gravitational Pontryagin counterterms. In 3+1d, this type of $Y$ can be written as a Gauss sum as (e.g.~Refs.~\cite{Hsin:2018vcg,Choi:2021kmx,Apte:2022xtu})
\begin{equation}\label{3dZNY}
    Y=\begin{cases} \sum_be^{\frac{2\pi i}{2N}\int \mathcal{P}(b)}=e^{\frac{\pi i}{4}\sigma(X)}=e^{\frac{\pi i}{12}\int p_1}\qquad N\text{ even}\\
    \sum_be^{\frac{2\pi i(N+1)}{2N}\int \mathcal{P}(b)}=e^{-\frac{\pi i}{4}(N-1)\sigma(X)}=e^{-\frac{\pi i}{12}(N-1)\int p_1}\qquad N\text{ odd}\end{cases}~,
\end{equation}
and thus $q(b)=\frac{2\pi}{2N}\mathcal{P}(b)$ for $N$ even and $\frac{2\pi(N+1)}{2N}\mathcal{P}(b)$ for $N$ odd. Here, the quadratic form $\mathcal{P}(b)$ satisfies
\begin{equation}
\mathcal{P}(b+b')=\mathcal{P}(b)+\mathcal{P}(b')+2b\cup b'~,
\end{equation}
so the polarization gives the bilinear pairing $\langle b,b'\rangle=\frac{2\pi}{N}b\cup b'$.\footnote{Note that we have intentionally chosen the odd $N$ SPT so that 
\begin{equation}
    2\frac{N+1}{2N}=\frac{1}{N}\quad\text{mod 1}~,
\end{equation}
in order to get the correct bilinear pairing $\frac{2\pi}{N}b\cup b'$. If we instead had $q(b)=\frac{2\pi}{N}\mathcal{P}(b)=\frac{2\pi}{N}b\cup b$ for $N$ odd then we would instead have polarization $\frac{4\pi}{N}b\cup b'$ which does not match the standard bilinear pairing for gauging $\mathbb{Z}_N$ symmetry. Since $N$ is odd, $N+1$ is even, so $\frac{N+1}{2N}$ is an element of the $\mathbb{Z}_N$ classification of $\mathbb{Z}_N$ 1-form SPTs for $N$ odd.} We have
\begin{equation}
    SZ[B]=\sum_bZ[b]e^{\frac{2\pi i}{N}\int b\cup B}\qquad TZ[B]=Z[B]e^{\frac{2\pi it_N}{2N}\int\mathcal{P}(b)}~,
\end{equation}
where $t_N=1$ for $N$ even and $t_N=N+1$ for $N$ odd. Throughout the paper we will omit the overall normalization when we sum over gauge fields.

When $Y\neq 1$, it defines a state that requires a nontrivial QCA \cite{Haah:2018jdf,Haah:2019fqd}. This comes from the incompatibility of a 2+1d boundary topological order with nonzero chiral central charge mod 8 with a commuting projector realization \cite{Haah:2018jdf}. This is also compatible with Conjecture~\ref{conj:nontrivialqca}: such $Y$ take nontrivial value on orientable manifolds.

\paragraph{Generalization to $4k$ spacetime dimension}

When the spatial dimension is $(4k+3)$, we can define $S,T$ operations with a choice of orientation for $\mathbb{Z}_N$ internal $(2k-1)$-form symmetry. 
We will focus on even $N$ in the following.

We define $S,T$ operations as follows:
\begin{align}
    &SZ[B_{2k}]=\sum_b Z[b_{2k}] e^{\frac{2\pi i}{N}\int b_{2k}\cup B_{2k}}\cr 
    &T Z[B_{2k}]=Z[B_{2k}]e^{\frac{2\pi i }{2N}\int {\cal P}(B_{2k})}~.
\end{align}
Here, we define the operations with a choice of orientation. Throughout the paper we will omit the overall normalization when we sum over gauge fields.
 The ${\cal P}$ is a generalization of Pontryagin square review in Appendix \ref{appendix:review}:
\begin{equation}
    {\cal P}(b_{2k}):=b_{2k}\cup b_{2k}-b_{2k}\cup_1 db_{2k}~.
\end{equation}
For $\mathbb{Z}_{N}$ $2k$-cocycle $b_{2k}$, this gives $\mathbb{Z}_{2N}$ $4k$-cocycle ${\cal P}(b_{2k})$.
This Pontryagin square operation on higher degree cocycles satisfies the same quadratic property of the usual Pontryagin square operation on 2-cocycles.

There is an invertible phase with partition function
\begin{equation}
    Y:=\sum_b e^{\frac{2\pi i}{2N}\int{\cal P}(b)}~,
\end{equation}
where we omit an overall normalization in the sum over gauge field \cite{Hsin:2021qiy}.
The boundary of the invertible phase is characterized by anyonic $(2k-2)$-dimensional excitations that have $\mathbb{Z}_N$ fusion rule and generalized statistics \cite{Gukov:2020btk,Kobayashi:2024dqj,Feng:2025mdg,Gukov:2025dol} of spin $e^{\pi i/N}$. The excitations correspond to the operator on the boundary that terminates the bulk operator $e^{\frac{2\pi i}{N}\int b_{2k}}$ similar to the discussion in Ref.~\cite{Hsin:2018vcg} for 3+1d.
We have
\begin{theorem}
    The $S,T$ operations satisfy
    \begin{equation}
        \left(ST\right)^3=Y~.
    \end{equation}
\end{theorem}
\begin{proof}
Consider $ST$ on a general partition function $Z[B]$:
\begin{equation}
    ST Z[B]=\sum_b Z[b] e^{\frac{2\pi i}{2N}\int{\cal P}(b)+\frac{2\pi i}{N}\int b\cup B}~.
\end{equation}
Performing the operation three times gives
\begin{align}
&\left(    {\cal S}{\cal T}_\rho\right)^3 Z[B]=
\sum_{b,b',b''} Z[b] e^{\frac{2\pi i}{2N}\int{\cal P}(b)+\frac{2\pi i}{N}\int b\cup b'} e^{\frac{2\pi i}{2N}\int {\cal P}(b')+\frac{2\pi i}{N}\int b'\cup b''} e^{\frac{2\pi i}{2N}\int {\cal P}(b'')+\frac{2\pi i}{N}\int b''\cup B}\cr 
&=\sum_{b,\tilde b',b''}Z[b] e^{\frac{2\pi i}{2N}\int {\cal P}(\tilde b')}e^{\frac{2\pi i}{N}\int (B-b)\cup b''}=Z[B] \sum_{\tilde b'} e^{\frac{2\pi i}{2N}\int{\cal P}(\tilde b')}=Z[B]Y~,
\end{align}
where $\tilde b'=b'+b+b''$.
\end{proof}

\subsubsection{Class (2): Arf-Brown-Kervaire invariants}
Brown-Kervaire invariants are partition functions of the form $\sum_bi^{\int q(b)}$ where $q(b)$ is a quadratic refinement of middle-dimensional intersection pairing $H^{(d+1)/2}(M;\mathbb{Z}_2)\times H^{(d+1)/2}(M;\mathbb{Z}_2)\to\mathbb{Z}_2$. Here we are working in even spacetime dimension $d+1$, so $(d+1)/2$ is an integer. We review properties of quadratic functions and refinements in Appendix~\ref{appendix:review}. One such refinement is the Pontryagin square: for $B$ a $\mathbb{Z}_2$ $(2n)$-form gauge field, $\mathcal{P}(b):H^{2n}(M,\mathbb{Z}_2)\times H^{2n}(M,\mathbb{Z}_2)\to\mathbb{Z}_4$ allows for the definition of the root $\mathbb{Z}_2$ $(2n-1)$-form SPT in (2n-1)+1d as $i^{\int \mathcal{P}(b)}$. For $n=1$, this is a special case of $N=2$ in the above class (1) of gravitational topological responses. We will see below that for $N=2$, we can use more general quadratic refinements to obtain other choices of $Y$.

Quadratic refinements in general can include Wu structures of the tangent bundle (see Appendix~\ref{appendix:review}). The most familiar case of this is in 1+1d, where $b$ is a 1-form gauge field and the theory $\sum_bi^{\int q_{\rho_1}(b)}$ is the partition function of the Kitaev Wire (the Arf invariant). Here, $\rho_1$ is the spin structure: $d\rho_1=w_2$ where $w_2$ is the second Stiefel-Whitney class. In 3+1d, we can define
\begin{equation}
    Y[\rho_2]=\sum_bi^{\int q_{\rho_2}(b)}=\sum_bi^{\int \mathcal{P}(b)+2b\cup \rho_2}\qquad d\rho_2=\nu_3~,
\end{equation}
where $\nu_3$ is the third Wu class.
Throughout the paper we will omit the overall normalization when we sum over gauge fields.
$\rho_2$ can be thought of as a background gauge field for a spacetime 1-form symmetry. The usual Pontryagin square is recovered by setting $\rho_2=0$. The above Brown-Kervaire invariant describes an invertible phase where the 1-form symmetry mixes with time-reversal in a 2-group \cite{Hsin:2021qiy}. Specifically, it is described by the chiral semion Walker-Wang model, whose boundary theory hosts an anti-semion topological order \cite{Hsin:2021qiy}. It is a phase that can only be entangled by a order-eight nontrivial QCA \cite{Shirley:2022lhu}. There is a similar invertible phase and corresponding order-eight nontrivial quantum cellular automaton in $4k$ spacetime dimensions; its fourth power gives analogs of the 3 fermion Walker Wang model in higher $4k$ spacetime dimensions with partition function $(-1)^{\int \nu_{2k}^2}$ \cite{Fidkowski:2024hpz}.

$Y[\rho]$, where $\rho$ is a Wu $(d+1)/2$ structure, participates in a projective $\mathrm{SL}(2,\mathbb{Z}_2)$ relation involving $\mathbb{Z}_2$ $(d+1)/2$-form symmetries (recall that we are working in even spacetime dimension):
\begin{equation}
    (ST[\rho])^3=Y[\rho]~,
\end{equation}
where
\begin{equation}
    SZ[B]= \sum_b Z[b] (-1)^{\int b\cup B}\qquad T[\rho] Z[B]= Z[B] i^{\int q_\rho(B)}~.
\end{equation}
and $\rho,b,B\in H^{(d+1)/2}(M,\mathbb{Z}_2)$. In general, $\rho$ is the background gauge field for a $\mathbb{Z}_2$ $(d+1)/2$-form symmetry that participates in a spacetime higher group symmetry, so $Y$ has the interpretation as a spacetime higher group SPT.

\subsubsection{Class (3): Stiefel-Whitney topological responses
}\label{sec:stiefelwhitneyst}
Stiefel-Whitney gravitational topological responses include all invertible phases of the form $(-1)^{\int w_iw_j\cdots w_k}$ with $i+j+\cdots +k=d+1$ equal to the spacetime dimension (note that we allow for repeats i.e. $i=j$). To write this partition function as a Gauss sum, we can split $w_iw_j\cdots w_k$ into $w_i$ and $w_j\cdots w_k$. Then 
\begin{equation}
    (-1)^{\int w_i\cup w_j\cup\cdots\cup w_k}=\sum_{b,c}(-1)^{\int b\cup c+b\cup w_i+c\cup \left(w_j\cup\cdots\cup w_k\right)}~,
\end{equation}
where $c$ is a $\mathbb{Z}_2$ $i$-form gauge field and $b$ is a $\mathbb{Z}_2$ $(d+1-i)$-form gauge field. 
Throughout the paper we will omit the overall normalization when we sum over gauge fields.
The above equality can be seen by first summing over $b$, which sets $c=w_i$. Even if we have $\int w_{d+1}$, we can in principle split it into $1$ and $w_{d+1}$ if we allow for using 0-form background gauge fields for $-1$ form symmetries (though these actions are deformation-trivial anyways). To obtain the above Gauss sum via $S,T$, we define
\begin{equation}
SZ[B,C]=\sum_{b,c}Z[b,c](-1)^{\int b\cup C+c\cup B}\qquad TZ[B,C]=Z[B,C](-1)^{\int B\cup C+B\cup w_i+C\cup \left(w_j\cup\cdots\cup w_k\right)}~,
\end{equation}
where $B$ is a background $\mathbb{Z}_2$ $(d+1-i)$-form gauge field and $C$ is a background $i$-form gauge field. Therefore, this corresponds to the case where the $S$ simultaneously gauges a $\mathbb{Z}_2$ $i$-form symmetry and a $\mathbb{Z}_2$ $(d+1-i)$-form symmetry. 
Throughout the paper we will omit the overall normalization when we sum over gauge fields.
A sum of monomials $w_iw_j\cdots w_k$ can be handled similarly, using pairs $\{b_i,c_i\}$.

There are some interesting special cases of this construction:
\begin{itemize}
    \item $Y=(-1)^{\int \nu_2^2}$: Here, $\nu_2$ is the second Wu class. On orientable manifolds, this can be written as $Y=(-1)^{\int w_2^2+w_1^4}$. This phase is the 3+1d beyond-cohomology 3-fermion Walker-Wang theory, with boundary chiral central charge 4 mod 8. It is equivalent to four copies of the chiral semion Walker-Wang model, which is an example of $Y$ in class (1)/(2). In this case we choose the quadratic form $q(B,C)=\pi (B\cup C+B\cup \nu_2+C\cup \nu_2)$ for $B,C$ both $2$-form $\mathbb{Z}_2$ gauge fields. Since $(-1)^{\int w_2^2+w_1^4}=(-1)^{\int w_2^2}$ is still nontrivial on some orientable manifolds (i.e. it is $(-1)$ on $\mathbb{CP}^2$), Y corresponds to a nontrivial QCA.
    \item $Y=(-1)^{\int w_2\cup w_3}$: This is an absolutely stable phase, meaning it is an invertible phase without any symmetry requirements. To obtain this $Y$ we choose $q(B,C)=\pi(B\cup C+B\cup w_3+C\cup w_2)$ where $B$ is $2$-form and $C$ is $3$-form. $Y$ is nontrivial on non-orientable manifolds such as the Wu manifold $W^5=SU(3)/SO(3)$, so it is entangled by a nontrivial QCA according to Ref.~\cite{Chen:2021xks,Fidkowski:2024hpz}.
    \item $Y=(-1)^{\int \nu_4^2}$ where $\nu_4$ is the fourth Wu class: This 7+1d phase is a higher dimensional analog of the 3-fermion Walker Wang phase \cite{Fidkowski:2024hpz}. It is equivalent to four copies of the 7+1d analog of the chiral semion Walker-Wang theory $\sum_be^{\frac{i\pi}{2}\int q_{\rho_4}(b)}$ where $\rho_4$ is the Wu structure $d\rho_4=\nu_5$ covered in class (1)/(2). On orientable manifolds, $(-1)^{\int\nu_4^2}=(-1)^{\int w_4^2+w_2^4}$ since on these manifolds, $\nu_2=w_2^2+w_4$. Actually, the theories $(-1)^{\int w_2^4}$ and $(-1)^{\int w_4^2}$ are individually nontrivial on orientable manifolds: $(-1)^{\int w_2^4}$ is nontrivial on $\mathbb{CP}^4$ while $(-1)^{\int w_4^2}$ is nontrivial on $\mathbb{CP}^2\times\mathbb{CP}^2$. Therefore, although Ref.~\cite{Fidkowski:2024hpz} only specifies the combination of these two states, they should actually individually give nontrivial QCA according to Conjecture~\ref{conj:nontrivialqca}. It is straightforward to obtain quadratic forms giving the corresponding $S,T$: for $w_2^4$ we use $q(B,C)=\pi(B\cup C+B\cup w_2^2+C\cup w_2^2)$ and for $w_4^2$ we use $q(B,C)=\pi(B\cup C+B\cup w_4+C\cup w_4)$ where $B,C$ are both $\mathbb{Z}_2$ $4$-form gauge fields. $(-1)^{\int w_2^4}$ can actually be written as a Gauss sum in two different ways: we can also choose the quadratic form $q(B,C)=\pi(B\cup C+B\cup w_2+C\cup w_2^3)$ for 6-form $B$ and 2-form $C$. 
    \item $Y=(-1)^{\int w_2^3}$: This is an example of an invertible phase (protected by time-reversal symmetry) that is trivial on every orientable manifold,\footnote{This follows from $\pi\int w_2^3=\pi\int Sq^2(w_2^2)=\pi\int Sq^1w_2\cup Sq^1w_2=\pi\int w_1\cup w_2\cup (w_3+w_1\cup w_2)$ which is trivial on orientable manifolds.
   In fact, there is no nontrivial 5+1d topological terms from Stiefel-Whitney classes that take nontrivial value on orientable 6-manifolds due to the relations $w_2\cup w_4=0,w_3^2=0,w_6=0$ for orientable 6-manifolds $w_1=0$ (see e.g. Ref.~\cite{Wen:2014zga}).
    } so would not give a nontrivial QCA according to Conjecture~\ref{conj:nontrivialqca}. It is a simple example of an action that is not naively a pairing, but by splitting it into $w_2$ and $w_2^2$ we can get a quadratic form $q(B,C)=\pi(B\cup C+B\cup w_2+C\cup w_2^2)$. The resulting $S$ corresponds to gauging $\mathbb{Z}_2$ $1$-form and $3$-form symmetries.

\end{itemize}

It will be particularly useful when $T$ can be written entirely in terms of background gauge fields $\{B_i\}$ rather than explicitly in terms of Stiefel-Whitney classes. This turns out to be the case for the first three examples above; in Sec.~\ref{sstiefellattice}, we show that $Y$ in these cases can be written as $\sum_{\{b_i\}}(-1)^{\int q(\{b_i\})}$ where $q(\{b_i\})$ is a quadratic form involving symmetry gauge fields only. This allows for efficient implementation of the action of the QCA on product states to get invertible states with these gravitational topological responses. Otherwise, it is unclear how to obtain couplings between gauge fields and Stiefel-Whitney classes via local unitaries on the lattice.

\section{Mapping Gauging to QCA: Squares of Dualities and Generalized Fermionization}
\label{sec:QCAfermionization}

In this section, we will show that $S$ and $T$ together form a square of dualities that map gauging of $\mathbb{Z}_N$ 1-form symmetry to QCA. For the special case of $N=2$, we can use an analog of fermionization/Jordan-Wigner transformation that we call ``twisted Wu gauging" to map the 1-form symmetry to a dual spacetime 1-form symmetry with background field given by the Wu structure. We will motivate this construction by first reviewing the square of dualities in 1+1d. For more general $N$, the dual 1-form symmetries are always just ordinary 1-form symmetries; since characteristic classes are always $\mathbb{Z}_2$ we cannot use them for $N\neq 2$. In this case, instead of twisted Wu gauging, we just have the usual twisted gauging operations such as $ST$.

\subsection{Fermionization/bosonization, gauging, and QCA in 1+1d}
To begin, we first review how topological manipulations of $\mathbb{Z}_2$ 0-form symmetry in 1+1d fit together into a square of dualities. In 1+1d it is convenient to introduce a modified gauging $S_{\rho_1}$ that implements gauging with fermionic ancillas. Here, $\rho_1$ is a spin structure, which we will describe further below, and serves as a background gauge field for fermion parity symmetry. In this case we will recover the familiar fact that fermionization/bosonization maps Kramers-Wannier duality to the Kitaev wire entangler, which is a nontrivial fermionic QCA.

At the level of lattice operators, fermionization/bosonization is given explicitly by 
\begin{equation}\label{jw}
X_j\leftrightarrow i\tilde{\gamma}_j\gamma_{j+1}\qquad Z_jZ_{j+1}\leftrightarrow i\gamma_{j+1}\tilde{\gamma}_{j+1}~,
\end{equation}
where a complex fermion on site $j$ is represented by two real Majoranas $\gamma_j$ and $\tilde{\gamma}_j$. 

Fermionization corresponds to gauging the $\mathbb{Z}_2$ symmetry with fermionic ancillas. Here we will provide a description of fermionization that makes clear the connection to Arf-Brown-Kervaire invariants and generalizations to higher dimensions. First, we need to introduce $\rho_1$, which is a $\mathbb{Z}_2$ $1$-cochain corresponding to spin structures:
\begin{equation}
    d\rho_1=w_2~.
\end{equation}
As discussed in Sec.~\ref{sec:general} (see also Appendix~\ref{appendix:review}), there is a quadratic refinement $q_{\rho_1}(b):H^1(M,\mathbb{Z}_2)\times H^1(M,\mathbb{Z}_2)\to \mathbb{Z}_4$ labeled by $\rho_1$ that gives a polarization that matches the bilinear pairing for gauging of the $\mathbb{Z}_2$ symmetry:
\begin{equation}
    q_{\rho_1}(b+b')=q_{\rho_1}(b)+q_{\rho_1}(b')+2b\cup b'\quad \text{mod 4}~.
\end{equation} 
The factor of 2 in front of $b\cup b'$ is important for matching the $(-1)^{\int b\cup b'}$ bilinear pairing in $S$. 

\begin{figure}[h]
   \centering
   \includegraphics[width=.6\columnwidth]{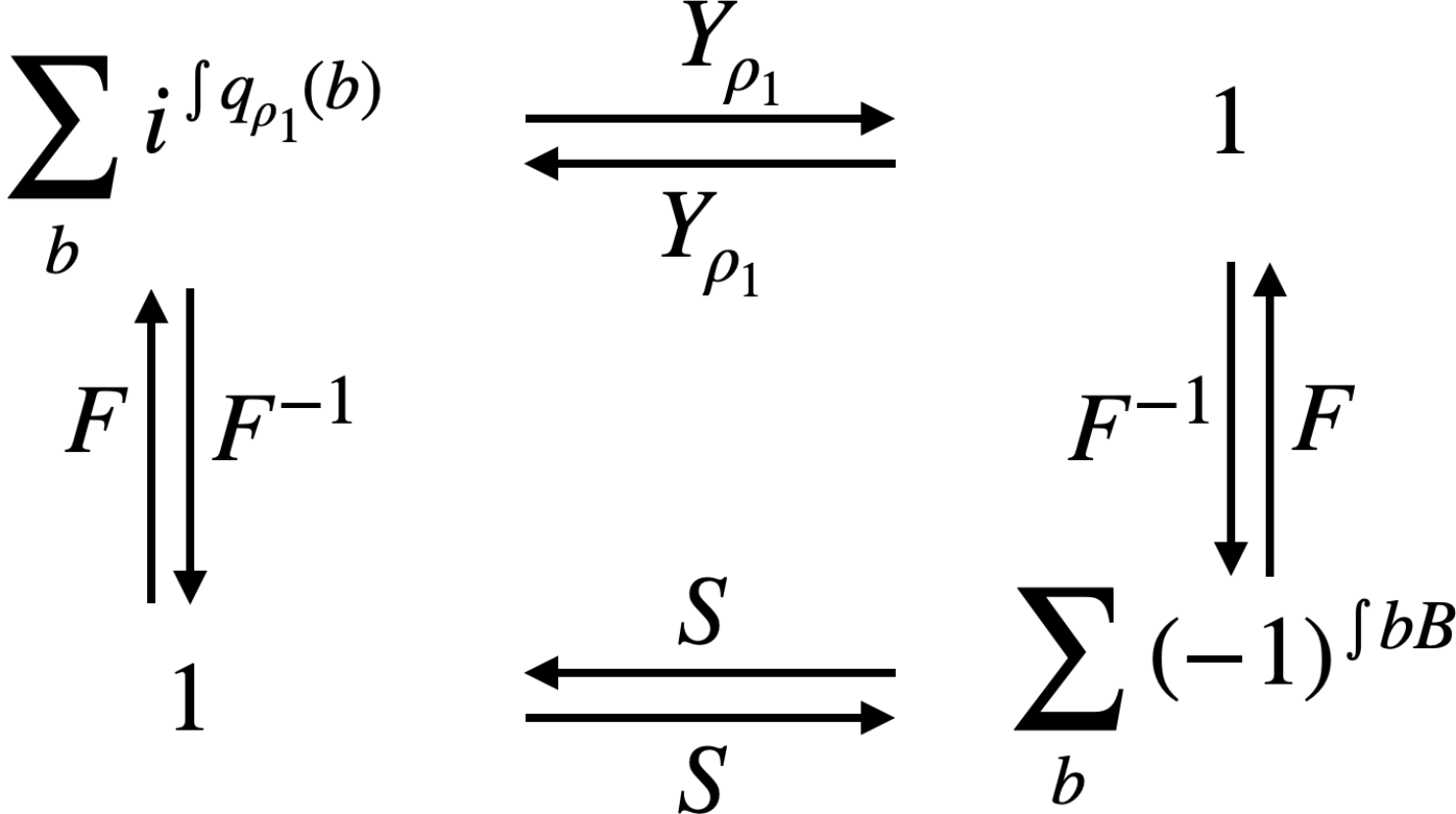} 
   \caption{Square of dualities in 1+1d. The bottom two phases are bosonic while the top two phases are fermionic. Gauging refers to gauging the $\mathbb{Z}_2$ 0-form symmetry (i.e. Kramers-Wannier duality) while stacking refers to stacking with the Kitaev wire $\sum_be^{\frac{i\pi}{2}\int q_{\rho_1}(b)}=(-1)^{\mathrm{Arf}}$.}
   \label{fig:1dsquare}
   \end{figure}
   
With these preliminaries at hand, we can now identify fermionization via a two-step process:
\begin{enumerate}
\item Stack the theory with a $\mathbb{Z}_2\times\mathbb{Z}_2^f$ SPT: $Z[B]\to Z[B]e^{\frac{i\pi}{2}q_{\rho_1}(B)}$ where $q_{\rho_1}(B)$ is a quadratic form labeled by $\rho_1$, reviewed in Appendix~\ref{appendix:review}. This SPT is simply the $\mathbb{Z}_2\times\mathbb{Z}_2^f$ cluster state, where the $\mathbb{Z}_2$ domain walls are decorated by fermion parity odd operators \cite{gufermionic,tantikramerswannier}.

\item Sum over the background field for the $\mathbb{Z}_2$ symmetry, leaving a theory with just $\mathbb{Z}_2^f$ symmetry: $Z[B]e^{\frac{i\pi}{2}\int q_{\rho_1}(B)}\to \sum_bZ[b]e^{\frac{i\pi}{2}\int q_{\rho_1}(b)}$ 
\end{enumerate}

We will denote this transformation by $F_{1}[Z[B]]$. There is a similar transformation $F_{1}^{-1}[Z[\rho_1]]$ for mapping a fermionic system to a bosonic system with $\mathbb{Z}_2$ 0-form symmetry. We again stack with a $\mathbb{Z}_2\times\mathbb{Z}_2^f$ SPT, and then we sum over spin structures. In summary, we have
\begin{equation}\label{fermionization}
    F_1 Z[B]=\sum_bZ[b]e^{\frac{i\pi}{2}\int q_{\rho_1}(b)}\qquad F_1^{-1}Z[\rho_1]= \sum_{\rho_1}Z[\rho_1]e^{\frac{i\pi}{2}\int q_{\rho_1}(b)}~,
\end{equation}
where we omit the overall normalization when we sum over gauge fields.

Note that although $e^{\frac{i\pi}{2}\int q_{\rho_1}(B)}$ looks to be order four, it is actually order two on orientable manifolds (i.e. without time-reversal symmetry). This is because
\begin{equation}
\pi\int q_{\rho}(B)=\pi\int B\cup B=\pi\int B\cup w_1~,
\end{equation}
where the last relation uses the Wu formula, and $\nu_1=w_1$ is the first Wu class. The above theory is trivial because $w_1=0$ on orientable manifolds. This is consistent with the $\mathbb{Z}_2\times\mathbb{Z}_2^f$ cluster state being order two without time reversal symmetry.

The second step of promoting $B$ to be a dynamical field $b$ and sum over it removes the bosonic degrees of freedom, resulting in theory with just $\mathbb{Z}_2^f$ symmetry. Although the partition function includes a sum, it actually describes an invertible theory. If we stack the theory with its complex conjugation, with gauge field $b'$, we get the total action
\begin{equation}\label{kitorder2}
\frac{i\pi}{2}\int q_{\rho_1}(b)-\frac{\pi}{2}\int q_{\rho_1}(b')=\frac{\pi}{2}\int q_{\rho_1}(b'')+\pi\int b''\cup b'~,
\end{equation}
where $b''=b+b'$. $b'$ acts as a Lagrange multiplier: summing over $b'$ fixes $b''=0$, resulting in a trivial theory.

\subsubsection{Square of dualities in 1+1d}
Suppose that we start in the trivial phase of the bosonic model, so $Z[B]=1$. Then $\mathcal{F}_1[Z[B]]$ is simply the partition function of the Kitaev wire:
\begin{equation}\label{kitaevwire}
\sum_{b\in H^1(M,\mathbb{Z}_2)}e^{\frac{i\pi}{2}\int q_{\rho_1}(b)}=Z[\rho_1]=(-1)^{\mathrm{Arf}(\rho_1)}~,
\end{equation}
where we omit the overall normalization when we sum over gauge fields.

This reproduces the Jordan-Wigner map between the trivial bosonic paramagnet phase and the Kitaev wire. Now suppose that we gauge the $\mathbb{Z}_2$ symmetry in the bosonic theory, giving the partition function $\propto\sum_{b}e^{i\pi\int b\cup B}=\delta(B)$ where $B$ is a background one-form gauge field. Now we can again fermionize to get
\begin{equation}\label{fermins}
F_1[\delta(B)]=\sum_{b,b'}e^{i\pi\int b\cup b'+\frac{1}{2} q_{\rho_1}(b')}~.
\end{equation}
Summing over $b$ sets $b'=0$ so the above partition function is trivial. It is the trivial fermionic insulator: $\mathcal{F}[\delta(B)]=1$.

How do we get from (\ref{kitaevwire}) to (\ref{fermins}) on the fermionic side? We do so by stacking with a Kitaev wire. Specifically, we have
\begin{align}
\begin{split}\label{derivestack}
\sum_{b,b'}e^{i\pi\int b\cup b'+\frac{1}{2} q_{\rho_1}(b')}&=\sum_{b,b'} e^{\frac{i\pi}{2}\int q_{\rho_1}(b'+b)-q_{\rho_1}(b)}\\
&=\left(\sum_{b''}e^{\frac{i\pi}{2}\int q_{\rho_1}(b'')}\right)\left(\sum_b e^{-\frac{i\pi}{2}\int q_{\rho_1}(b)}\right)~,
\end{split}
\end{align}
where we define $b''=b'+b$. Because the Kitaev wire is order two on orientable manifolds, the action (\ref{fermins}) is obtained from (\ref{kitaevwire}) by stacking with $\sum_{b''} e^{\frac{i\pi}{2}\int q_{\rho_1}(b'')}$. We omit the overall normalization when we sum over gauge fields.

Let's call stacking with the Kitaev wire $Y_{\rho_1}$. Then to summarize the discussion above, we have in 1+1d a familiar square of dualities with fermionization/bosonization on the vertical legs, gauging on the bottom leg, and the Kitaev wire QCA on the top leg (see Fig.~\ref{fig:1dsquare}). This can be summarized as
\begin{equation}\label{trivialmap}
    Y_{\rho_1}F_{1}[1]=F_1[S[1]]~,
\end{equation}
where the input $Z[B]=1$ is the trivial partition function. 

This square of dualities in fact holds for any input $Z[B]$:
\begin{align}
\begin{split}\label{fermionizationsquare}
    F_{1}[S[Z[B]]]&=\sum_{b,b'}Z[b]e^{i\pi\int b\cup b'}e^{\frac{i\pi}{2}\int q_{\rho_1}(b')}\\
    &=\sum_{b,b'}Z[b]e^{-\frac{i\pi}{2}\int q_{\rho_1}(b)}e^{\frac{i\pi}{2}\int q_{\rho_1}(b'+b)}\\ &=\left(\sum_{b''}e^{\frac{i\pi}{2}\int q_{\rho_1}(b'')}\right)\sum_{b}Z[b]e^{-\frac{i\pi}{2}\int q_{\rho_1}(b)}=Y_{\rho_1}F_{1}[Z[B]]~,
\end{split}
\end{align}
where the last line comes from the fact that $Y_{\rho_1}$ is order two (i.e. $\sum_{b''}e^{\frac{i\pi}{2}\int q_{\rho_1}(b'')}=\sum_{b''}e^{-\frac{i\pi}{2}\int q_{\rho_1}(b'')}$) and $e^{-\frac{i\pi}{2}\int q_{\rho_1}(b)}=e^{\frac{i\pi}{2}\int q_{\rho_1}(b)}$ since it is $\pm 1$. Physically, $F_{1}[Z[b]]=\tilde{S}_{\rho_1}T[Z[b]]$ where $\tilde{S}$ means gauging with fermionic ancillas with background gauge field $\rho_1$. In other words, instead of using an ordinary $\mathbb{Z}_2$ 1-form background gauge field, we use the spin structure $\rho_1$. Rearranging the operations in (\ref{fermionizationsquare}), we get
\begin{equation}\label{YtoS}
    Y_{\rho_1}=F_{1}S F_{1}^{-1}~,
\end{equation}
where $F_{1}^{-1}$ defined in (\ref{fermionization}) implements bosonization. 

While we derived (\ref{fermionizationsquare}) via explicit calculation, we can also obtain it from the $\mathrm{SL}(2,\mathbb{Z}_2)$ relation
\begin{equation}
    (ST_{\rho_1})^3=Y_{\rho_1}~.
\end{equation}

To see this, define the operation $\epsilon_0$ as evaluation at trivial background: $\epsilon_0 Z[B]=Z[0]$. Then since $ST_{\rho_1} Z[B]=\sum_bZ[b]e^{\frac{i\pi}{2}\int q_{\rho_1}(b)}(-1)^{\int b\cup B}$,
\begin{equation}
    F_1=\epsilon_0ST_{\rho_1}~.
\end{equation}
Rearranging $(ST_{\rho_1})^3=Y_{\rho_1}$ gives 
\begin{equation}\label{STS}   ST_{\rho_1}S=Y_{\rho_1}T_{\rho_1}^{-1}S^{-1}T_{\rho_1}^{-1}~.
\end{equation}
For $\mathbb{Z}_2$ gauge fields, $S^{-1}=CS=S$ since $Z[B]=Z[-B]$. We also have $T_{\rho_1}=T_{\rho_1}^{-1}$ on orientable manifolds. Applying $\epsilon_0$ to both sides of (\ref{STS}) gives
\begin{equation}
\epsilon_0ST_{\rho_1}S=Y_{\rho_1}\epsilon_0T_{\rho_1}ST_{\rho_1}=ST_{\rho_1}~,
\end{equation}
where we used $\epsilon_0T_{\rho_1}Z[B]=Z[0]e^{\frac{i\pi}{2}\int q_{\rho_1}(0)}=\epsilon_0Z[B]$. It follows that
\begin{equation}
    F_1S=Y_{\rho_1}F_1~,
\end{equation}
recovering (\ref{fermionization}) and (\ref{YtoS}). Therefore, the $\mathrm{SL}(2,\mathbb{Z}_2)$ relation with the Arf invariant $Y_{\rho_1}$ immediately gives the square of dualities in Fig.~\ref{fig:1dsquare}, where going from the bottom left to the top right is $Y_{\rho_1}F_1$ and $F_1S$. This gives a mapping between the gauging operation $S$ and the QCA implemented by $Y_{\rho_1}$ via fermionization $F_1$.

\subsection{Generalized $\mathbb{Z}_2$ higher-form Wu gauging in 3+1d}
There is an analogous story in 3+1d with some subtle differences. In 3+1d, when gauging 1-form symmetries, we can instead of using an ordinary $\mathbb{Z}_2$ 2-form gauge field, use the Wu structure $\rho_2$. Intuitively, this is similar to gauging with fermionic ancillas: here we gauge with ancillas carrying a 1-form symmetry with background gauge field $\rho_2$. So writing $\mathcal{F}_{\rho_2}[Z[B]]=\tilde{S}_{\rho_2}T[Z[B]]$, we can view $\tilde{S}_{\rho_2}$ as a ``Wu gauging" similar to how $\tilde{S}_{\rho_1}$ is fermionic gauging. Here, $T[Z[B]]$ stacks with an invertible phase $e^{\frac{i\pi}{2}q_{\rho_2}(B)}$. To parallel the previous 1+1d story, we introduce the Wu structure $\rho_1$ as a $\mathbb{Z}_2$ $2$-cochain satisfying
\begin{equation}
    d\rho_2=\nu_3~.
\end{equation}
We can now obtain a quadratic refinement $q_{\rho_2}(b):H^2(M,\mathbb{Z}_2)\times H^2(M,\mathbb{Z}_2)\to\mathbb{Z}_4$ labeled by $\rho_2$ that again gives a polarization that mathces the bilinear pairing for gauging $\mathbb{Z}_2$ $1$-form symmetry:
\begin{equation}
    q_{\rho_2}(b+b')=q_{\rho_2}(b)+q_{\rho_2}(b')+2b\cup b'\quad\text{mod 4}~.
\end{equation}
On $2$-cocycles, there is a more explicit expression for $q_{\rho_2}(b)$:
\begin{equation}
    q_{\rho_2}(b)=\mathcal{P}(b)+2b\cup\rho_2~,
\end{equation}
where $\mathcal{P}$ is the Pontryagin square used in class (1) gravitational topological responses. From this expression, it is clear that $\sum_bZ[b]e^{\frac{i\pi}{2}\int q_{\rho_2}(b)}$ performs gauging with $\rho_2$ acting as the background gauge field. 
\begin{figure}[h]
   \centering
   \includegraphics[width=.6\columnwidth]{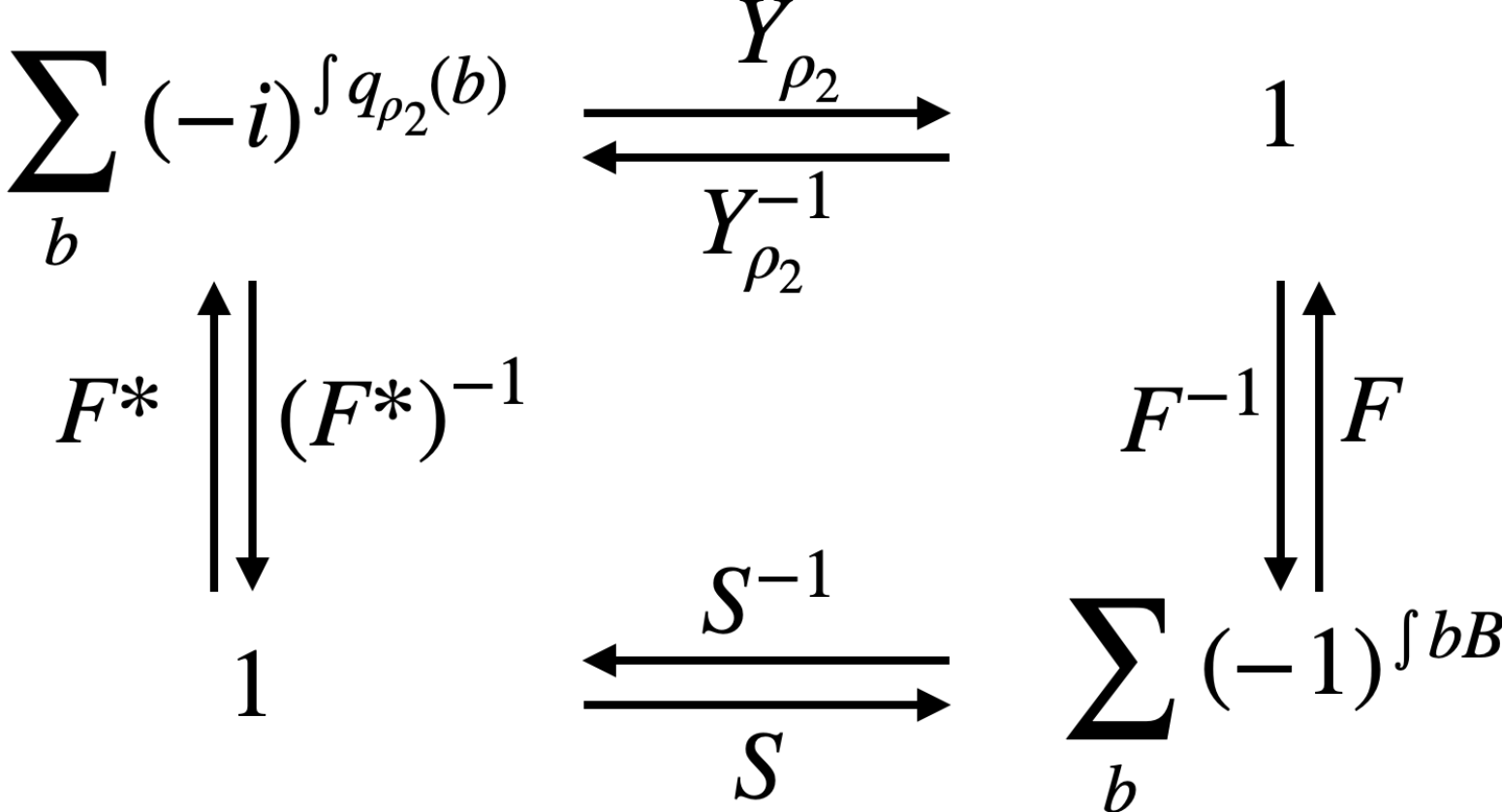} 
   \caption{Square of dualities in 3+1d. The bottom two phases are bosonic with background 2-form gauge field $B$ while the top two phases have background 2-form gauge field $\rho_2$. Gauging refers to gauging the $\mathbb{Z}_2$ 1-form symmetry (i.e. Kramers-Wannier duality) while stacking refers to stacking with $Y_{\rho_2}=\sum_be^{\frac{i\pi}{2}\int q_{\rho_2}(B)}$.}
   \label{fig:3dsquare}
   \end{figure}
   
Wu gauging in 3+1d is an operation similar to fermionization in 1+1d, and is implemented via two analogous steps:
\begin{enumerate}
    \item Stack the theory with a $\mathbb{Z}_2^{(1)}\times\mathbb{Z}_2^{(1)}$ SPT: $Z[B]\to Z[B]e^{\frac{i\pi}{2}\int q_{\rho_2}(B)}$. This SPT is a root $\mathbb{Z}_2$ $1$-form SPT in 3+1d \cite{Gaiotto:2014kfa,Hsin:2018vcg,PhysRevB.101.035101}, together with a $1$-form cluster state entangler. 
    Since $q_{\rho_2}(B)$ depends on $\rho$ by $\rho\cup B$,
    it is entangled by the composition of the root $\mathbb{Z}_2$ $1$-form SPT entangler and the 3+1d $1$-form cluster state entangler that represents the $\rho\cup B$ part.
    
    \item Sum over the background field for the $\mathbb{Z}_2^{(1)}$ symmetry $b$, leaving a theory with just a $1$-form symmetry with background gauge field $\rho_2$: $Z[B]e^{\frac{i\pi}{2}\int q_{\rho_2}(B)}\to \sum_bZ[b]e^{\frac{i\pi}{2}\int q_{\rho_2}(b)}$.
\end{enumerate}

We will denote this transformation by $F_2Z[B]$. As in the 1+1d case, there is an inverse operation $F_2^{-1}$ that involves a sum over $\rho_2$:
\begin{equation}\label{F2}
    F_2Z[B]=\sum_bZ[b]e^{\frac{i\pi}{2}\int q_{\rho_2}(b)}\qquad F_2^{-1}Z[\rho_2]=\sum_{\rho_2}Z[\rho_2]e^{\frac{-i\pi}{2}\int q_{\rho_2}(B)}~.
\end{equation}
It is clear from (\ref{F2}) that $F^{-1}$ can also be written as $T^{-1}\tilde{S}_{\rho_2}^{-1}$ where $\tilde{S}_{\rho_2}^{-1}$ gauges the $1$-form symmetry with background $\rho_2$ using an ordinary $2$-form background gauge field $B$. $F$ and $F^{-1}$ satisfy
\begin{equation}
    F_2^{-1}FZ[B]=\sum_{b,\rho_2}Z[b]e^{\frac{i\pi}{2}\int \mathcal{P}(b)+2b\rho_2}e^{-\frac{i\pi}{2}\int\mathcal{P}(B)+2B\rho_2}=Z[B]~,
\end{equation}
where we used the fact that the sum over $\rho_2$ sets $b=B$.

\subsubsection{Square of dualities in 3+1d}
Starting with the trivial theory $Z[B]=1$ we get an equation similar to (\ref{trivialmap}):
\begin{equation}
    Y_{\rho_2}F_{2}^*[1]=F_{2}[S[1]]=1~.
\end{equation}

Note however that we have $Y_{\rho_2}F_{2}^*[1]=Y_{\rho_2}\sum_be^{-\frac{i\pi}{2}\int q_{\rho_2}(b)}$ instead of $F_{2}[1]$. This is because in this case, $Y$ is order eight not order 2 (even on orientable manifolds) so in order for $Y_{\rho_2}$ to cancel with $F_2[1]$ we need an extra minus sign.

For more general theories, we have
\begin{align}
\begin{split}\label{genmap3d}
    F_2SZ[B]&=\sum_{b,b'}Z[b]e^{i\pi\int b\cup b'}e^{\frac{i\pi}{2}\int q_{\rho_2}(b')}\\
&=\left(\sum_{b''}e^{\frac{i\pi}{2}\int q_{\rho_2}(b'')}\right)\sum_{b}Z[b]e^{-\frac{i\pi}{2}\int q_{\rho_1}(b)}=Y_{\rho_2}F_2^*Z[B]~.
\end{split}
\end{align}
We can implement $F_2^*Z[B]=(KF_2K)[Z[B]]$ where $K$ is complex conjugation. It follows that similar to (\ref{YtoS}) we can write
\begin{equation}
    Y_{\rho_2}=F_2S(F_2^*)^{-1}=F_2S(KF_2 K)^{-1}~.
\end{equation}

How does (\ref{genmap3d}) relate to the $\mathrm{SL}(2,\mathbb{Z}_2)$ relation $(ST_{\rho_2})^3=Y_{\rho_2}$? Here, $S$ is the usual $\mathbb{Z}_2$ 1-form gauging so 
\begin{equation}
    ST_{\rho_2}[Z[b]]=\sum_bZ[b]e^{\frac{i\pi}{2}\int q_{\rho_2}(b)}e^{i\pi\int b\cup B}~.
\end{equation}
so we see that like in the 1+1d case, $F_2=\epsilon_0ST_{\rho_2}$ and $F_2^*=\epsilon_0ST_{\rho_2}^{-1}$ where again, $\epsilon_0$ is evaluation at trivial background: $\epsilon_0[Z]=Z[0]$. 

Now starting with $(ST_{\rho_2})^3=Y_{\rho_2}$, we can again rearrange to get $ST_{\rho_2}S=Y_{\rho_2}T_{\rho_2}^{-1}ST_{\rho_2}^{-1}$. In this case, we cannot simplify further because $T_{\rho_2}^{-1}\neq T_{\rho_2}$. Applying $\epsilon_0$ to both sides gives
\begin{equation}
    \epsilon_0ST_{\rho_2}S=Y_{\rho_2}\epsilon_0T_{\rho_2}^{-1}ST_{\rho_2}^{-1}~.
\end{equation}
Since the quadratic phase is trivial at zero, $\epsilon_0T_{\rho_2}^{-1}=\epsilon_0$ as in the 1+1d case. It follows that
\begin{equation}
    \epsilon_0ST_{\rho_2}S=Y_{\rho_2}\epsilon_0ST_{\rho_2}^{-1}\to F_2S=Y_{\rho_2}F_2^*~,
\end{equation}
matching (\ref{genmap3d}). Therefore, the relation (\ref{genmap3d}) mapping $S$ to the QCA $Y_{\rho_2}$ can be derived as a particular case of $(ST_{\rho_2})^3=Y_{\rho_2}$ obtained by evaluation at trivial background.

\subsection{Generalized higher-form Wu gauging and squares of dualities}

Consider generalized higher form fermionization
\begin{equation}
    F_nZ[B]=\sum_b Z[b] e^{\frac{\pi i}{2}\int q_{\rho_n}(b)}~,
\end{equation}
where in $2n$ spacetime dimension $d\rho_n=\nu_{n+1}$, and we omit an overall normalization in the sum over gauge fields.

The complex conjugated version is
\begin{equation}
    F^* Z[B]=\sum_b Z[b] e^{-\frac{\pi i}{2}\int q_{\rho_n}(b)}~.
\end{equation}
We note that
\begin{equation}
    \frac{\pi i}{2}\int q_{\rho_n}(b)-\left(-\frac{\pi i}{2}\int q_{\rho_n}(b)\right)=\pi \int b\cup b=\pi\int b\cup \nu_{n}~,
\end{equation}
where the last equality uses the Wu formula.
Thus, if the system depends on $\nu_n$ structure, i.e. $\nu_n=d\xi_{n-1}$, then the generalized higher form fermionization is real
\begin{equation}
    2n\text{ spacetime dimension } F_n=F_n^* \quad \text{iff } \nu_n=d\xi_{n-1}~.
\end{equation}

For example, when $n=1$, this means that the 1+1d spacetime has orientation, i.e. we do not enrich the 1+1d system with time-reversal symmetry. This is the case in Ref.~\cite{Ji:2019ugf}.

When $n=2$, this means that the 3+1d spacetime has $pin^-$ structure. This means that the 3+1d system has local fermions with time reversal symmetry that satisfies ${\cal T}^2=1$. For example, the Wu gauging $F_2$ for the trivial theory gives the effective theory for the chiral semion Walker Wang model, while $F_2^*$ produces the anti-semion Walker-Wang model. These two models are equivalent if there are local fermions, where a semion can be converted to an anti-semion by tensoring with a local fermion.

\subsection{$\mathbb{Z}_N$ $ST$ gauging in 3+1d }
For $\mathbb{Z}_N$ for more general $N$, we can only use an ordinary background $\mathbb{Z}_N$ 2-form gauge field for the 1-form symmetry rather than $\rho_2$. In this case, $\sum_bZ[b]e^{\frac{2\pi i}{2N}\int\mathcal{P}(b)+2b\cup B}$ (for $N$ even) simply performs $ST$. Using $(ST)^3=Y$, we can construct a square of dualities illustrated in Fig.~\ref{fig:generalN}. Again we have $S$ on the bottom edge, $ST$ on the right edge, $SCT^{-1}$ on the left edge, and $YT^{-1}$ on the upper edge. Here we used the fact that $S^2=C$ for more general $N$, so $STS=YT^{-1}S^{-1}T^{-1}=YT^{-1}SCT^{-1}$. Since $T$ is a finite depth circuit while $Y$ is a QCA for appropriate $N$, this still maps a QCA on the upper edge to gauging on the lower edge. Now the left/right edges are genuinely $ST$ and $ST^{-1}$ rather than twisted Wu gauging.

Specifically,
\begin{equation}
    Y_NZ[B]=Z[B]\sum_{b}e^{\frac{\pi i}{N}\int \mathcal{P}(b)}\quad (N\text{ even})\qquad Y_NZ[B]=Z[B]\sum_{b}e^{\frac{2\pi i(N+1)}{2N}\int \mathcal{P}(b)}\quad (N\text{ odd})~.
\end{equation}
\begin{figure}[h]
   \centering
   \includegraphics[width=.6\columnwidth]{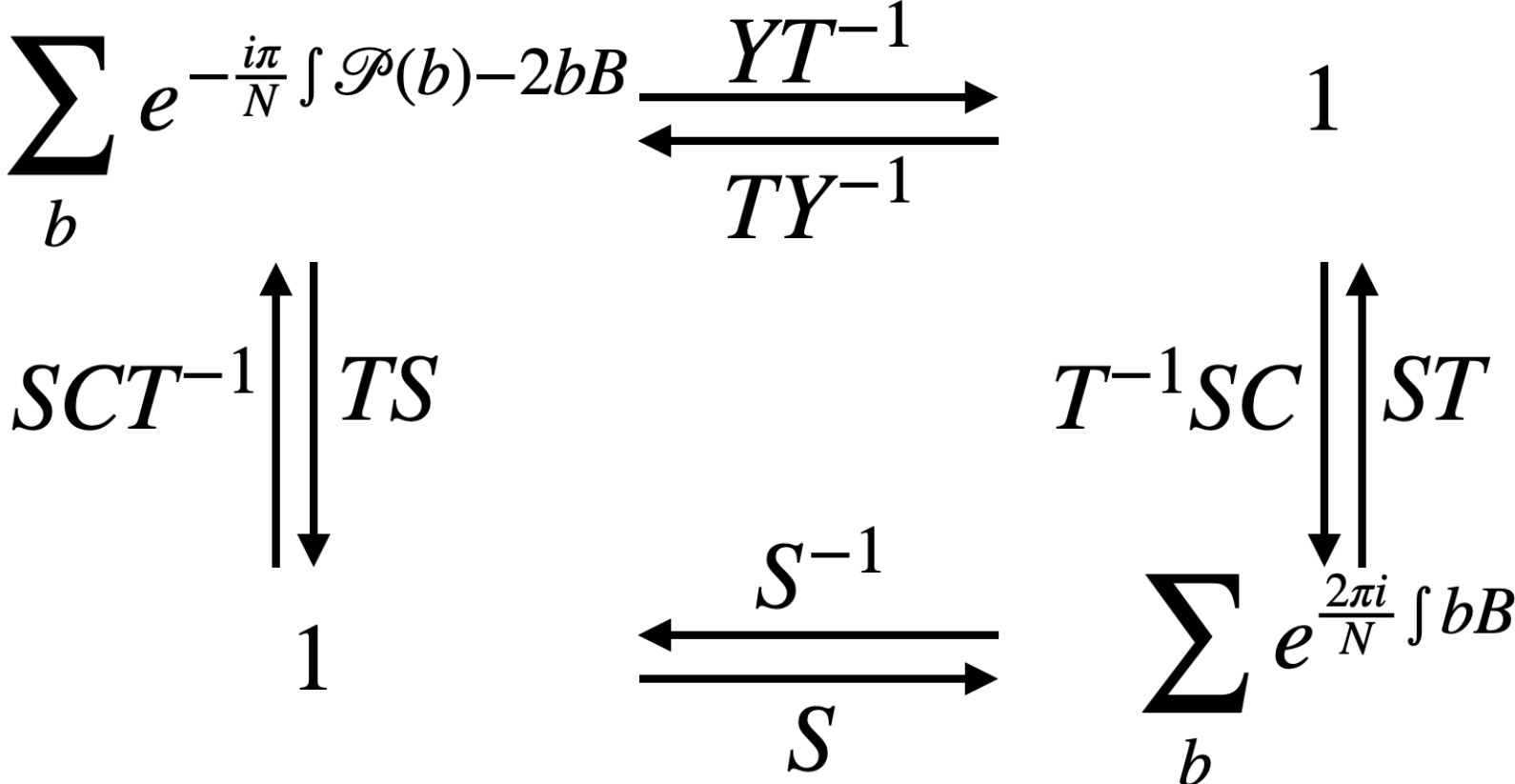} 
   \caption{Square of dualities for $\mathbb{Z}_N$ 1-form symmetries in 3+1d. Here $\mathcal{C}$ is charge conjugation. We gave labels for even $N$; for odd $N$ the definition of $\mathcal{T}$ is slightly different.}
   \label{fig:generalN}
   \end{figure}
   
Simplifying, we get the Pontryagin terms in (\ref{3dZNY}). Since $Y$ is a background independent gravitational invertible phase and $T^{-1}$ is a 1-form SPT entangler, $YT^{-1}$ entangles an invertible 1-form SPT (Walker-Wang model) with boundary chiral central central charge. For example, for $N=2$ this would give a boundary anti-semion theory with chiral central charge $-1$ mod 8 (note that there is a sign flip from the $+1$ mod 8 above because the anomalous framing dependence of the boundary must be the inverse of the bulk signature phase). Indeed, $YT^{-1}$ stacks with the same phase as $Y_{\rho_2}$ according to Ref.~\cite{Hsin:2021qiy}. This also works for other $N$, which would give other QCA along the upper edge, related to gauging on the bottom edge via twisted gauging. 

\subsection{$ST$ gauging for Stiefel-Whitney counterterms}
We can also massage the $\mathrm{SL}(2,\mathbb{Z}_2)$ relations involving Stiefel-Whitney terms into a square of dualities. The benefit of this presentation is that, as we will see in the next section, it will clarify the role of QCA in the lattice implementation of these operations. For $Y$ corresponding to a Stiefel-Whitney term, we have the same square of dualities as in the previous section: $S$ on the bottom leg, $ST$ on the right leg, $SCT^{-1}=ST$ on the left leg (since in this case we only have $\mathbb{Z}_2$ gauge fields and $T$ is order two), and $YT$ on the top leg. We illustrate this square of dualities for $Y=(-1)^{\int \nu_2^2}$ in Fig.~\ref{fig:stiefel}.
\begin{figure}[h]
   \centering
   \includegraphics[width=.6\columnwidth]{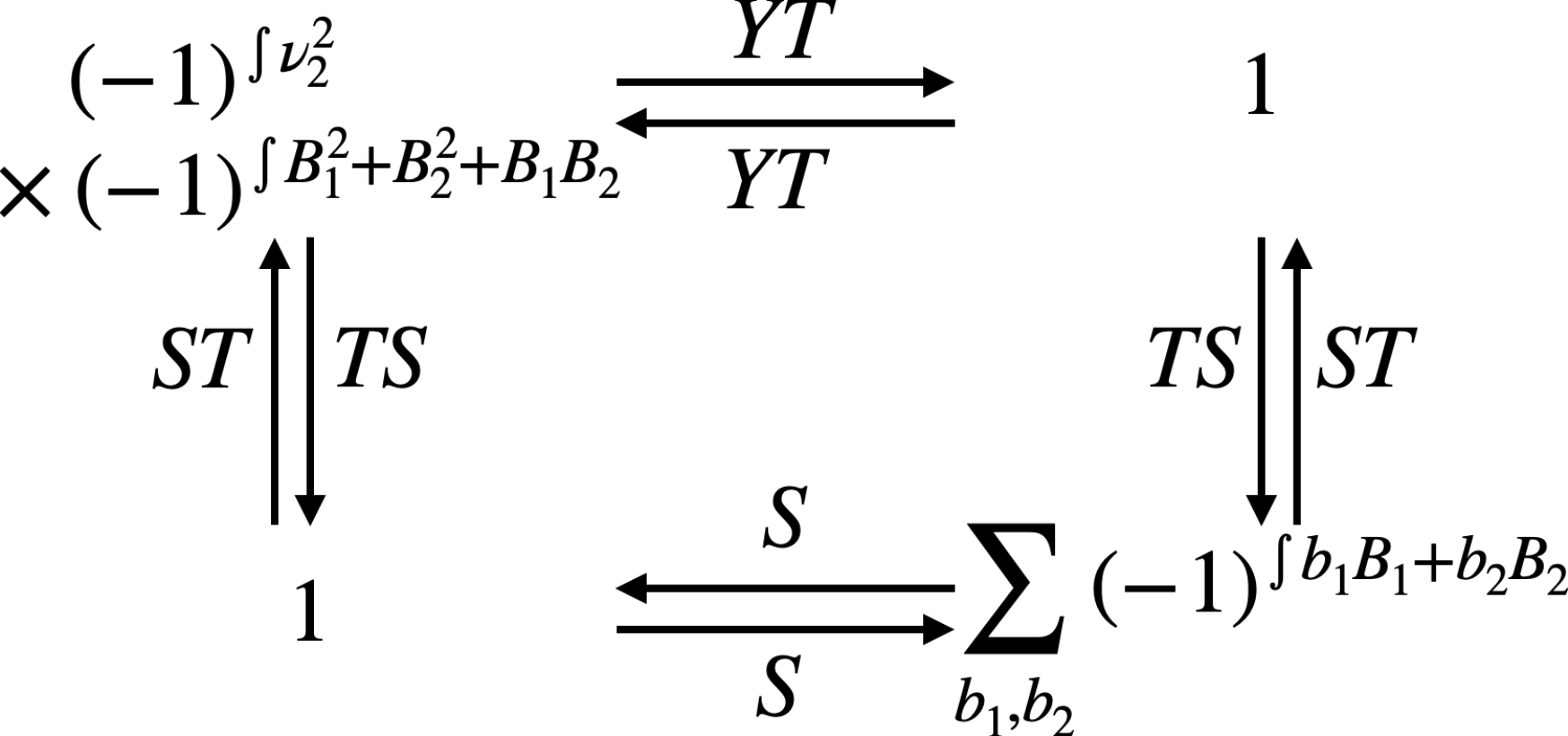} 
   \caption{Square of dualities for $Y$ corresponding to Stiefel-Whitney counterterms. Since everything is $\mathbb {Z}_2$, $T^2=S^2=Y^2=1$ so we do not need to use inverses.}
   \label{fig:stiefel}
   \end{figure}

\section{Lattice Models and QCA Implementation via $S$ and $T$}
\label{sec:lattice}

In this section, we will show how to implement all of the steps discussed in previous sections in lattice models. In the context of lattice models, it is of particular interest to distinguish states entangled from a product state by FDQC vs QCA. Here we must be careful with terminology: when $Y$ is a gravitational term that takes a nontrivial value on some orientable manifold, it corresponds to a QCA according to Conjecture~\ref{conj:nontrivialqca} from Ref.~\cite{Fidkowski:2024hpz}. However, on the lattice, two QCAs are equivalent if they differ by a FDQC. According to this definition, the QCA producing $Y$ is equivalent to that producing i.e. $YT^{-1}$. The important aspect is that the QCA is implemented on the lattice by a sequence of $S$ and $T$ operations, where $T$ is a FDQC and $S$ is implemented by an FDQC + measurement + error correction.

For a simple example that clarifies what we mean, consider the Stiefel-Whitney term $Y=(-1)^{\int \nu_2^2}=\sum_{b_1,b_2}(-1)^{\int b_1^2+b_2^2+b_1\cup b_2}$. On the lattice, one might think that this would be implemented by entangling $(-1)^{\int b_1^2+b_2^2+b_1\cup b_2}$ via a finite depth circuit similar to the cluster state entangler, and then measuring the $b_1$ and $b_2$ degrees of freedom to decouple them from the theory. But on the lattice, this would produce a trivial state with all of the degrees of freedom measured out. Indeed, the above is not an $ST$ operation; it is really an $\epsilon_0ST$ operation. While the first step is indeed a circuit, the second step is not quite gauging because it does not introduce new gauge field degrees of freedom. 

Instead, to get the 3-fermion Walker-Wang state, we really need to measure $b_1,b_2$ in the resource state $(-1)^{\int b_1^2+b_2^2+b_1\cup b_2+b_1\cup B_3+b_2\cup B_4}$. This gives
\begin{align}
\begin{split}
    \sum_{b_1,b_2}(-1)^{\int b_1^2+b_2^2+b_1\cup b_2+b_1\cup B_3+b_2\cup B_4}&=\sum_{b_1,b_2}(-1)^{\int (\nu_2+b_2+B_3)\cup b_1+(\nu_2+B_4)\cup b_2}\\
    &=(-1)^{\int(\nu_2+B_4)\cup (\nu_2+B_3)}\\
    &=(-1)^{\int \nu_2^2}(-1)^{\int B_3^2+B_4^2+B_3\cup B_4}=ST[1]~.
\end{split}
\end{align}
The state $(-1)^{\int \nu_2^2}(-1)^{\int B_3^2+B_4^2+B_3\cup B_4}$ rather than just $(-1)^{\int \nu_2^2}$ is the state entangled by the 3-fermion Walker Wang QCA, and is obtained by following the operation $T$ with an $S$ that introduces new degrees of freedom $B_3,B_4$ (see Fig.~\ref{fig:stiefel}). Summing over $b_1,b_2$, which on the lattice corresponds to measuring them, leaves a nontrivial wavefunction on the $B_3,B_4$ degrees of freedom that forms a 3-fermion Walker-Wang state (in fact, this precise process was discussed in Appendix B of Ref.~\cite{tantikramerswannier}). We can of course always compose the QCA that entangles the above phase with the finite-depth circuit that entangles $(-1)^{\int B_3^2+B_4^2+B_3\cup B_4}$ to isolate the gravitational term $(-1)^{\int \nu_2^2}$. In this case that simply corresponds to $TST[1]$. 

In the following, we will show how the three classes of $Y$ can produce states entangled by nontrivial QCA. By viewing these states as obtained from the left leg of their respective square of dualities (see Fig.~\ref{fig:stiefel}), we obtain a shortcut to these QCA-entangled states by using local unitary evolution, measurement, and error correction to implement $S$ and $T$. It is not clear whether or not all of the nontrivial QCA in class (3) can be obtained by a finite sequence of $T$ that is the exponential of a quadratic form and $S$ corresponding to the polarization of that quadratic form. We show that if we allow for more general $T$ that is not the exponential of a quadratic form, we can straightforwardly obtain all states of the form $U_{\mathrm{QCA}}|0\rangle$ (where $|0\rangle$ is a product state) by a FDQC followed by gauging of a finite number of higher-form symmetries.

We will also perform the nontrivial check that $ST[1]=YT^{-1}[1]$ is equivalent to $(ST)^3[1]=Y$ in the sense that they give stabilizer groups related by a FDQC. To do so, we use the polynomial formalism for translation-invariant stabilizer codes \cite{Haah:2013oba}. We will focus on a particular $\mathbb{Z}_3$ example, but a similar calculation likely applies for more general $\mathbb{Z}_N$, at least for odd $N$. For even $N$, the QCA may not be Clifford, in which case the polynomial formalism would not help.

In the following discussion, we will focus on the lattice description for QCAs from the gravitational Pontryagin responses and Stiefel-Whitney classes responses. The lattice descriptions for the QCA from ABK invariants have been discussed in the literature, e.g. the Kitaev chain (see e.g. Ref.~\cite{Seiberg:2023cdc}). To construct the generalizations in higher dimensions that depend on higher Wu structures \cite{Hsin:2021qiy} using the higher-form gauge theories, one needs to first understand how to present Wu structures on the lattice similar to the spin structure (see e.g. Refs.~\cite{Bhardwaj:2016clt,Chen_2018}). For a statesum model construction, see Ref.~\cite{Kobayashi:2021pst}.  We will reserve the problem of constructing more explicit wavefunction/stabilizer lattice models for future study.

 \subsection{QCA involving gravitational Pontryagin response}
 \label{sqcapontryagin}
 \begin{figure}[h]
   \centering
   \includegraphics[width=.75\columnwidth]{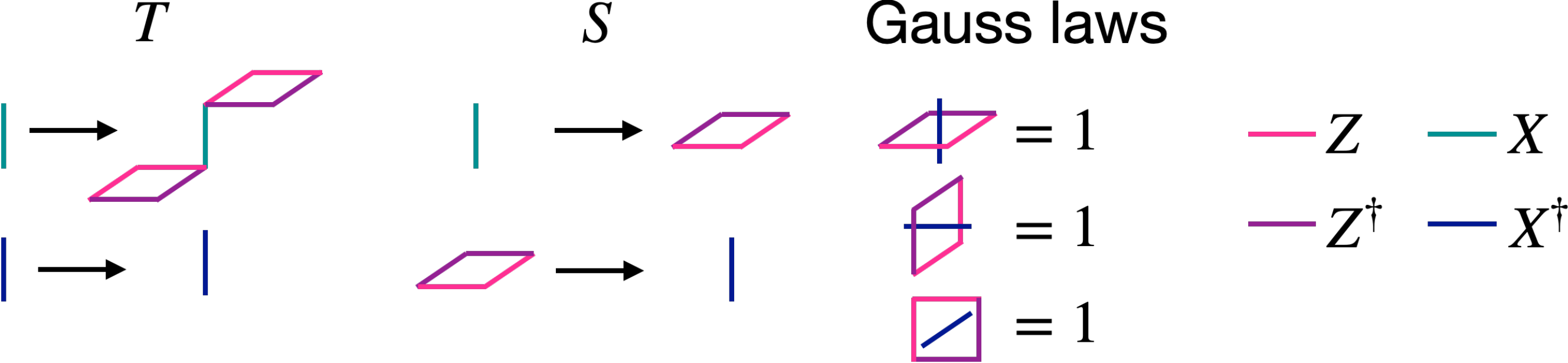} 
   \caption{The $T,S$ relations for $\mathbb{Z}_3$ $1$-form symmetry. Here $T$ entangles the SPT $e^{\frac{2\pi i 2}{3}\int\mathcal{P}(B)}$ and $S$ performs gauging (see Ref.~\cite{Cordova:2023bja}).}
   \label{fig:TSZ3}
   \end{figure}
   
 \begin{figure}[h]
   \centering
   \includegraphics[width=.9\columnwidth]{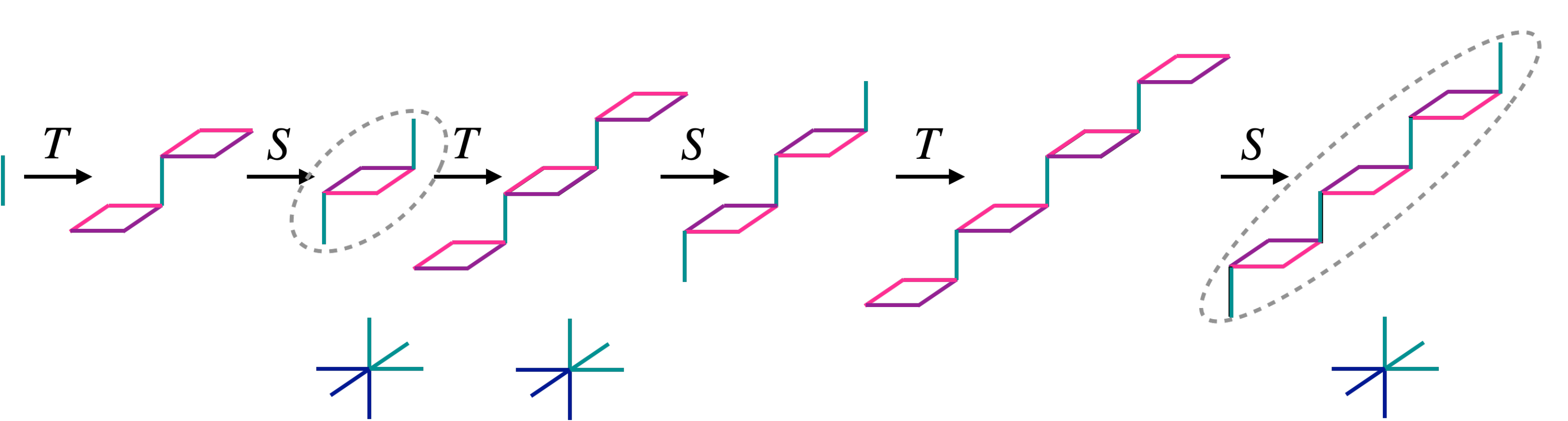} 
   \caption{The action of $S$ and $T$ (up to translations) on $\mathbb{Z}_3$ qudits, starting from a product state stabilizer. At $ST[1],TST[1]$, and $(ST)^3[1]$, there is an additional vertex stabilizer added to the stabilizer group in addition to the above stabilizer, its analogous versions in the other two directions, and their translations. Even though $TST[1]=(ST)^3[1]$ in field theory, the lattice stabilizers differ, so it is nontrivial to show that the state specified by the stabilizers are related by a FDQC. Here we will show instead that the simpler stabilizers from $ST$ are related to those of $(ST)^3$ by a FDQC, since those of $ST$ are related to those of $TST$ by a FDQC $T$.}
   \label{fig:ST3Z3}
   \end{figure}
Our first example of using $S,T$ to obtain a state entangled by a QCA is when $Y$ is a gravitational Pontryagin counterterm. Here our $ST[1]=YT^{-1}[1]$ calculation recovers results from Ref.~\cite{sunqca} and written implicitly in Ref.~\cite{Cordova:2023bja}. Note however that Ref.~\cite{sunqca} argues for the nontriviality of their QCA from the boundary anyon theory of the entangled Walker-Wang model; here we obtain the same result via the nontrivial bulk Pontryagin term $Y$, which gives a signature that must be canceled by a chiral boundary theory. We will also show that $(ST)^3[1]=Y[1]$ in the sense that it is FDQC equivalent to $YT^{-1}[1]$. This turns out to be nontrivial because the lattice model obtained from $(ST)^3$ on a product state looks quite different from that of $ST[1]$.

For simplicity, we will work with a $\mathbb{Z}_3$ 1-form symmetry on a cubic lattice. In this case, $T[1]=e^{\frac{2\pi i 2}{3}\int\mathcal{P}(B)}$. Since we are working with $\mathbb{Z}_3$ symmetry, in this section we will always use $\mathbb{Z}_3$ qutrits with $Z_e^3=X_e^3=1$ and $X_eZ_e=e^{-2\pi i/3}Z_eX_e$. The SPT entangler was discussed in Refs.~\cite{Tsui:2019ykk,Cordova:2023bja} and is presented in Fig.~\ref{fig:TSZ3}. We also present the gauging map in Fig.~\ref{fig:TSZ3}. One can read off from the gauging map that $S^2=C$, which acts nontrivially on $\mathbb{Z}_3$ qudit operators. 

The combination $ST[1]$ gives a model that matches that of Ref.~\cite{PhysRevB.107.125425} up to a Hadamard on every edge and a charge conjugation operator on every $\hat{x}$ edge. This was shown in Ref.~\cite{sunqca} to be a nontrivial QCA (under standard assumptions) due to the chiral anyon boundary theory. On the other hand, $(ST)^3$ gives much more complicated stabilizers, as shown in Fig.~\ref{fig:ST3Z3}. A priori, it is not obvious that these two sets of stabilizers define the same QCA, even though this is indicated from field theory. Indeed, in field theory we have
\begin{equation}
    TST[1]=Y[1]\qquad (ST)^3[1]=Y[1]~.
\end{equation}
so $ST[1]$ should give $(ST)^3[1]$ up to a FDQC corresponding to $T^{-1}$, but on the lattice, $TST$ and $(ST)^3$ may differ by an additional FDQC. We show in this section (with details in Appendix \ref{sec:STST3}) that this is indeed the case. Our approach will be to use a bulk boundary correspondence: we will show that the stabilizers obtained from $ST$ and $(ST)^3$ give rise to equivalent boundary algebras/boundary chiral topological orders. Following standard assumptions, this means that the bulk QCA are equivalent. 

We will go into more detail on the approach here, and will summarize the high level steps at the end of this section. We will show that the stabilizers obtained from $ST$ on the product state stabilizers $\{X_e\}$ (where $e$ indicates that it acts on an edge of the cubic lattice) can be written as a set of ``locally flippable separators," $\mathcal{S}_{ST}=\{UX_eU^\dagger\}$ which are qutrit $X_e$ operators conjugated by a locality preserving Clifford QCA. $(ST)^3$ on the product state stabilizers can also be written as a set of locally flippable separators $\mathcal{S}_{(ST)^3}$. This step is important because the stabilizers in Fig.~\ref{fig:ST3Z3} (if we include the two rotated versions shown in Fig.~\ref{fig:Z3example} together with all of their cubic lattice translations) are not independent: multiplying six plaquette terms around a cube gives a product of two star terms\cite{Sun:2025vta}. For the purpose of checking whether two sets of stabilizers correspond to the same QCA, it is important to first convert both sets into locally flippable separator generators. Indeed we show in Appendix \ref{sec:STST3} that it would be impossible to map the individual stabilizers circled in (\ref{fig:ST3Z3}) together with the two rotated versions and their translations directly, due to incompatibility of the ground state degeneracies of a Hamiltonian constructed from the $ST$ stabilizers vs $(ST)^3$ stabilizers of Fig.~\ref{fig:ST3Z3}) if we omit the star terms.

Once we have written the stabilizers from $ST$ and $(ST)^3$ as $\mathcal{S}_{ST}$ and $\mathcal{S}_{(ST)^3}$ respectively, it remains to show that these sets of locally flippable separators are equivalent, meaning they can be related by a FDQC:
\begin{equation}
    (ST)[1]\sim (ST)^3[1]: \mathcal{S}_{ST}\sim \mathcal{S}_{(ST)^3}~.
\end{equation}

To do so, we proceed by the polynomial approach in Ref.~\cite{Haah:2013oba}. The polynomial formalism allows for efficient encoding of translation-invariant stabilizers. Although all of these stabilizers commute in the 3+1d bulk, introducing a physical boundary causes boundary stabilizers to get truncated. After truncation, the boundary stabilizers may fail to commute with each other. Their commutation relations are encoded in a skew-Hermitian matrices $\Xi_{ST}, \Xi_{(ST)^3}$. The two sets of locally flippable separators are equivalent \cite{Haah:2019fqd}, meaning they are obtained via equivalent QCA, if and only if
\begin{equation}
    E^\dagger(\Xi_{ST}\oplus \lambda_q)E=\Xi_{(ST)^3}\oplus \lambda_{q'}~,
\end{equation}
where $\lambda_q,\lambda_{q'}$ are standard symplectic matrices of dimension $2q$ and $2q'$ respectively and $E$ is an invertible matrix. Physically, $\lambda_q$ and $\lambda_{q'}$ represent adding $q,q'$ ancillary boundary qutrits per boundary unit cell, each with a trivial onsite separator and its conjugate local flipper.
The existence of such invertible $E$ means that the boundary algebras are related by a change of basis, so they carry the same anomaly/chiral central charge. While this does not construct an explicit FDQC to relate $\mathcal{S}_{ST}$ and $\mathcal{S}_{(ST)^3}$, this does indicate that such an FDQC exists. 

Indeed, we will find that
\begin{equation}\label{boundarymap}
    E^\dagger \Xi_{ST}E=\lambda_1\oplus\lambda_1\oplus\Xi_{(ST)^3},\qquad\lambda_1=\begin{pmatrix} 0 & 1\\ -1 & 0\end{pmatrix}~,
\end{equation}
indicating that the two stabilizer groups are related by a FDQC. To summarize, our line of reasoning is as follows:
\begin{enumerate}
    \item Obtain locally flippable separators $\mathcal{S}_{ST}, \mathcal{S}_{(ST)^3}$.
    
    \item Introduce a boundary along the $(1,1,1)$ direction (this is natural because the stabilizers propagate in this direction).
    
    \item Compute $\Xi_{ST}$ and $\Xi_{(ST)^3}$, which encode the commutation relations of the boundary stabilizers.
    
    \item Show that they can be related by (\ref{boundarymap}).
\end{enumerate}
\subsubsection{Construction of locally flippable separators}
We place a qutrit on every edge of the cubic lattice and denote the qutrit on the $a$-oriented edge anchored at the vertex $r$ by $(a,r)$ where $a\in\{x,y,z\}$. For cyclically ordered $(a,b,c)$ we define the oriented $ab$ plaquette
\begin{equation}
    P_{ab}(r)=Z_{a,r}Z_{b,r+\hat{a}}Z_{a,r+\hat{b}}^\dagger Z_{b,r}^\dagger~.
\end{equation}
In this notation, the vertex star is 
\begin{equation}
    A_v(r)=\prod_{a=x,y,z}X_{a,r}X_{a,r-a}^\dagger~.
\end{equation}
$(ST)[1]$ and $(ST)^3[1]$ give two different plaquette type operators. Denoting by $c$ the perpendicular direction, we have:
\begin{equation}
    B_{c,r}=B_{ab}(r)X_{c,r-\hat{c}}X_{x,r+d-\hat{c}},\qquad \tilde{B}_{c,r}=\left(\prod_{n=0}^2B_{ab}(r+nd)\right)\left(\prod_{m=0}^3X_{c,r+md-\hat{c}}\right)~.
\end{equation}
The anchor $r-\hat{c}$ puts the first leg immediately before the first plaquette, and $r+d-\hat{c}=r+\hat{a}+\hat{b}$ puts the second leg at its opposite corner (here $d=\hat{a}+\hat{b}+\hat{c})$. This notation takes car of all translations and orientations of the stabilizers. We have the following stabilizer groups for $(ST)[1]$ and $(ST)^3[1]$:
\begin{equation}
    \tilde{\mathcal{S}}_{ST}=\langle A_v(r),B_{c,r}\rangle,\qquad \tilde{\mathcal{S}}_{(ST)^3}=\langle A_v(r),\tilde{B}_{c,r}\rangle~,
\end{equation}
where $\langle...\rangle$ means it is ``generated by." Here we use $\tilde{\mathcal{S}}$ because we will later define modified stabilizers forming locally flippable separators by $\mathcal{S}$.

We will use the polynomial formalism introduced in Ref.~\cite{Haah:2013oba}. For a translation-invariant Pauli stabilizer problem in 3+1d, we can use the Laurent ring $\mathcal{R}=\mathbb{F}_3[x^{\pm},y^{\pm},z^{\pm}]$. The monomial $x^iy^jz^k$ translates an operator by $i\hat{x}+j\hat{y}+k\hat{z}$. A Pauli operator is a vector
\begin{equation}
    p=\begin{pmatrix} p_X\\ p_Z\end{pmatrix}\in \mathcal{R}^6,\qquad p_x,p_Z\in \mathcal{R}^3~
\end{equation}
ordered as $(X_x,X_y,X_z|Z_x,Z_y,Z_z)^T$. We will sometimes use a stabilizer i.e. $B_x$ interchangeably with its polynomial vector representation. A coefficient $1$ or $-1$ gives the corresponding qutrit Pauli exponent. We also have the standard symplectic form 
\begin{equation}
    \lambda_3=\begin{pmatrix} 0 & I_3\\ -I_3 & 0\end{pmatrix}~,
\end{equation}
and let $M^\dagger =\overline{M}^T$ where $\overline{f(x,y,z)}=f(x^{-1},y^{-1},z^{-1})$. Then for Pauli strings $p$ and $q$, $p$ commutes with every translate of $q$ if and only if
\begin{equation}
    p^\dagger\lambda_3q=0~.
\end{equation}
If there are $t$ translation types of stabilizer generators, we can collect them as columns of a $6\times t$ matrix $\sigma$. The commuting condition is
\begin{equation}
    \sigma^\dagger \lambda_3\sigma=0~.
\end{equation}
The excitation map is 
\begin{equation}
    \epsilon=\sigma^\dagger\lambda_3~,
\end{equation}
and a local multiplicative identity among stabilizers (i.e. a product of stabilizers being equal to 1) is a vector in $\mathrm{ker}(\sigma)$. 

A presentation can contain redundant Hamiltonian generators. To study QCA structure, we need to massage a presentation into a \emph{locally flippable separator}, which is a free set of commuting stabilizers $\mathcal{S}=(\mathcal{S}_x,\mathcal{S}_y,\mathcal{S}_z)$ and local Pauli columns $\mathcal{F}=(\mathcal{F}_x,\mathcal{F}_y,\mathcal{F}_z)$ satisfying
\begin{equation}\label{separator}
    \mathcal{S}^\dagger\lambda_3\mathcal{S}=0,\qquad\mathcal{S}^\dagger\lambda_3\mathcal{F}=I_3,\qquad \mathcal{F}^\dagger\lambda_3\mathcal{F}=0~.
\end{equation}

Then $(-F|S)$ is a symplectic Laurent matrix (up to a harmless sign convention) and defines a translation invariant Clifford QCA. The $\mathcal{S}_1,\mathcal{S}_2,\mathcal{S}_3$ can be thought of as the conjugated $X_{x,r},X_{y,r},X_{z,r}$ operators while $\mathcal{F}_1,\mathcal{F}_2,\mathcal{F}_3$ are the conjugated $Z_{x,r},Z_{y,r}Z_{z,r}$ operators under the QCA, since they have the same commutation relations.

We now define our stabilizers $\tilde{S}_{ST}$ and $\tilde{S}_{(ST)^3}$ in the polynomial formalism. Let
\begin{equation}
    u=xyz,\qquad \delta=\begin{pmatrix} 1-x\\1-y\\1-z\end{pmatrix},\qquad D=\mathrm{diag}(x^{-1},y^{-1},z^{-1})~.
\end{equation}
The star term is then
\begin{equation}
    a=-D\delta =\begin{pmatrix} 1-x^{-1}\\ 1-y^{-1}\\1-z^{-1}\end{pmatrix}~.
\end{equation}
The three plaquette terms ordered by normals $x,y,z$ form the columns of
\begin{equation}
    C=\begin{pmatrix}
        0 & z-1 & -1y\\ 1-z & 0 & x-1\\ y-1 & 1-x & 0\end{pmatrix}~.
\end{equation}
Using column order $(A_v,B_x,B_y,B_z)$ and $(A_v,\tilde{B}_x,\tilde{B}_y,\tilde{B}_z)$ we have the following compact encoding of the stabilizers:
\begin{equation}
    \sigma_{ST}=\begin{pmatrix} a & (1+u)D\\ 0 & C\end{pmatrix},\qquad \sigma_{(ST)^3}=\begin{pmatrix} a & (1+u+u^2+u^3)D\\ 0 & (1+u+u^2)C\end{pmatrix}~.
\end{equation}
In particular, $(1+u)c^{-1}$ gives the two $X_c$ legs for the $ST$ stabilizers at $r-\hat{c}$ and $r+d-\hat{c}$. For the $(ST)^3$ stabilizers, $1+u+u^2$ gives the three plaquettes at $r,r+d,r+2d$ and $(1+u+u^2+u^3)c^{-1}$ gives the four $X_c$ legs.

To check that all of the stabilizers commute, we compute as a sanity check,
\begin{equation}
    \sigma_{ST}^\dagger\lambda_3\sigma_{ST}=0,\qquad\sigma_{(ST)^3}^\dagger\lambda_3\sigma_{(ST)^3}=0~.
\end{equation}
This should hold because each set of stabilizers were obtained by a finite sequence of FDQCs and $\mathbb{Z}_2$ higher form symmetry gauging operations. 

The two four generator presentations have relations
\begin{equation}
    r_{ST}=\begin{pmatrix} 1+u\\ 1-x\\1-y\\1-z\end{pmatrix},\qquad r_{(ST)^3}=\begin{pmatrix}1+u+u^2+u^3\\1-x\\1-y\\1-z\end{pmatrix}~.
\end{equation}
$r_{ST}$ comes from the fact that multiplying the six plaquette terms $B_{c,r}$ around a plaquette cancels all of the $Z_{a,r}$ operators and produces vertex stars at two opposite vertices, and $r_{(ST)^3}$ comes from a similar relation on the larger $(ST)^3$ stabilizers. We can check that $\sigma_{ST}r_{ST}=\sigma_{(ST)^3}r_{(ST)^3}=0$, so i.e. $(1+u)A_v+(1-x)B_p+(1-y)B_y+(1-z)B_z$ represents the identity Pauli operator. Due to the redundancies, we would like to use a different generating set of stabilizers that give a locally flippable separator. Specifically, we would like to find a local reversible change of basis $V_L$ such that
\begin{equation}
    \sigma_{ST}V_{ST}=(0,S_x,S_y,S_z),\qquad\sigma_{(ST)^3}V_{(ST)^3}=(0,\tilde{S}_x,\tilde{S}_y,\tilde{S}_z) ~.
\end{equation}

We obtain
\begin{equation}
    V_{ST}=\begin{pmatrix} 1+u & -yz & -z & -1\\ 1-x & 1 & 0 & 0\\ 1-y & 0 & 1 & 0\\ 1-z & 0 & 0 & 1\end{pmatrix},\qquad \qquad V_{(ST)^3}=\begin{pmatrix} f_R & -hyz & -hz & -h\\ 1-x & 1 & 0 & 0\\ 1-y & 0 & 1 & 0\\ 1-z & 0 & 0 & 1\end{pmatrix}~,
\end{equation}
where $h=u(1-u)$ and $f_R=1+u+u^2+u^3$. Note that the first column is simply $r_{ST}$ and $r_{(ST)^3}$, ensuring that $\sigma_{ST}V_{ST}$ and $\sigma_{(ST)^3}V_{(ST)^3}$ having zero as their first components. We also have
\begin{align}
\begin{split}
    S_x&=B_x-yzA_v,\qquad S_y=B_y-zA_v,\qquad S_z=B_z-A_v\\
    \tilde{S}_x&=\tilde{B}_x-u(u-1)yzA_v,\qquad \tilde{S}_y=\tilde{B}_y-u(u-1)zA_v,\qquad \tilde{S}_z=\tilde{B}_z-u(u-1)A_v~.
\end{split}
\end{align}
The top line recovers the locally flippable separator construction in Ref.~\cite{Sun:2025vta}, while the bottom line is new. We also need to find local flippers $\mathcal{F}$ satisfying (\ref{separator}). We find these by computer search and specify them in Appendix~\ref{sec:STST3}. Thus both stabilizer states admit exact locally flippable separators, defining parent Hamiltonians obtained from the paramagnet Hamiltonian via a (Clifford) QCA mapping $X_{a,r}\to \mathcal{S}_{a,r}$ and $Z_{a,r}\to \mathcal{F}_{a,r}$. Finally we define the $6\times 6$ matrix
\begin{equation}
    Q=(\mathcal{S}|\mathcal{F})~,
\end{equation}
which completely specifies the translation invariant Clifford QCA and satisfies $Q^\dagger\lambda_3Q=\lambda_3$ (note that we use a slightly different convention compared to Ref.~\cite{Haah:2019fqd,Sun:2025vta} which defines $Q=(-{\cal F}|{\cal S})$ with separators from the last few columns of $Q$).

\subsubsection{Boundary skew-Hermitian forms}
Cutting perpendicular to one lattice direction produces a boundary algebra over a Laurent ring in $d-1$ variables. The commutation relations of a free boundary basis form a nonsingular skew-Hermitian matrix $\Xi=-\Xi^\dagger$. Two QCAs with boundary forms $\Xi_{ST}$ and $\Xi_{(ST)^3}$ are equivalent precisely when
\begin{equation}\label{equivqca}
    E^\dagger(\Xi_{ST}\oplus\lambda_q)E=\Xi_{(ST)^3}\oplus\lambda_{q'}~.
\end{equation}
The $\lambda_q$,$\lambda_{q'}$ factors can be thought of as $q,q'$ additional qutrits added to the boundary unit cells. Since $\lambda_q=\begin{pmatrix} 0 & 1_q\\ -1_{q} & 0\end{pmatrix}$ is of the standard symplectic form, they are trivial factors that do not need to be further disentangled. Here, $E$ is an invertible Laurent matrix.

For the case at hand, the stabilizers naturally propagate along the diagonal $(1,1,1)$ direction, so we choose $u=xyz$ as the normal variable and retain $y,z$ as tangential variables. On monomials this is a change of exponent coordinates
\begin{equation}
    x^ay^bz^c=u^ay^{b-a}z^{c-a}~.
\end{equation}
For each separator type, we shift its minimum $u$ power to zero, translate it across the cut, retain its nonnegative-$u$ part (these give the truncated boundary stabilizers), and compute all boundary commutators. The three $(ST)$ separator columns have $u$ ranges $(-1,0),(-1,1),(-1,1)$ so there are five crossing generators (one from $\mathcal{S}_x$ and two each from $\mathcal{S}_y,\mathcal{S}_z$). The $(ST)^3$ ranges are $(-1,2),(0,3),(0,3)$ resulting in nine crossing generators. 

To obtain the boundary skew-Hermitian matrix for $\Xi_{ST}$, we will proceed as follows. A similar method constructs $\Xi_{(ST)^3}$. For $\Xi_{ST}$, we established that the largest size of a separator in the normal direction is 2, so we will group 2 consecutive layers into one supercell with an internal $\mathbb{Z}_2$ flavor index. The resulting $Q^{(2)}$ is $12\times 12$ (more generally, after grouping $k$ layers into a supercell, we would get $Q^{(k)}$ that is $6k\times 6k$). Let $s=u^k$. On the internal layer labels, one original $u$ translation is represented by 
\begin{equation}
    T_k(s)=\begin{pmatrix} 0 & 0 & \cdots & 0 & s\\ 1 & 0 & \cdots & 0 & 0\\ 0 & 1 & \cdots & 0 & 0\\ \vdots & \vdots  & \cdots & \vdots & \vdots \\ 0 & 0 & \cdots  & 1 & 0\end{pmatrix},\qquad T_k(s)^k=sI_k~.
\end{equation}
For example, $T_2(s)=\begin{pmatrix} 0 & s\\ 1 & 0\end{pmatrix}$. Then $Q^{(k)}(s,y,z)=Q(T_k(s),y,z)$ where we replace $u$ by $T_k(s)$. 

If $k$ is at least the total normal width of $Q$, then after choosing the relative origin of the input and output supercells, the enlarged matrix has the two-supercell form 
\begin{equation}
    Q^{(k)}=A+sB~,
\end{equation}
where $A,B$ contain no $s$. To get $\Xi$ we actually do not need the flippers (since we are only truncating the Hamiltonian, which is constructed from the separators) so we only need $k$ to be the width of the separators. We use
\begin{equation}
    S^{(k)}=A_S+sB_S~.
\end{equation}
Physically, $A_S$ is the part of each separator lying in the same supercell as its anchor site, while $sB_S$ is the part lying in the next normal supercell. Since the cut is placed between these two supercells, $B_S$ records the separator pieces crossing the cut. We then define the Pauli commutation matrix 
\begin{equation}
    \Xi_{\mathrm{raw},ST}=B_S^\dagger\lambda_{3k}B_S~
\end{equation}
in one enlarged supercell, and print this explicitly in Appendix \ref{sec:STST3}. This matrix records how each boundary operator $b_i$ which is a column of $B_S$ commutes with every boundary translation of $b_j$. $\Xi_{\mathrm{raw},ST}$ is $5\times 5$ while $\Xi_{\mathrm{raw},(ST)^3}$ is $9\times 9$ coming from the 5 (resp. 9) stabilizers that cross the cut respectively. However, these matrices have rank 2 and 6 respectively (using ordinary elimination, temporarily allowing ratios of polynomials in $y,z$), so the raw crossing lists are redundant. In other words, three combinations of the five crossing strings commute with every crossing string. The QCA classification theorem assumes that these null combinations have been removed, so it cannot directly be applied to $\Xi_{\mathrm{raw},ST}$. 

We now identify from the five raw strings (five columns of $\Xi_{\mathrm{raw},ST}$) two local generators. We need to be careful here because for the actual local generators we do not want ratios of polynomials, which may be nonlocal. Let the five boundary crossing strings be $g_{x1},g_{y1},g_{y2},g_{z1},g_{z2}$.
The $5\times 5$ commutator matrix has rank two because the following combinations commute with all five strings:
\begin{equation}
        q_{y1}-zg_{z1},\qquad g_{x1}-(1+yz)g_{z1},\qquad (1+yz)g_{z1}-yz[(1-y)g_{y2}+(1-z)g_{z2}]~.
    \end{equation}
These null directions carry no information about the noncommuting boundary algebra. Any vector (boundary Pauli string) can be therefore written as $c_1e_1+c_2e_1+$ terms that commute with everything. After ignoring these commuting combinations, we get the explicit forms for the remaining two independent noncommuting directions $e_1,e_2$ as the columns of (see also Ref.~\cite{sunqca})
\begin{equation}
    \Xi_{ST}=\begin{pmatrix} z^{-1} -z & -1+z^{-1}+y+yz^{-1}\\ -z+1-zy^{-1}-y^{-1} & y-y^{-1}
    \end{pmatrix}~,
\end{equation}
which satisfies $\Xi_{ST}^\dag = -\Xi_{ST}$ and $\mathrm{det}(\Xi_{ST})=1$ so neither of these two boundary operators commutes with the entire boundary algebra; each has a nontrivial conjugate partner. 

By a similar approach, we obtain a $6\times 6$ matrix $\Xi_{(ST)^3}$. We print this in Appendix~\ref{sec:STST3}. All that remains is to check that (\ref{equivqca}) is satisfied. We show this by constructing $E$ explicitly in Appendix~\ref{sec:STST3}. Note that $\lambda_1\oplus\lambda_1$ correctly modifies the $5\times 5$ matrix $\Xi_{ST}$ to match the dimension of $\Xi_{(ST)^3}$. 

This establishes that the QCAs corresponding to $(ST)$ and $(ST)^3$ are equivalent.

\subsection{QCAs involving Stiefel-Whitney topological responses}
 \label{sstiefellattice}
In general, for $Y=(-1)^{\int w_iw_j\cdots w_k}$, we obtained in Sec.~\ref{sec:stiefelwhitneyst} an expression
\begin{equation}
    Y=\sum_{b,c}(-1)^{\int bc+bw_i+c(w_j\cdots w_k)}=\epsilon_0ST[1]~,
\end{equation}
where $T[1]=(-1)^{\int bc+bw_i+c(w_j\cdots w_k)}$. We also showed in (\ref{STYT}) that $ST[1]=YT^{-1}[1]$. Therefore, to get a state entangled by the QCA corresponding to $YT^{-1}$ we simply need to act $ST$ on the trivial product state. However, in Sec.~\ref{sec:stiefelwhitneyst} we used $T$ that includes direct coupling to Stiefel-Whitney classes. This can often be rewritten to be entirely in terms of cohomological operations on  $\mathbb{Z}_2$ background gauge fields via Wu formulas (see e.g. Refs.~\cite{milnor1974characteristic,Chen:2021xks,Chen:2021ppt}). These can be implemented in lattice models via the methods in Ref.~\cite{Chen:2021ppt}. Specifically, we can obtain
\begin{itemize}
    \item For $Y=(-1)^{\int \nu_2^2}$, we can write $T=(-1)^{\int \nu_2(B_1+B_2)+B_1B_2}=(-1)^{\int B_1^2+B_2^2+B_1B_2}$ which can be entangled by a FDQC on two flavors of qubits. $S$ proceeds via FDQC coupling to two more flavors of qubits, measurement, and error correction (see Appendix B of Ref.~\cite{PhysRevX.14.021040}). A similar construction holds for the 7+1d version with $Y=(-1)^{\int \nu_4^2}$. 
    \item For $Y=(-1)^{\int w_2w_3}$, we can write $T=(-1)^{\int BC+Bw_3+Cw_2}=(-1)^{\int BC+\mathrm{Sq}^2(C)+\mathrm{Sq}^2\mathrm{Sq}^1(B)}$ on closed orientable manifolds. This follows from $\int\mathrm{Sq}^2(C)=\int w_2C$ and on orientable 5-manifolds, $\int \mathrm{Sq}^2\mathrm{Sq}^1(B)=\int w_2 \mathrm{Sq}^1(B)=\int \left(\mathrm{Sq}^1(w_2B)+\mathrm{Sq}^1(w_2)B\right)=\int w_3B$ (see e.g. Refs.~\cite{Gukov:2020btk,Chen:2021xks,Fidkowski:2024hpz}).
    \item For $Y=(-1)^{\int w_2^4}$, which is conjectured to be a nontrivial QCA, the naive way to write $Y$ as a Gauss sum is as $Y=\sum_{A,B}(-1)^{\int AB+Aw_2+Bw_2^3}$ for $A$ a $6$-form gauge field and $B$ and $2$-form gauge field (we can also write it similarly in terms of two $4$-form gauge fields by splitting $w_2^4$ into two factors of $w_2^2$). Another way to write $Y$ as a Gauss sum is as $\sum_{A,B}(-1)^{\int AB+\mathrm{Sq}^2A+B^4}$ because the sum over $A$ sets $B=w_2$. This Gauss sum uses $T$ written entirely in terms of background gauge fields, defined by a quadratic form with the same polarization as $AB+Aw_2+Bw_2^3$, so forms a projective $\mathrm{SL}(2,\mathbb{Z}_2)$ with the same $S$ as for $AB+Aw_2+Bw_2^3$.
\end{itemize}
However, a general Stiefel-Whitney counterterm $Y$ may not be of the form $\sum_{\{b_i\}}(-1)^{\int q(\{b_i\})}$ for any quadratic form $q(\{b_i\})$ involving symmetry gauge fields. We will now show that for all gravitational terms $Y$, one can in principle obtain the corresponding action concretely in a lattice model (up to a FDQC) using an $S\tilde{T}$ operation where $S$ is gauging an abelian finite higher form symmetry (implemented by a FDQC, measurement, and error correction) and $\tilde{T}$ is a more general FDQC, not necessarily the exponential of a quadratic form. For $\tilde{T}$ that is not the exponential of a quadratic form, it may not participate in a $\mathrm{SL}(2,\mathbb{Z}_2)$ relation, but it can still be implemented by a FDQC.

Refs.~\cite{Chen:2021xks,Fidkowski:2024hpz} obtained for a general action of the form $(-1)^{\int w_i\cdots w_j}$ a formulation in terms of symmetry gauge fields. To do so, one starts with
\begin{equation}
    \int S_i(a_{d+1-i})=\int a_{d+1-i}\alpha_i~,
\end{equation}
where $S_i(a_{d+1-i})$ is iterated Steenrod operation (implemented on the lattice via nested higher cup products such as in Refs.~\cite{Chen:2021ppt,Chen:2021xks}; see (\eqref{eqn:nestedhughercup} for an example)) on the $(d+1-i)$-form gauge field $a_{d+1-i}$ and
\begin{equation}
    \alpha_i=w_i+\mathrm{poly}(\{w_1,w_2,\cdots w_{i-1}\})~,
\end{equation}
where $\mathrm{poly}(\{w_1,w_2,\cdots w_{i-1}\})$ is a polynomial only involving lower Stiefel Whitney classes. This is a unitriangular change of generators, so it can be inverted recursively: $w_i=\alpha_i+\mathrm{poly}(\{\alpha_1,\alpha_2,\cdots \alpha_{i-1}\})$ for another polynomial. This means that any product of $\{w_i\}$ can be written as a polynomial $f$ (possibly with several terms) of $\{\alpha_i\}$. Writing
\begin{equation}
    Y=(-1)^{\int w_i\cdots w_j}=(-1)^{\int f(\{\alpha_i\})}~.
\end{equation}
we get
\begin{equation}\label{eqn:QCAoriginal}
    Y=(-1)^{\int\prod w_i}=\sum_{a_i,b_i}(-1)^{\int (\sum a_{d+1-i}b_i+S_i(a_{d+1-i}))+f(\{b_i\})}~.
\end{equation}

The sum over $a_{d+1-i}$ sets $b_i=\alpha_i$ resulting in the desired gravitational action $\int f(\{\alpha_i\})$. Note that the above formulation of $Y$ is of the form $\epsilon_0S\tilde{T}$ where, again, $\epsilon_0$ denotes setting background fields to zero. Therefore, (\ref{eqn:QCAoriginal}) is not quite of the form $S\tilde{T}$, where 
\begin{equation}
    \tilde{T}[1]=(-1)^{\int (\sum a_{d+1-i}b_i+S_i(a_{d+1-i}))+f(\{b_i\})}
\end{equation}
is a generally non-quadratic invertible phase. To connect $Y$ and $S\tilde{T}$, we will argue that $Y$ can be obtained from $S\tilde{T}$ by a FDQC, so $S\tilde{T}$ lies in the same QCA class as $Y$. 

We can construct the stabilizers for the invertible phases as follows. Denote
\begin{equation}
    \omega(\{a_i\},\{b_j\})=\sum a_{d+1-i}b_i+S_i(a_{d+1-i})+f(b_i)~.
\end{equation}
From the above cocycle, one can obtain
\begin{align}
    &\omega(\{a_i+d\tilde s_{i-1},a_{j\neq i}\},\{b_k\})-\omega(\{a_j\},\{b_k\}):=d\phi_i(s_{i-1},\{a_i,a_{j\neq i}\},\{b_k\})\cr 
  &\omega(\{a_k\},\{b_i+d\tilde s_{i-1},b_{j\neq i}\})-\omega(\{a_{k}\},
  \{b_{j}\}):=d\eta_i(s_{i-1},\{a_k\},\{b_i,b_{j\neq i}\})~,
\end{align}
where $\tilde s_i$ are the $i$-cochain that takes value 1 on $i$-simplex $s_i$ and 0 on other $i$-simplices.
Denote $B_i=du_{i-1}$ and $A_i=dc_{i-1}$. The stabilizers for the theory with coupling to $A,B$ are
\begin{align}
    &{\cal S}_{s_{i-1}}^{a_i}=\left(\prod_{s_{i-1}\in s_i}X^{a_i}_{s_{i}}\right)(-1)^{\int \phi_i(s_{i-1},\{a_j\},\{b_k\})}(-1)^{\int d\tilde s_{i-1}\cup u_{d+1-i}},\cr 
    &{\cal S}_{s_{i-1}}^{b_i}=\left(\prod_{s_{i-1}\in s_i}X^{a_i}_{s_{i}}\right)(-1)^{\int \eta_i(s_{i-1},\{a_j\},\{b_k\})}(-1)^{\int d\tilde s_{i-1}\cup c_{d+1-i}},\cr
    &{\cal S}^{u_{i}}_{s_{i}}=X^{u_i}_{s_i}(-1)^{\int a_{d-i}\cup \tilde s_i},\quad {\cal S}^{c_i}_{s_i}=X^{c_i}_{s_i}(-1)^{\int b_{d-i}\cup \tilde s_i}~,
\end{align}
as well as the flux $Z$-stabilizers for $a,b$.
Thus the theory that couples to $A,B$ has additional entangler
\begin{equation}
    U=(-1)^{\int \sum_i\left(a_{d-i}\cup u_i+b_{d-i}\cup c_i\right)}~.
\end{equation}
This is a finite depth quantum circuit that breaks the symmetries.
The circuit satisfies
\begin{align}
    &U{\cal S}_{s_{i-1}}^{a_i} U^{-1}=\left(\prod_{s_{i-1}\in s_i}X^{a_i}_{s_{i}}\right)(-1)^{\int \phi_i(s_{i-1},\{a_j\},\{b_k\})}\cr
    &U {\cal S}_{s_{i-1}}^{b_i} U^{-1}=\left(\prod_{s_{i-1}\in s_i}X^{a_i}_{s_{i}}\right)(-1)^{\int \eta_i(s_{i-1},\{a_j\},\{b_k\})}\cr 
    &U{\cal S}^{u_{i}}_{s_{i}} U^{-1}=X^{u_i}_{s_i},\quad U{\cal S}^{c_{i}}_{s_{i}} U^{-1}=X^{c_i}_{s_i}~,
\end{align}
while leaving the flux $Z$-stabilizers invariant.

\subsubsection{Example: $(-1)^{\int w_2^4}$ QCA}

\paragraph{Non-Clifford Presentation of stabilizers}
Let's consider the following model for the invertible topological field theory for the QCA:
\begin{equation}
        (-1)^{\int w_2^4}=\sum_{u_6,a_2}(-1)^{\int \left(Sq^2 u_6+a_2\cup u_6+a_2^4\right)}~,
\end{equation}
where $Sq^2 u_6=w_2\cup u_6$, and $u_6$ is a Lagrangian multiplier that enforces $a_2=w_2$ \cite{Chen:2021xks}. We can also express $Sq^2 u_6=u_6\cup_4 u_6$.
The lattice model can be constructed as follows. First, we construct the corresponding SPT model whose response is given by the same action but replacing the dynamical gauge fields by background gauge field. The Hamiltonian for such SPT can be constructed using the method in Ref.~\cite{Chen:2011pg}. Then, we gauge the symmetry in the lattice model \cite{Levin_2012} to obtain a lattice model for the gauge theory.
Denote
\begin{equation}
    \omega(u,a,b,c)=\pi \left(Sq^2 u_6 +a_2\cup u_6+ a_2^4\right)~.
\end{equation}
We have
\begin{align}
    &\omega(u=dx_u,a=dx_a)=d\phi\cr 
    &\phi=\pi\left(x_u\cup_3 x_u+x_u\cup_4 dx_u+x_a\cup dx_u+x_a\cup dx_a\cup dx_a\cup dx_a\right)~.
\end{align}
The entangler of the SPT is $e^{i\int \phi}$. After gauging the symmetry, the Hamiltonian is a sum of the following non-Pauli stabilizers:
(denote the Pauli operators for $u,a$ by $X^u,X^a$, etc)
\begin{align}
    &{\cal S}_{s_5}^u=\left(\prod_{s_5\in s_6} X^u_{s_6}\right) (-1)^{\int u\cup_4 \tilde s_5+\tilde s_5\cup_4 u+a\cup \tilde s_5},\cr 
    &{\cal S}^a_{s_1}=\left(\prod_{s_1\in s_2} X^a_{s_2}\right)
        (-1)^{\int \left(\tilde s_1\cup u+ \tilde s_1\cup a^3+a\cup \tilde s_1 \cup a^2+a^2\cup \tilde s_1\cup a+a^3\cup \tilde s_1\right)},\cr
    &{\cal S}'^u_{s_7}=\prod_{s_6\in s_7} Z^u_{s_6},\;\;
    {\cal S}'^a_{s_3}=\prod_{s_2\in s_3} Z^a_{s_2}~.
\end{align}
Using the formula $(-1)^{\int \left(u\cup_4 \tilde s_5+\tilde s_5\cup_4 u\right)}=(-1)^{\int \left(u\cup_5 d\tilde s_5+du\cup_5 \tilde s_5\right)}=(-1)^{\int u\cup_5 d\tilde s_5}\prod_{s_7} \left({\cal S}'^u_{s_7}\right)^{\int \tilde s_7\cup_5 \tilde s_5}$, We can replace ${\cal S}^u_{s_5}$ by
\begin{equation}
    {\cal S}^u_{s_5}=\left(\prod_{s_5\in s_6} X^u_{s_6}\right) (-1)^{\int u\cup_5 d\tilde s_5}~.
\end{equation}
Likewise, using the formula 
\begin{equation}
 (-1)^{\int \left(\tilde s_1\cup a^3+a\cup \tilde s_1 \cup a^2+a^2\cup \tilde s_1\cup a+a^3\cup \tilde s_1\right)}=(-1)^{\int\left((d\tilde s_1\cup_1 a+\tilde s_1\cup_1 da)\cup a^2+a^2\cup (d\tilde s_1\cup_1 a+\tilde s_1\cup_1 da)\right)}~,   
\end{equation}
and $da$ is even by the ${\cal S}'^a_{s_3}$ stabilizers, we can replace ${\cal S}^a_{s_1}$ by
\begin{equation}
    {\cal S}^a_{s_1}=\left(\prod_{s_1\in s_2} X^a_{s_2}\right)
        (-1)^{\int \left(\tilde s_1\cup u+ (d\tilde s_1\cup_1 a )\cup a^2+a^2\cup (d\tilde s_1 \cup_1 a)\right)}~.
\end{equation}
The above stabilizers are Pauli operators except for ${\cal S}_{s_1}^a$, which is a non-Clifford operator due to the physical CCZ gates for the qubits for $a$ from the factor $(-1)^{\int \left( (d\tilde s_1\cup_1 a )\cup a^2+a^2\cup (d\tilde s_1 \cup_1 a)\right)}$~.

Let's consider 7-dimensional hypercubic lattice, where for each $n$-cell there is a corresponding $(7-n)$-cell related by cup product. We can rewrite the stabilizers ${\cal S}^a,{\cal S}^u$ as follows:
\begin{align}
    &\overline{Z}^u_{s_6}:=Z^u_{s_6}\prod_{s_1} (G_{s_1}^a)^{\int \tilde s_1\cup \tilde s_6},\quad G_{s_1}^a=\left(\prod_{s_1\in s_2} X^a_{s_2}\right)(-1)^{\int\left( (d\tilde s_1\cup_1 a )\cup a^2+a^2\cup (d\tilde s_1 \cup_1 a)\right) }\cr 
    &\overline{Z}^a_{s_2}=Z^a_{s_2} \prod_{s_5} (G^u_{s_5})^{\int \tilde s_2\cup \tilde s_5},\quad G^u_{s_5}=\left(\prod_{s_5\in s_6} X^u_{s_6}\right)(-1)^{\int u\cup_5 d\tilde s_5}~.
\end{align}
Note that the stabilizers $\overline{Z}^u,\overline{Z}^a$ together already determine a unique ground state. Thus we can take them as the separators of the QCA. From the stabilizers one can obtain the flippers. For example,
\begin{equation}
    \overline{X}^u_{s_6}=X^u_{s_6}(-1)^{\int \tilde s_6\cup_5 u}~.
\end{equation}
The QCA maps $Z^u_{s_6},Z^a_{s_2}$ to $\overline{Z}^u_{s_6},\overline{Z}^a_{s_2}$ respectively, and similarly action on the flippers.

\paragraph{Clifford Presentation of stabilizers}

While the above presentation of the QCA using non-Clifford operators, there is in fact another presentation of the same QCA using Clifford operators. Thus these two presentations are related by a non-Clifford FDQC.\footnote{
In the Clifford presentation we introduce two extra qubits on each 4-cell. The previous Non-Clifford operator that represents the QCA acts trivially on these additional qubits. 
}

Let's consider the following higher-form gauge theory model using four dynamical $\mathbb{Z}_2$ gauge fields $u_6,a_2,b_4,c_4$ that realizes $(-1)^{\int w_2^4}$ QCA:
\begin{equation}
    (-1)^{\int w_2^4}=\sum_{u_6,a_2,b_4,c_4}(-1)^{\int \left(Sq^2 u_6 +a_2\cup u_6+b_4\cup c+b_4\cup a_2^2+c_4\cup a_2^2\right)}~,
\end{equation}
where $Sq^2 u_6=w_2\cup u_6$, and $u_6$ is a Lagrangian multiplier that enforces $a_2=w_2$ \cite{Chen:2021xks}. We can also express $Sq^2 u_6=u_6\cup_4 u_6$.
Denote
\begin{equation}
    \omega(u,a,b,c)=\pi \left(Sq^2 u_6 +a_2\cup u_6+b_4\cup c_4+b_4\cup a_2^2+c_4\cup a_2^2\right)~.
\end{equation}
We have
\begin{align}
    &\omega(u=dx_u,a=dx_a,b=dx_b,c=dx_c)=d\phi\cr 
    &\phi=\pi\left(x_u\cup_3 x_u+x_u\cup_4 dx_u+x_a\cup dx_u+x_b\cup dx_c+x_b\cup dx_a\cup dx_a+x_c\cup dx_a\cup dx_a\right)~.
\end{align}
The corresponding stabilizers are (denote the Pauli operators for $u,b,c,a$ by $X^u,X^b,X^c,X^a$, etc)
\begin{align}
    &{\cal S}_{s_5}^u=\left(\prod_{s_5\in s_6} X^u_{s_6}\right) (-1)^{\int \left(u\cup_5 d\tilde s_5+a\cup \tilde s_5\right)},\quad 
    {\cal S}^b_{s_3}=\left(\prod_{s_3\in s_4} X^b_{s_4}\right)(-1)^{\int \tilde s_3\cup \left(c+a\cup a\right)},\cr 
    &{\cal S}^c_{s_3}=\left(\prod_{s_3\in s_4} X^c_{s_4}\right)(-1)^{\int \left(b\cup \tilde s_3+\tilde s_3\cup a\cup a\right)},\cr
    &{\cal S}^a_{s_1}=\left(\prod_{s_1\in s_2} X^a_{s_2}\right)(-1)^{\int \left(\tilde s_1\cup u+ b\cup \left(d\tilde s_1\cup_1 a\right)+c\cup \left(d\tilde s_1\cup_1 a\right)\right)},\cr
    &{\cal S}'^u_{s_7}=\prod_{s_6\in s_7} Z^u_{s_6},\;\;
    {\cal S}'^b_{s_5}=\prod_{s_4\in s_5} Z^b_{s_4},\;\;
    {\cal S}'^c_{s_5}=\prod_{s_4\in s_5} Z^c_{s_4}.\;\;
    {\cal S}'^a_{s_3}=\prod_{s_2\in s_3} Z^a_{s_2}~.
\end{align}
Let's consider 7-dimensional hypercubic lattice, where for each $n$-cell there is a corresponding $(7-n)$-cell related by cup product.  We can rewrite the stabilizers ${\cal S}^a,{\cal S}^c,{\cal S}^b,{\cal S}^u$ respectively as follows:
\begin{align}
    &\overline{Z}^u_{s_6}:=Z^u_{s_6}\prod_{s_1} (G_{s_1}^a)^{\int \tilde s_1\cup \tilde s_6},\quad G_{s_1}^a=\left(\prod_{s_1\in s_2} X^a_{s_2}\right)(-1)^{\int (b+c)\cup (d\tilde s_1\cup_1 a)}\cr 
    &\overline{Z}^b_{s_4}=Z_{s_4}^b\prod_{s_3} (G_{s_3}^c)^{\int \tilde s_4\cup\tilde s_3},\quad G^c_{s_3}=\left(\prod_{s_3\in s_4}X^c_{s_4}\right)(-1)^{\int \tilde s_3\cup a^2}\cr 
    &\overline{Z}^c_{s_4}=Z_{s_4}^c\prod_{s_3} (G_{s_3}^b)^{\int \tilde s_4\cup\tilde s_3},\quad G^b_{s_3}=\left(\prod_{s_3\in s_4}X^b_{s_4}\right)(-1)^{\int \tilde s_3\cup a^2}\cr 
    &\overline{Z}^a_{s_2}=Z^a_{s_2} \prod_{s_5} (G^u_{s_5})^{\int \tilde s_2\cup \tilde s_5},\quad G^u_{s_5}=\left(\prod_{s_5\in s_6} X^u_{s_6}\right)(-1)^{\int u\cup_5d\tilde s_5}~.
\end{align}
Note that the stabilizers $\overline{Z}^u,\overline{Z}^b,\overline{Z}^c,\overline{Z}^a$ together already determine a unique ground state. Thus we can take them as the separators of the QCA. From the stabilizers one can obtain the flippers. The QCA maps $Z^u_{s_6},Z^a_{s_2},Z^b_{s_4},Z^c_{s_4}$ to $\overline{Z}^u_{s_6},\overline{Z}^a_{s_2},\overline{Z}^b_{s_4},\overline{Z}^c_{s_4}$ respectively, and similar action on the flippers.
We remark that the QCA is non-Clifford since it maps Pauli $Z^u_{s_6}$ to Clifford operator $\overline{Z}^u_{s_6}$, etc.

\section{QCA and Generalized Time-Reversal Symmetry}
\label{sec:non-invT}

In this section, we will show that $YT^{-1}[1]$ and $T^{-1}[1]$ correspond to different phases if we include generalized time-reversal symmetries. We begin by reviewing this in 3+1d in the case where $YT^{-1}[1]\sim Y_{\rho_2}$ describes a Brown-Kervaire invariant (i.e. $Y$ is a gravitational Pontryagin term); see an extended discussion in Ref.~\cite{Hsin:2021qiy}. In this case, $YT^{-1}[1]$ has a 1-form symmetry that mixes with time reversal ${\cal T}$ in a $2$-group. We will show in this section that for more general $YT^{-1}$ with $\mathbb{Z}_{N\neq 2}$ $1$-form symmetries in 3+1d, $YT^{-1}$ differs from $T^{-1}$ in its non-invertible time-reversal symmetry. In continuum field theory, where some $\mathbb{Z}_N$ $1$-form SPTs are said to be invariant under $S$ \cite{Apte:2022xtu,Cordova:2023bja}, mixing with nontrivial $Y$ therefore precludes true invariance in the presence of these other spacetime symmetries.

From a lattice perspective, these time reversal symmetries in 3+1d distinguish $1$-form SPTs entangled by QCA \cite{Haah:2018jdf,Haah:2019fqd,shirleyqca}, which have boundary chiral anyon theories, and very similar ones entangled by FDQC \cite{PhysRevB.101.035101}. In principle, forgetting about the time reversal symmetry would allow these SPTs to be local unitary related to each other.\footnote{By this we mean the ground states are local unitary related to each other. It is presumably not possible to extend this unitary locally to the full Hilbert space due to arguments for the nontriviality of the QCA entangling chiral anyon Walker-Wang models.} In this section, we will always be acting on the trivial state $Z=1$, so will drop the argument of $Y,T,S$ and use i.e. $T$ interchangeably with $T[1]$. 
The time-reversal symmetry will be denoted by ${\cal T}$.

\subsection{Invertible phases protected by higher group time-reversal symmetry}
We begin by first reviewing how $Y_{\rho_2}=YT^{-1}$ with $Y=e^{\frac{i\pi}{12}\int p_1}$ and $T^{-1}=e^{-\frac{i\pi}{2}\int \mathcal{P}(b)}$ can be differentiated from $T^{-1}$ via higher group time reversal symmetry. The former is an SPT with boundary chiral anti-semion anyon theory and is entangled by a nontrivial QCA, while latter has a boundary doubled semion theory and is entangled by a FDQC \cite{PhysRevB.101.035101}. We will then generalize these statements to higher $4k$ spacetime dimensions.

\subsubsection{Brown-Kervaire invariant in 3+1d: invertible phase protected by time-reversal symmetry}
Consider the invertible phase $Y_{\rho_2}$ in 3+1d with partition function 
\begin{equation}\label{chiralsemion}
    Y_{\rho_2}=\sum_b e^{\frac{\pi i}{2}\int q_\rho(b)}=F_2[1]~.
\end{equation}
This state is described by the chiral anti-semion Walker-Wang model \cite{shirleyqca}. It is equivalent to the 1-form SPT $T^{-1}=e^{-\frac{i\pi}{2}\int \mathcal{P}(B)}$ together with an extra gravitational instanton term $e^{\frac{i\pi}{12}\int p_1}$. The gravitational instanton term allows for $Y_{\rho_2}$ to be time reversal invariant. Without time reversal symmetry, the two theories $Y_{\rho_2}$ and $T^{-1}$ are equivalent because in this case we can smoothly deform the coefficient of the gravitational Pontryagin term to zero. With time-reversal symmetry, we cannot do this, so the invertible phase $Y_{\rho_2}$ is distinct from the 1-form SPT $e^{-\frac{i\pi}{2}\int \mathcal{P}(\rho_2)}$.  

Even though the partition function (\ref{chiralsemion}) is not obviously time-reversal invariant, Ref.~\cite{Hsin:2021qiy} showed that it in fact is invariant under a higher group time reversal symmetry. Specifically, the time-reversal symmetry forms a higher group with the 1-form symmetry with background $\rho_2$. By definition, 
\begin{equation}\label{eqn:2grouprho}
    d\rho_2=\nu_3=w_1\cup w_2~,
\end{equation}
where $w_1$ is the background field for time-reversal symmetry, i.e. it is Poincare dual to the time-reversal domain wall.

On the boundary of the Walker-Wang model is the anti-semion theory $U(1)_{-2}$, which is also known to have the same higher group time-reversal symmetry \cite{Hsin:2019gvb}.

 \paragraph{Framed Gapped Domain Walls and Boundaries}

Let's investigate the domain wall generating the higher group time-reversal symmetry.
In the absence of 2-group structure, the domain wall that generates the time-reversal symmetry would flip $U(1)_{-2}$ to $U(1)_{2}$. Applying the folding trick, a gapped domain wall would indicate a gapped boundary of $U(1)_{-2}\times U(1)_{-2}$. However, such a theory in bosonic systems does not admit a gapped boundary, because $c_-\neq 0$ mod 8. Thus the presence of time-reversal symmetry in the bosonic semion theory $U(1)_2$ would seem to present a counterexample to the folding trick. 

The resolution to this paradox is that the time-reversal operation in the semion theory is a 2-group transformation, which consists of the usual time-reversal transformation that complex conjugate all correlation functions, followed by changing the background of one-form symmetry generated by anti-semion
\begin{equation}
    \rho_2\rightarrow \rho_2+w_2~.
\end{equation}
This means that the transition function of the tangent bundle is correlated with the one-form transformation in the region where time-reversal symmetry $\mathcal{T}$ is applied: while changing the spin structure on the support of the anti-semion does not do anything in the region without applying $\mathcal{T}$ symmetry, it will induces a one-form transformation in the region $\mathcal{T}$ symmetry is applied. Since the anti-semion is charged under the one-form symmetry, this means that changing the spin structure produces a minus sign, and thus effectively the anti-semion in the region where $\mathcal{T}$ symmetry is applied is attached to a local fermion--although the theory remains bosonic. Thus in the region where $\mathcal{T}$ symmetry is applied, the anti-semion becomes a semion attached to a fermion, which again becomes a spin $-1/4$ anti-semion.

As a consequence of the 2-group symmetry, the domain wall between anti-semion theory ${\cal C}$ and ${\cal T}({\cal C})$ becomes the gapped boundary of ${\cal C}\times \overline{{\cal T}({\cal C})}$, where $\overline{{\cal T}({\cal C})}$ is the complex conjugation of ${\cal T}({\cal C})$. As discussed above, this is effectively the theory of the semion-fermion composite, i.e. an anti-semion. Thus under the folding trick, the domain wall is a gapped boundary of the double semion theory--which indeed admits a gapped boundary. This example tells us that the map from gapped domain wall to gapped boundary can be many-to-one if the domain wall is not purely topological, such as when it is involved in a higher group with the spacetime Lorentz group.

The two-group symmetry also makes the invertible phase $Y_{\rho_2}$ order eight rather than order four. Clearly, $e^{-\frac{i\pi}{2}\int \mathcal{P}(b)}$ is an order four SPT. Four copies of $Y_{\rho_2}$ is given by
\begin{equation}
    \sum_{b_i}e^{\frac{i\pi}{2}\sum_{i=1}^4\int q_{\rho_2}(b_i)}~.
\end{equation}

As shown in Ref.~\cite{Hsin:2021qiy}, a change of variables together with the Wu formula results in the following action
\begin{equation}
  \sum_{b_i}e^{\frac{i\pi}{2}\sum_{i=1}^4\int q_{\rho_2}(b_i)}=e^{i\pi\int w_2^2+w_1^4}  ~.
\end{equation}

which describes the three fermion Walker Wang model, which has gapped boundary given by the three fermion anyon theory. This time-reversal SPT is order two, so $Y_{\rho_2}$ is order eight, consistent with the conjectured Witt group classification of QCAs \cite{shirleyqca,Haah:2019fqd}. 

We note that while the invertible phase and ground state requires time-reversal symmetry to be stable against perturbations, the corresponding QCA does not. While the ground states of two models differing by a gravitational topological response may be local unitary related to each other in the absence of symmetry, due to QCA obstructions there may not exist such an interpolation between their parent Hamiltonians. Specifically, if $H=U_{QCA} H_{X} U_{QCA}^\dagger$ where $U_{QCA}$ is a nontrivial QCA and $H_X$ is a trivial paramagnet Hamiltonian, additional conjugation by any FDQC $U_{FDQC}$ should not change the QCA class of $H$. The corresponding ground states may have emergent time-reversal symmetry $U_{FDQC}\cdot K$.

\subsubsection{Generalization to higher dimensions}
More generally, when the spacetime dimension is $4k$, the QCA from gauging $\mathbb{Z}_2$ $(2k-1)$-form symmetry corresponds to an invertible phase whose partition function depends on the trivialization of the Wu class $v_{2k+1}=w_1\cup f(\{w_i\})$ where $f$ is a polynomial of Stiefel-Whitney classes under cup product. Since the Wu class has $w_1$, this implies that the invertible phase is protected by time-reversal symmetry. Moreover, the time-reversal symmetry mixes with the $(2k-1)$-form $\mathbb{Z}_2$ symmetry to form a higher group.

For example, consider $8$ spacetime dimension, and the QCA from gauging $\mathbb{Z}_2$ 3-form symmetry. The QCA corresponds to the invertible phase whose partition function depends on $\nu_5=w_1\cup \left(w_1^2w_2+w_2^2+w_1w_3+w_4\right)$, given by
\begin{equation}
\sum_b    e^{\frac{\pi i}{2}\int q_{\rho_4}(b)}~,
\end{equation}
where $b$ is a 4-form $\mathbb{Z}_2$ gauge field, and $d\rho_4=\nu_5$. The theory admits boundary chiral 6+1d 3-form Wu Chern-Simons theory that generalizes the anti-semion theory (see e.g. Refs.~\cite{Witten:1998wy,Fiorenza:2012ec,Gukov:2020btk,Gukov:2025dol}) where the anyon-like excitations are membrane excitations created by volume operators.

\subsection{Invertible phases protected by non-invertible time-reversal symmetry}

Consider the QCA from gauging $\mathbb{Z}_N$ 1-form symmetry in 3+1d, which corresponds to the partition function (in the following, we will enrich the theory with 1-form symmetry)
\begin{equation}
\text{Even }N:\;\;    Y[B]=\sum_b e^{\frac{2\pi i}{2N}\int {\cal P}(b)+\frac{2\pi i}{N}\int b\cup B},\quad 
\text{Odd }N:\;\;    Y[B]=\sum_b e^{\frac{2\pi i}{N}\int {\cal P}(b)+\frac{2\pi i}{N}\int b\cup B}~.
\end{equation}
For simplicity, let's focus on even $N$ in the following discussion, while the case of odd $N$ is similar. Note that here we do not require the specific SPT classes discussed in Sec.~\ref{sec:modular}, because general 1-form SPTs beyond those that participate in modular relations may have gravitational counterterms occuring from gauging. We will show that these theories differ from $e^{-\frac{2\pi i}{2N}\mathcal{P}(b)}$ by non-invertible time-reversal symmetries. Our non-invertible time reversal symmetry is inspired by the invertible one above, where the usual time reversal was accompanied by shifting the semion by a fermion. 

Let's consider the boundary anyon theory ${\cal A}^{N,-1}$\cite{Hsin:2018vcg}. The theory has the following non-invertible time-reversal symmetry:
\begin{equation}
   {\cal T}{\cal A}^{N,-1}:=\frac{\left({\cal A}^{N,-1}\right)^\star\times {\cal C}}{\mathbb{Z}_N},\quad {\cal C}={\cal A}^{N,-1}\times {\cal A}^{N,-1}~.
\end{equation}
Note that the theory is dual to ${\cal A}^{N,-1}$ itself, and thus this is a non-invertible time-reversal symmetry. In a companion paper we will study the properties and applications of non-invertible time-reversal symmetries in 2+1d. These non-invertible time reversal symmetries are of a similar flavor as those in Ref.~\cite{Choi:2022rfe}, which also involve stacking with fractional quantum Hall states.

Back in 3+1d, the non-invertible time-reversal symmetry corresponds to
\begin{equation}
    {\cal T}Y[B]=\sum_b Y[b]^\star Y[b] Y[B]=Y[B]~.
\end{equation}
Thus the 3+1d invertible phase is protected by non-invertible time-reversal symmetry has the form of ${\cal T}=T\epsilon_0STK$ where here $T$ is stacking with the invertible state $Y$ (it includes the gravitational part, so it is not the exponential of a quadratic form).

\paragraph{Comparison with 1-form SPT}

The invertible phase has the 1-form symmetry response of that class $(-1)$ 1-form SPT \cite{Hsin:2018vcg,Tsui:2019ykk}. However, we will show that the 1-form SPT does not have such non-invertible time-reversal symmetry.
Denote the partition function of the 1-form SPT by $Z[B]=e^{-\frac{2\pi i }{2N}\int {\cal P}(B)}$,
\begin{equation}
    {\cal T}Z[B]=\sum_b Z[b]^\star Y[b] Y[B]=\sum_{b,b'}
    e^{\frac{2\pi i}{2N}\int {\cal P}(b)}
    e^{\frac{2\pi i}{2N}\int {\cal P}(b')+\frac{2\pi i}{N}\int b'\cup b}Y[B]=Y[0] \left(\sum_{b'} 1\right)Y[B]~,
\end{equation}
where in the last equality we redefine $\tilde b=b+b'$. This shows that the time-reversal operation on the 1-form SPT produces a non-invertible phase, and thus does not get back to the 1-form SPT itself. Therefore the 1-form SPT does not have the non-invertible time-reversal symmetry.

We note that similarly, the QCAs in  $(4k+3)$ spatial dimensions from gauging $\mathbb{Z}_N$ higher-form symmetry correspond to invertible phases protected by non-invertible time-reversal symmetry in addition to the $\mathbb{Z}_N$ $(2k+1)$-form symmetry. In particular, the non-invertible time-reversal symmetry distinguish such invertible phases from the $\mathbb{Z}_N$ $(2k+1)$-form SPT phases.

\section{Outlook}
\label{sec:outlook}

In this work we show how QCAs arise on the lattice from implementations of topological operations $S$ (gauging of finite invertible higher-form symmetries) and $T$ (stacking with invertible phases). QCAs come up in several contexts including projective modular relations, squares of dualities related to higher dimensional analogs of fermionization, and generalized time reversal symmetries. By using the relation to gauging, we show that broad classes of QCAs, including those from gravitational Pontryagin terms and Stiefel-Whitney terms, can be implemented concretely and efficiently in lattice models. Below, we highlight several interesting future directions.

\paragraph{Rigorous proofs for nontrivial QCAs} We use in this work Conjecture~\ref{conj:nontrivialqca} from Ref.~\cite{Fidkowski:2024hpz}. It would be interesting to see if this conjecture can be proven rigorously on the lattice, perhaps using incompatibility of certain topological responses to commuting projector realizations. In 3+1d, the standard argument for nontriviality of QCA is based on the incompatibility of Witt nontrivial 2+1d boundary topological orders to commuting projector realizations. By Witt nontrivial, we mean there is no gapped boundary to the vacuum \cite{Haah:2018jdf,Haah:2019fqd,Shirley:2022lhu}. It would be interesting if incompatibilities in higher dimensions can also be related to a generalized notion of Witt class.

\paragraph{Concrete implementation of gravitational Stiefel-Whitney terms in lattice models}
In this work we gave a schematic construction for lattice models whose ground states carry nontrivial gravitational Stiefel-Whitney terms, building off of Ref.~\cite{Chen:2021xks,Fidkowski:2024hpz}. We used steps involving gauging of finite invertible higher form symmetries and operations that we argue to be implementable by FDQCs. We leave to future work a more microscopic description of these FDQCs.

\paragraph{Relation between QCAs and invertible phases protected by symmetries}

While nontrivial invertible phases without symmetry give rise to nontrivial QCAs (e.g. the constructions in Refs.~\cite{Chen:2021xks,Fidkowski:2024mof}), the ground states entangled by nontrivial QCAs may be related to those entangled by circuits if we do not impose further generalized time-reversal symmetries \ref{sec:non-invT}. It will be interesting to explore what minimal symmetries are required in general for the QCAs to correspond to nontrivial invertible phases and what are general  underlying mechanisms.

In particular, in 3+1d there is a nontrivial QCA corresponds to the entangler of the Walker-Wang model for the $U(1)_2\times U(1)_{-4}$ anyon theory. However, the effective action of the Walker-Wang model does not have a gravitational term, unlike for the other cases described in this paper. It would be interesting to explore what symmetry can distinguish the model from 1-form SPT phases similar to the discussion in section \ref{sec:non-invT}.

\paragraph{QCAs from gauging subsystem symmetries}

While our discussion focus on topological symmetries, i.e. symmetries that have emergent topological properties, it would be interesting to explore the QCAs related to gauging more general subsystem symmetries. For example, if the system is a stack of lower-dimensional systems living on the leaves of foliations, it is nature to discuss QCAs related gauging a symmetry that only act on subsystems. Relatedly, there are Kramers-Wannier dualities and their generalizations for gauging subsystem symmetries \cite{Cao:2023doz,hsintoapppear:2026}.

\paragraph{Gravitational topological response and modular relation}

In our work, we show that a new family of gravitational topological responses constructed in terms of Stiefel-Whitney classes can arise as the projective phase in the modular relations for suitable $\mathbb{Z}_2$ higher-form symmetries. This poses the following question: more broadly, what classes of gravitational terms can arise in modular relations? It would be particularly interesting if there are genuinely different examples beyond the Stiefel-Whitney/Pontryagin/Brown-Kervaire examples presented here. In 3+1d, the Pontryagin class also shows up in the projective representation of the $S,T$ matrices describing the 2+1d boundary topological order \cite{Kitaev_2006}, due to the Gauss sum formula for the chiral central charge mod 8 in terms of anyon data. It would be interesting to study if other such bulk projective modular relations also show up in similar ways in their corresponding boundary theories.

\paragraph{QCAs and (non)-tensor product Hilbert spaces} In 1+1d, it has been proven rigorously which kinds of non-invertible symmetries must mix with translations and which kinds do not, on tensor product Hilbert spaces \cite{Wen:2026ncw}. In non-tensor product Hilbert spaces such as anyon-chain Hilbert spaces \cite{Feiguin:2006ydp,Buican:2017rxc,Aasen:2020jwb}, all fusion category symmetries can be produced on the nose, without mixing with translation operators. It would be interesting to study how the mixing of non-invertible operations like $S$ studied in this paper with QCA may get modified when one considers non-tensor product Hilbert spaces. There has been recent work on analogs of anyon chains in higher dimensions which describe higher dimensional systems with non-tensor product Hilbert spaces \cite{inamura2024fusion,Inamura:2025cum}.

\section*{Acknowledgments}

C.Z. is supported by the Harvard Society of Fellows and the Department of Physics at the University of British Columbia.
P.-S.H. is supported by Department of Mathematics, King's College London.
P.-S. H. and C.Z. were also supported in part by grant NSF PHY-2309135 to the Kavli Institute for Theoretical Physics (KITP). 
P.-S.H. and C.Z. thank Kavli Institute for Theoretical Physics for hosting the program ``Correlated Gapless Quantum Matter'' in 2024 and Perimeter Institute for hosting the conference ``Physics of Quantum Information'' in 2024, during which part of the work was completed. P.-S. H. thanks University of Oxford for hosting the workshop ``Topological Families'' in 2026, during which part of the work was completed. C.Z. and P.-S.H. thank Clay Cordova and Lukasz Fidkowski for helpful conversations, and acknowledge the use of ChatGPT 5.6 for calculations in Sec.~\ref{sqcapontryagin} and Appendix~\ref{sec:STST3}, and Claude Fable 5 for help with writing.

Near the completion of the work, we learned of two separate works Refs.~\cite{oxfordcoordination,natcoordination} that have partial overlaps with our work. We thank the authors for coordinating submissions.

 \appendix

\section{Mathematical Preliminaries on Quadratic Functions}\label{appendix:review}

\subsection{Properties of cup product and higher cup products}

In this appendix, we will summarize useful properties for cup and higher cup products on triangulated or hypercubic lattices. For more detail, see e.g.~Refs.~\cite{4fc46540-015c-3832-bc1c-ffd13cef2752,milnor1974characteristic,mosher2008cohomology,Hatcher} (for physics introduction, see e.g.~Refs.~\cite{Kapustin:2013qsa,Benini:2018reh,Tsui:2019ykk,Tata:2020qca,Chen:2021ppt,Chen:2021xuc}). In the following discussion, we will focus on general triangulated lattices, while the definitions of cup product and higher cup products on hypercubic lattices are give in e.g.~Refs.~\cite{Chen:2021ppt,Chen:2021xuc}.

The cup product for triangulated lattice is defined as follows. For $p$-cochain $f$ and $q$-cochain $g$, their cup product is a $(p+q)$-cochain whose value on $(p+q)$-simplex with vertices $(0,1,2\cdot,p+q)$ equals
\begin{equation}
    f\cup g(0,1,2\cdots,p+q)=f(0,1,\cdots,p)g(p,p+1,p+2,\cdots,p+q)~.
\end{equation}
The cup product is associative. In the discussion throughout the paper, we will often omit the cup product for simplicity of notation.
The higher cup products $\cup_{i}$ for $i\geq 1$ can be defined in similar way, and it produces a $(p+q-i)$-cochain from $p,q$-cochains.
For example, $\cup_1$ product of $f$ and $g$ gives a $(p+q-1)$-cochain whose value on $(p+q-1)$ simplex is
\begin{equation}
    f\cup_1 g (0,1,2\cdots,p+q-1)=\sum_{j=0}^{p-1}(-1)^{(p-j)(q+1)}f(0,1,\cdots,j,j+q,j+q+1,\cdots,p+q-1)g(j,j+1,\cdots,j+q)~.
\end{equation}
In the paper we will not omit the higher cup products $\cup_i$ for $i\geq 1$, since the higher cup products are no associative.

Below we summarize som useful identities for cup and higher cup products.
For $p$-cochain $f$ and $q$-cochain $g$, we have
\begin{align}
    &f\cup_i g=(-1)^{pq-i}g\cup_i f+(-1)^{p+q-i-1}\left(
    d\left(f\cup_{i+1}g\right)-df\cup_{i+1}g+(-1)^{p+1} f\cup_{i+1}dg\right)\cr 
    &d\left(f\cup_i g\right)=df\cup_i g+(-1)^p f\cup_i dg+(-1)^{p+q-i}f\cup_{i-1}g+(-1)^{pq+p+q}g\cup_{i-1}f~,
\end{align}
where $\cup_{0}=\cup$ is the usual cup product, and $\cup_j=0$ with $j<0$.

For $\mathbb{Z}_2$ cocycle $x_n$ of degree $n$, the Steenrod square operation $Sq^m$ on $x_n$ produces a $\mathbb{Z}_2$ $(m+n)$-cocycle given by
\begin{equation}
    Sq^m x_n=x_n\cup_{n-m} x_n=x_n\cup_{(m+n)-2m}x_n~,
\end{equation}
and $Sq^m x=0$ for $m>n$. The degree of the higher cup product equals to the total degree $(n+m)$ of the cocycle subtracting twice the superscript of $Sq$.
Moreover, on closed $D$-manifolds, $(-1)^{\int Sq^{D-n}x_n}=(-1)^{\int \nu_{D-n}\cup x_n}$ for $(D-n)$th Wu class $\nu_{D-n}$ of the tangent bundle.

Similarly, iterated Steenrod square operations can be related higher cup products in a nested way. For example,
\begin{equation}\label{eqn:nestedhughercup}
    Sq^{i_1} Sq^{i_2} x_n=\left( Sq^{i_2} x_n\right)\cup_{i_1+i_2+n-2i_1}\left( Sq^{i_2} x_n\right)=\left(x_n\cup_{i_2+n-2i_2} x_n\right)\cup_{i_1+i_2+n-2i_1}\left(x_n\cup_{i_2+n-2i_2} x_n\right)~.
\end{equation}

\subsection{Pontryagin square for higher even degree cocycles}

In this appendix we will review the Pontryagin square operation that acts on $\mathbb{Z}_N$ cocycles with even degree for even $N$ \cite{Whitehead196227}.

For $\mathbb{Z}_N$ cocycle for even $N$ and even degree $2n$, we can define the (generalized) Pontryagin square operation
\begin{equation}
    {\cal P}(B_{2n})=B_{2n}\cup B_{2n}-B_{2n}\cup_1 dB_{2n}~.
\end{equation}
The object ${\cal P}(B_{2n})$ has the following properties.

\begin{theorem}
    The generalized Pontryagin square ${\cal P}(B_{2n})$ for $2n$-cocycle $B_{2n}\in H^{2n}(M,\mathbb{Z}_N)$ satisfies
\begin{itemize}
\item[1. ] ${\cal P}(B_{2n})$ is a $\mathbb{Z}_{2N}$ cocycle, i.e. $d{\cal P}(B_{2n})=0$ mod $2N$.
    \item[2. ] The cohomology class of ${\cal P}(B_{2n})$ does not depend on the choice of lift of $B_{2n}$ as $\mathbb{Z}_{2N}$ cochain.
\item[3. ] The cohomology class of ${\cal P}(B_{2n})$ only depends on the cohomology class of $B_{2n}$:
\begin{equation}
 {\cal P}(B_{2n}+d\lambda_{2n-1})={\cal P}(B_{2n})+dz\text{ mod }2N~,   
\end{equation}
for integer cochain $z$ given by
\begin{equation}
    z=B_{2n}\cup \lambda_{2n-1}+\lambda_{2n-1}\cup d\lambda_{2n-1}+B_{2n}\cup_1 d\lambda_{2n-1}+dB_{2n}\cup_2 d\lambda_{2n-1}~.
\end{equation}

\item[4. ] For $\mathbb{Z}_N$ cocycles $B_{2n},B'_{2n}$ with the same degree $2n$, 
\begin{equation}
    {\cal P}(B_{2n}+B_{2n'})-{\cal P}(B_{2n})-{\cal P}(B_{2n}')=2B_{2n}\cup B_{2n}'+dz'\text{ mod }2N~,
\end{equation}
where $z'$ is an integer cochain given by
\begin{equation}
    z'=B_{2n}\cup_1 B_{2n}'+dB_{2n}\cup_2 B_{2n}'~.
\end{equation}
\end{itemize}
Therefore the generalized Pontryagin square ${\cal P}$ is a cohomology operation ${\cal P}:\; H^{2n}(M,\mathbb{Z}_N)\rightarrow H^{4n}(M,\mathbb{Z}_{2N})$ on manifold $M$.
\end{theorem}

\begin{proof}
Each statement in the theorem is proven as follows:
    \begin{itemize}
        \item[1. ] Leibnitz rule gives
\begin{align}
    &d{\cal P}(B_{2n]})=dB_{2n}\cup B_{2n}+B_{2n}\cup dB_{2n}-dB_{2n}\cup_1 dB_{2n}-
    B_{2n}\cup dB_{2n}\cr 
    &\qquad\qquad\quad +dB_{2n}\cup B_{2n}=2dB_{2n}\cup B_{2n}-dB_{2n}\cup_1 dB_{2n}=0\text{ mod }2N~,
\end{align}
where we used $dB_{2n}=0$ mod $N$.

\item[2. ] For $B_{2n}\rightarrow B_{2n}+N C_{2n}$ with integer cochain $C_{2n}$, ${\cal P}(B_{2n})$ changes by
\begin{align}
    &{\cal P}(B_{2n}+NC_{2n})-{\cal P}(B_{2n})\cr 
    &=N^2\left(C_{2n}\cup C_{2n}-C_{2n}\cup_1 dC_{2n}\right)
    +N\left(C_{2n}\cup B_{2n}+B_{2n}\cup C_{2n}-C_{2n}\cup_1 dB_{2n}-B_{2n}\cup_1 dC_{2n}\right)\cr 
    &=Nd\left(-B_{2n}\cup_1 C_{2n}+dB_{2n}\cup_2 C_{2n} \right)
    + NdB_{2n}\cup_2 dC_{2n}=dc\text{ mod }2N~,
\end{align}
where $c=N\left(-B_{2n}\cup_1 C_{2n}+dB_{2n}\cup_2 C_{2n} \right)$ and we used $dB_{2n}=0$ mod $N$ is even.

    \item[3. ] A direct computation gives
    \begin{align}
        &{\cal P}(B_{2n}+d\lambda_{2n-1})-{\cal P}(B_{2n})\cr 
        &=d \lambda_{2n-1}\cup d\lambda_{2n-1}+2B_{2n}\cup d\lambda_{2n-1}+d\left(
        B_{2n}\cup d\lambda_{2n-1}+dB_{2n}\cup_2 d\lambda_{2n-1}
        \right)\cr
        &=dz+2dB_{2n}\cup\lambda_{2n-1}=dz\text{ mod }2N~,
    \end{align}
    where we used $dB_{2n}=0$ mod $N$ and
    \begin{equation}
        z=B_{2n}\cup \lambda_{2n-1}+\lambda_{2n-1}\cup d\lambda_{2n-1}+B_{2n}\cup_1 d\lambda_{2n-1}+dB_{2n}\cup_2 d\lambda_{2n-1}~.
    \end{equation}

    \item[4.] A direct computation gives
\begin{align}
    &{\cal P}(B_{2n}+B_{2n}')-{\cal P}(B_{2n})-{\cal P}(B_{2n}')\cr 
    &=B_{2n}\cup B_{2n}'+B_{2n}'\cup B_{2n}-B_{2n}\cup_1 dB_{2n}'-B_{2n}'\cup_1 dB_{2n}\cr 
    &=2B_{2n}\cup B_{2n}'+ dz'-2B_{2n}\cup_1 dB_{2n}'+dB_{2n}\cup_2 dB_{2n}'=dz'\text{ mod }2N~,
\end{align}
where we used $dB_{2n}=dB_{2n}'=0$ mod $N$ are even, and
\begin{equation}
    z'=B_{2n}\cup_1 B_{2n}'+dB_{2n}\cup_2 B_{2n}'~.
\end{equation}

    \end{itemize}
\end{proof}

\subsection{Property of Arf-Brown-Kervaire Invariants}

In this appendix we summarize some mathematical properties of quadratic function which refines the intersection pairing on $H^n(M,\mathbb{Z}_2)$ on $2n$-dimensional manifold $M$ \cite{Browder:1969,Brown:1972}. See also \cite{Hsin:2021qiy}.

Denote $B$ to be $\mathbb{Z}_2$ value $n$-form gauge field. The quadratic function $q_\rho$ satisfies
\begin{equation}\label{eqn:quadref}
     q_\rho(B+B')=q_\rho(B)+q_\rho(B')+2B\cup B'\quad\text{mod }4~.
\end{equation}
In particular, $q_\rho(2B)=q(0)=0=2q_\rho(B)+ 2B^2$ mod 4 {\it i.e.} $q_\rho(B)$ mod 2 $=B^2$.

Two solutions for the quadratic function $q_\rho$ differ by a linear function $B\rightarrow 2\int_M B\cup X$ for some $X\in H^n(M,\mathbb{Z}_2)$, so different quadratic functions on a given manifold $M$ are in one-to-one correspondence with elements in $H^n(M,\mathbb{Z}_2)$, but not in a canonical way \cite{Brown:1972}.

The solutions of the quadratic function are labeled by different $\rho_n$ that satisfies $d\rho_n=v_{n+1}$ where $v_{n+1}$ is the $(n+1)^{\text{th}}$ Wu class, where different $\rho_n$'s form an $H^n(M,\mathbb{Z}_2)$ torsor \cite{Browder:1969,Brown:1972} 
We will denote the quadratic function by $q_{\rho}$. For manifold with dimension less or equal to $2n+1$, the Wu class $v_{n+1}$ is exact. Therefore, there always exists such $\rho_n$ that $d\rho_n=v_{n+1}$ \cite{milnor1974characteristic}. 
Changing the Wu structure $\rho_{n}\rightarrow \rho_{n}+z_{n}$ by $z\in H^{n}(M,\mathbb{Z}_2)$ changes the quadratic function by
\begin{equation}
    q_{\rho_n+z_n}(b)=q_{\rho_n}(b)+2b\cup z_n~.
\end{equation}

The Arf-Brown-Kervaire invariant is the partition function
\begin{equation}
    Z[\rho_n]=\frac{1}{|H^n(M,\mathbb{Z}_2)|^{1/2}}\sum_{b\in H^n(M,\mathbb{Z}_2)} e^{{\pi i\over 2}\int q_{\rho_n}(b)}~.
\end{equation}
The partition function is an eighth root of unity.
It depends on the Wu structure $\rho_n$. Change the Wu structure $\rho_{n}\rightarrow \rho_{n}+z_{n}$ leads to the following change of the partition function:
\begin{equation}\label{eq:Zrhoz}
    Z[\rho_n]/Z[\rho_n+z]=e^{{\pi i\over 2}\int q_{\rho_n}(z)}~.
\end{equation}

\section{Mapping $ST[1]$ to $(ST)^3[1]$}
\label{sec:STST3}

\subsection{Explicit flippers}

\subsubsection{Flippers for $(ST)$}

We list the Flippers column by column as follows.

\paragraph{First column}
 
\begin{align*}
(F_{1})_{X_x} ={}& xyz + 2x + 2yz + 2z \\
&+ 1 + \frac{2z}{x} + \frac{2z}{x^{2}}
\\[1ex]
(F_{1})_{X_y} ={}& 2x^{2}yz + xyz + x + \frac{x}{y} \\
&+ 1 + \frac{z}{y} + \frac{1}{y} + \frac{z}{x} \\
&+ \frac{2z}{xy} + \frac{1}{xy}
\\[1ex]
(F_{1})_{X_z} ={}& x^{2}y + xyz + 2xy + 2x \\
&+ \frac{2x}{z} + yz + 2z + 1 \\
&+ \frac{2}{z} + \frac{z}{x}
\\[1ex]
(F_{1})_{Z_x} ={}& 2xyz + x + z + 2 \\
&+ \frac{2yz}{x} + \frac{2z}{x} + \frac{1}{x}
\\[1ex]
(F_{1})_{Z_y} ={}& x^{2} + 2x + z + \frac{2z}{x}
\\[1ex]
(F_{1})_{Z_z} ={}& x^{2}y + 2xy + 2x + \frac{1}{x}
\end{align*}

\paragraph{Second column}
 
\begin{align*}
(F_{2})_{X_x} ={}& xyz + xz + 2x + \frac{x}{y} \\
&+ 2z + 2 + \frac{z}{y} + \frac{z}{x} \\
&+ \frac{2}{x} + \frac{2z}{xy}
\\[1ex]
(F_{2})_{X_y} ={}& 2x^{2}yz + x^{2}z + 2xz + x \\
&+ \frac{xz}{y} + \frac{2x}{y^{2}} + 2z + 1 \\
&+ \frac{2z}{y} + \frac{1}{y} + \frac{2z}{y^{2}} + \frac{1}{y^{2}}
\\[1ex]
(F_{2})_{X_z} ={}& x^{2}yz + x^{2}y + 2x^{2} + 2xyz \\
&+ xz + x + \frac{2x}{z} + \frac{x}{y} \\
&+ \frac{x}{yz} + 2z + \frac{2}{z} + \frac{z}{y} \\
&+ \frac{2}{yz} + \frac{2}{xyz}
\\[1ex]
(F_{2})_{Z_x} ={}& xyz + 2xz + x + \frac{2x}{y} \\
&+ yz + \frac{2z}{y} + \frac{1}{x} + \frac{2}{xy}
\\[1ex]
(F_{2})_{Z_y} ={}& x^{2}z + x^{2} + \frac{2x^{2}}{y} + z \\
&+ \frac{2z}{y} + \frac{1}{xy}
\\[1ex]
(F_{2})_{Z_z} ={}& x^{2}y + 2x^{2} + 2x + \frac{x}{y} \\
&+ 2 + \frac{1}{y}
\end{align*}
 
\paragraph{Third column}

\begin{align*}
(F_{3})_{X_x} ={}& 2xyz + 2xy + 2x + \frac{x}{z} \\
&+ 2z + 1 + \frac{2}{z} + \frac{z}{x} + \frac{2}{x}
\\[1ex]
(F_{3})_{X_y} ={}& x^{2}yz + x^{2}y + 2xz + \frac{2x}{z} \\
&+ \frac{x}{y} + \frac{2x}{yz} + 2z + 2 \\
&+ \frac{z}{y} + \frac{1}{y} + \frac{2}{yz} + \frac{1}{xyz}
\\[1ex]
(F_{3})_{X_z} ={}& x^{2}y + \frac{2x^{2}y}{z} + 2xyz + xy \\
&+ 2x + \frac{x}{z^{2}} + 2z + 1
\\[1ex]
(F_{3})_{Z_x} ={}& xyz + 2xy + x + \frac{2x}{z} \\
&+ yz + 2y + z + 2 \\
&+ \frac{1}{x} + \frac{2}{xz}
\\[1ex]
(F_{3})_{Z_y} ={}& x^{2} + \frac{2x^{2}}{z} + z + 2
\\[1ex]
(F_{3})_{Z_z} ={}& 2x^{2}y + \frac{2x^{2}y}{z} + 2x + \frac{x}{z} \\
&+ 2 + \frac{1}{z} + \frac{1}{xz}
\end{align*}

\subsubsection{Flippers for $(ST)^3$}

We list the columns of the Flippers as follows.

\paragraph{First column}
 
\begin{align*}
(F_{1})_{X_x} ={}& 2x^{3}y^{4}z^{5} + x^{2}y^{4}z^{5} + 2x^{2}y^{3}z^{4} + 2x^{2}y^{2}z^{3} \\
&+ xy^{3}z^{4} + 2xy^{2}z^{4} + 2xy^{2}z^{3} + 2xy^{2}z^{2} \\
&+ 2xyz^{2} + xyz + xy + \frac{xy}{z} \\
&+ y^{2}z^{4} + 2y^{2}z^{3} + y^{2}z^{2} + 2yz^{3} \\
&+ 2yz^{2} + 2yz + 2y + \frac{2y}{z} \\
&+ 2z + 2 + \frac{2}{z} + \frac{yz^{3}}{x} \\
&+ \frac{2yz^{2}}{x} + \frac{2z^{2}}{x} + \frac{2z}{x} + \frac{2}{x} \\
&+ \frac{1}{xz} + \frac{z^{2}}{x^{2}} + \frac{2z}{x^{2}} + \frac{2}{x^{2}} \\
&+ \frac{2z}{x^{2}y} + \frac{2}{x^{2}y} + \frac{z}{x^{3}y} + \frac{1}{x^{3}y}
\\[1ex]
(F_{1})_{X_y} ={}& x^{4}y^{4}z^{5} + 2x^{3}y^{4}z^{5} + x^{3}y^{3}z^{5} + x^{3}y^{3}z^{4} \\
&+ 2x^{2}y^{3}z^{4} + x^{2}y^{2}z^{3} + x^{2}y^{2}z^{2} + x^{2}yz^{3} \\
&+ 2xy^{2}z^{4} + xy^{2}z^{3} + 2xy^{2}z^{2} + 2xyz^{2} \\
&+ xy + \frac{xy}{z} + xz^{2} + 2xz \\
&+ 2x + \frac{2x}{z} + 2yz^{3} + yz^{2} \\
&+ \frac{2}{z} + \frac{z}{y} + \frac{1}{y} + \frac{1}{yz} \\
&+ \frac{2z^{2}}{x} + \frac{z}{x} + \frac{1}{x} + \frac{1}{xyz} \\
&+ \frac{2z}{x^{2}y} + \frac{2}{x^{2}y} + \frac{2}{x^{2}y^{2}z^{2}}
\\[1ex]
(F_{1})_{X_z} ={}& 2x^{3}y^{4}z^{5} + x^{3}y^{3}z^{3} + x^{2}y^{3}z^{4} + 2x^{2}y^{2}z^{3} \\
&+ x^{2}y^{2}z^{2} + 2x^{2}y^{2}z + xy^{2}z^{3} + xy^{2}z \\
&+ 2xyz^{2} + 2xyz + \frac{2xy}{z^{2}} + yz^{2} \\
&+ yz + 2z + \frac{1}{z} + \frac{1}{z^{2}} \\
&+ \frac{z}{x} + \frac{2}{xz} + \frac{2}{xyz} + \frac{1}{xyz^{2}} \\
&+ \frac{1}{x^{2}yz}
\end{align*}
\begin{align*}
    (F_{1})_{Z_x} ={}& x^{3}y^{4}z^{5} + 2x^{3}y^{4}z^{4} + x^{2}y^{4}z^{4} + x^{2}y^{3}z^{3} \\
&+ x^{2}y^{2}z^{3} + xy^{3}z^{4} + xy^{3}z^{3} + xy^{2}z^{4} \\
&+ 2xy^{2}z^{3} + 2xy^{2}z^{2} + 2xyz + 2y^{2}z^{4} \\
&+ y^{2}z^{3} + y^{2}z^{2} + yz^{3} + 2yz^{2} \\
&+ yz + y + \frac{2yz^{3}}{x} + \frac{yz^{2}}{x} \\
&+ \frac{z^{2}}{x} + \frac{2}{x} + \frac{1}{xz} + \frac{2}{xy} \\
&+ \frac{1}{xyz} + \frac{2z^{2}}{x^{2}} + \frac{1}{x^{2}y} + \frac{2}{x^{2}yz} \\
&+ \frac{1}{x^{2}yz^{2}}
\\[1ex]
(F_{1})_{Z_y} ={}& 2x^{3}y^{3}z^{4} + x^{3}y^{2}z^{3} + x^{2}y^{3}z^{4} + 2x^{2}y^{2}z^{4} \\
&+ x^{2}y^{2}z^{3} + 2x^{2}yz + xy^{2}z^{4} + xy^{2}z^{3} \\
&+ xyz^{4} + 2xyz^{3} + 2xyz^{2} + xyz \\
&+ 2yz^{4} + yz^{3} + yz^{2} + z^{3} \\
&+ 2z^{2} + \frac{2}{y} + \frac{1}{yz} + \frac{2z^{3}}{x} \\
&+ \frac{z^{2}}{x} + \frac{z^{2}}{xy} + \frac{1}{xy} + \frac{2}{xyz} \\
&+ \frac{2z^{2}}{x^{2}y}
\\[1ex]
(F_{1})_{Z_z} ={}& 2x^{4}y^{4}z^{4} + x^{3}y^{4}z^{4} + 2x^{3}y^{3}z^{3} + x^{2}y^{3}z^{3} \\
&+ x^{2}y^{2}z^{3} + x^{2}y^{2}z^{2} + 2x^{2}y^{2}z + 2xy^{2}z^{3} \\
&+ 2xy^{2}z^{2} + xy^{2}z + xyz^{3} + 2xyz^{2} \\
&+ 2xyz + xy + 2yz^{3} + yz^{2} \\
&+ yz + 2y + z^{2} + 2z \\
&+ 2 + \frac{2}{z} + \frac{2z^{2}}{x} + \frac{z}{x} \\
&+ \frac{1}{x} + \frac{1}{xz} + \frac{z}{xy} + \frac{1}{xy} \\
&+ \frac{1}{xyz^{2}} + \frac{2z}{x^{2}y} + \frac{2}{x^{2}y} + \frac{2}{x^{2}yz^{2}}
\end{align*}

\paragraph{Second column}
 
\begin{align*}
(F_{2})_{X_x} ={}& 2x^{3}y^{4}z^{5} + 2x^{2}y^{3}z^{4} + 2x^{2}y^{2}z^{4} + x^{2}yz^{3} \\
&+ 2xy^{2}z^{4} + xy^{2}z^{3} + 2xy^{2}z^{2} + xyz^{4} \\
&+ 2xyz^{3} + xyz + xz^{2} + 2xz \\
&+ 2x + \frac{2x}{z} + 2yz^{3} + yz^{2} \\
&+ z^{3} + 2z^{2} + 2z + 2 \\
&+ \frac{1}{z} + \frac{z}{y} + \frac{1}{y} + \frac{1}{yz} \\
&+ \frac{2z^{2}}{x} + \frac{z}{x} + \frac{1}{x} + \frac{z^{2}}{xy} \\
&+ \frac{2z}{xy} + \frac{2}{xy} + \frac{1}{xyz} + \frac{2z}{x^{2}y} \\
&+ \frac{2}{x^{2}y} + \frac{z}{x^{2}y^{2}} + \frac{1}{x^{2}y^{2}} + \frac{1}{x^{2}y^{2}z^{2}}
\\[1ex]
(F_{2})_{X_y} ={}& x^{4}y^{4}z^{5} + 2x^{4}y^{3}z^{5} + x^{3}y^{3}z^{4} + 2x^{3}y^{2}z^{4} \\
&+ x^{2}y^{2}z^{4} + 2x^{2}y^{2}z^{3} + x^{2}y^{2}z^{2} + 2x^{2}yz^{4} \\
&+ 2x^{2}yz^{3} + 2x^{2}yz^{2} + 2x^{2}z^{3} + xyz^{3} \\
&+ 2xyz^{2} + 2xz^{3} + 2xz^{2} + 2xz \\
&+ 2x + \frac{2x}{z} + \frac{2xz^{2}}{y} + \frac{xz}{y} \\
&+ \frac{x}{y} + \frac{x}{yz} + z^{2} + 2z \\
&+ 2 + \frac{2z^{2}}{y} + \frac{2z}{y} + \frac{2}{y} \\
&+ \frac{1}{yz} + \frac{2z}{y^{2}} + \frac{2}{y^{2}} + \frac{2}{y^{2}z} \\
&+ \frac{z}{xy} + \frac{1}{xy} + \frac{2z}{xy^{2}} + \frac{2}{xy^{2}}
\\[1ex]
(F_{2})_{X_z} ={}& x^{4}y^{4}z^{5} + 2x^{3}y^{3}z^{4} + 2x^{3}y^{2}z^{3} + 2x^{2}y^{2}z^{3} \\
&+ 2x^{2}y^{2}z + x^{2}yz^{3} + 2x^{2}yz^{2} + x^{2}yz \\
&+ 2xyz^{2} + 2xyz + xz^{2} + xz \\
&+ \frac{x}{z^{2}} + 2z + \frac{1}{z} + \frac{z}{y} \\
&+ \frac{2}{yz} + \frac{2}{yz^{2}} + \frac{2}{xyz} + \frac{1}{xy^{2}z} \\
&+ \frac{2}{xy^{2}z^{2}}
\end{align*}

\begin{align*}
(F_{2})_{Z_x} ={}& 2x^{3}y^{4}z^{4} + x^{3}y^{3}z^{4} + 2x^{2}y^{3}z^{4} + 2x^{2}y^{3}z^{3} \\
&+ x^{2}y^{2}z^{4} + 2x^{2}y^{2}z^{3} + 2x^{2}yz^{3} + xy^{2}z^{4} \\
&+ 2xy^{2}z^{3} + 2xy^{2}z^{2} + 2xyz^{4} + xyz^{3} \\
&+ xyz^{2} + 2xyz + xz + yz^{3} \\
&+ 2yz^{2} + 2z^{3} + z^{2} + \frac{z^{2}}{x} \\
&+ \frac{2z^{2}}{xy} + \frac{2}{xy} + \frac{1}{xyz} + \frac{1}{xy^{2}} \\
&+ \frac{2}{xy^{2}z}
\\[1ex]
(F_{2})_{Z_y} ={}& x^{4}y^{3}z^{5} + 2x^{3}y^{3}z^{4} + 2x^{3}y^{2}z^{4} + 2x^{3}yz^{3} \\
&+ 2x^{2}y^{2}z^{4} + 2x^{2}y^{2}z^{3} + x^{2}yz^{4} + 2x^{2}yz^{3} \\
&+ 2x^{2}yz + x^{2}z + xyz^{4} + 2xyz^{3} \\
&+ 2xyz^{2} + 2xz^{4} + xz^{3} + xz^{2} \\
&+ x + z^{3} + 2z^{2} + \frac{2z^{3}}{y} \\
&+ \frac{z^{2}}{y} + \frac{1}{y} + \frac{2}{yz} + \frac{1}{y^{2}} \\
&+ \frac{2}{y^{2}z} + \frac{z^{2}}{xy} + \frac{2z^{2}}{xy^{2}} + \frac{1}{xy^{2}z^{2}}
\\[1ex]
(F_{2})_{Z_z} ={}& 2x^{4}y^{4}z^{4} + x^{4}y^{3}z^{4} + 2x^{3}y^{3}z^{3} + x^{3}y^{2}z^{3} \\
&+ x^{2}y^{2}z^{3} + x^{2}y^{2}z^{2} + 2x^{2}y^{2}z + 2x^{2}yz^{3} \\
&+ 2x^{2}yz^{2} + x^{2}yz + xyz^{3} + 2xyz^{2} \\
&+ 2xyz + xy + 2xz^{3} + xz^{2} \\
&+ xz + 2x + z^{2} + 2z \\
&+ 2 + \frac{2}{z} + \frac{2z^{2}}{y} + \frac{z}{y} \\
&+ \frac{1}{y} + \frac{1}{yz} + \frac{z}{xy} + \frac{1}{xy} \\
&+ \frac{1}{xyz^{2}} + \frac{2z}{xy^{2}} + \frac{2}{xy^{2}} + \frac{2}{xy^{2}z^{2}}
\end{align*}

 \paragraph{Third column}
 
\begin{align*}
(F_{3})_{X_x} ={}& x^{3}y^{4}z^{4} + x^{2}y^{3}z^{4} + 2x^{2}y^{2}z^{3} + x^{2}y^{2}z^{2} \\
&+ 2xy^{2}z^{4} + 2xy^{2}z^{3} + xy^{2}z + 2xyz^{2} \\
&+ 2xyz + 2yz^{3} + 2yz^{2} + 2yz \\
&+ 2z + \frac{2z^{2}}{x} + \frac{2z}{x} + \frac{2}{xz} \\
&+ \frac{2}{xyz^{2}} + \frac{2z}{x^{2}y} + \frac{1}{x^{2}yz}
\\[1ex]
(F_{3})_{X_y} ={}& 2x^{4}y^{4}z^{4} + 2x^{3}y^{3}z^{4} + x^{2}y^{2}z^{4} + x^{2}y^{2}z^{3} \\
&+ 2x^{2}y^{2}z + x^{2}yz^{3} + 2x^{2}yz^{2} + xyz^{3} \\
&+ xyz^{2} + xyz + xz^{2} + xz \\
&+ z^{2} + z + \frac{1}{z} + \frac{z}{y} \\
&+ \frac{z}{xy} + \frac{2}{xyz} + \frac{1}{xy^{2}z^{2}}
\\[1ex]
(F_{3})_{X_z} ={}& x^{3}y^{3}z^{3} + 2x^{3}y^{3}z^{2} + 2x^{2}y^{2}z^{3} + 2x^{2}y^{2}z^{2} \\
&+ x^{2}y^{2}z + x^{2}y^{2} + 2xyz^{2} + xy \\
&+ 2z + 1 + \frac{1}{z} + \frac{2}{z^{2}} \\
&+ \frac{2}{xyz} + \frac{1}{xyz^{2}}
\end{align*}

\begin{align*}
(F_{3})_{Z_x} ={}& 2x^{3}y^{4}z^{4} + x^{3}y^{4}z^{3} + 2x^{2}y^{3}z^{4} + x^{2}y^{3}z^{2} \\
&+ x^{2}y^{2}z^{3} + 2x^{2}y^{2}z^{2} + xy^{2}z^{4} + xy^{2}z^{3} \\
&+ xy^{2}z + 2xyz + xy + yz^{3} \\
&+ yz^{2} + yz + \frac{z^{2}}{x} + \frac{2z}{x} \\
&+ \frac{2}{xy} + \frac{2}{xyz} + \frac{2}{xyz^{2}}
\\[1ex]
(F_{3})_{Z_y} ={}& 2x^{3}y^{3}z^{4} + x^{3}y^{3}z^{3} + x^{3}y^{2}z^{3} + 2x^{3}y^{2}z^{2} \\
&+ 2x^{2}y^{2}z^{4} + x^{2}y^{2}z^{2} + 2x^{2}yz + x^{2}y \\
&+ xyz^{4} + xyz^{3} + xyz + z^{3} \\
&+ z^{2} + z + \frac{2}{y} + \frac{2}{yz} \\
&+ \frac{2}{yz^{2}} + \frac{z^{2}}{xy} + \frac{2z}{xy}
\\[1ex]
(F_{3})_{Z_z} ={}& x^{4}y^{4}z^{3} + x^{2}y^{2}z^{3} + x^{2}y^{2}z^{2} + x^{2}y^{2}z \\
&+ x^{2}y^{2} + xyz^{3} + xyz^{2} + 2xy \\
&+ z^{2} + z + \frac{2}{z} + \frac{2}{z^{2}} \\
&+ \frac{z}{xy} + \frac{2}{xyz} + \frac{1}{xyz^{2}}
\end{align*}

\subsection{Explicit commutation matrices}

We have
\begin{center}
\resizebox{\textwidth}{!}{%
$\displaystyle
\Xi_{(ST)^3}=
\begin{pmatrix}

\dfrac{-y^{2}z^{2}+yz^{2}-y+1}{yz}
&
\dfrac{y^{2}z^{2}-yz^{2}-y+1}{yz}
&
\dfrac{1-y}{yz}
&
\dfrac{y^{2}z^{3}-yz^{3}+yz^{2}+1}{y^{2}z^{2}}
&
\dfrac{-y^{2}z^{3}+yz^{3}+yz^{2}-z+1}{y^{2}z^{2}}
&
\dfrac{yz^{2}-z+1}{y^{2}z^{2}}
\\[8pt]

\dfrac{-y^{2}z^{2}+yz^{2}+y-1}{yz}
&
\dfrac{y^{2}z^{2}-yz^{2}+y-1}{yz}
&
\dfrac{y-1}{yz}
&
\dfrac{y^{2}z^{3}-yz^{3}-z-1}{y^{2}z^{2}}
&
\dfrac{-y^{2}z^{3}+yz^{3}+yz^{2}-z-1}{y^{2}z^{2}}
&
\dfrac{yz^{2}+z-1}{y^{2}z^{2}}
\\[8pt]

-yz+z
&
yz-z
&
0
&
\dfrac{y^{2}z^{2}-yz^{2}+1}{y^{2}z}
&
\dfrac{-y^{2}z^{2}+yz^{2}+1}{y^{2}z}
&
\dfrac{yz+1}{y^{2}z}
\\[8pt]

\dfrac{-y^{2}z^{3}-yz+y-1}{z}
&
\dfrac{y^{2}z^{3}+y^{2}z^{2}+y-1}{z}
&
\dfrac{-y^{2}z^{2}+y-1}{z}
&
\dfrac{y^{2}z^{4}-y^{2}z^{3}+z-1}{yz^{2}}
&
\dfrac{-y^{2}z^{4}+y^{2}z^{3}+z-1}{yz^{2}}
&
\dfrac{z-1}{yz^{2}}
\\[8pt]

\dfrac{-y^{2}z^{3}+y^{2}z^{2}-yz-y+1}{z}
&
\dfrac{y^{2}z^{3}+y^{2}z^{2}-yz-y+1}{z}
&
\dfrac{-y^{2}z^{2}-y+1}{z}
&
\dfrac{y^{2}z^{4}-y^{2}z^{3}-z+1}{yz^{2}}
&
\dfrac{-y^{2}z^{4}+y^{2}z^{3}-z+1}{yz^{2}}
&
\dfrac{1-z}{yz^{2}}
\\[8pt]

-y^{2}z^{2}+y^{2}z-y
&
y^{2}z^{2}-y^{2}z-y
&
-y^{2}z-y
&
yz^{2}-yz
&
-yz^{2}+yz
&
0

\end{pmatrix}.
$}
\end{center}

The complete matrix of $E$ is
\[
E =
\begin{pmatrix}
0 & 0      & 0    & y^{-1} & E_{15} & E_{16} \\
0 & 0      & 0    & y^{-1} & E_{25} & E_{26} \\
1 & 0      & 0    & 0      & E_{35} & E_{36} \\
0 & -y^{2}z & yz+1 & 0      & E_{45} & E_{46} \\
0 & -y^{2}z & yz+1 & 0      & E_{55} & E_{56} \\
0 & 0      & 1    & 0      & E_{65} & E_{66}
\end{pmatrix},
\]
where, element by element,
\begin{align*}
E_{15} &= \frac{-yz^{3} + yz^{2} - z + 1}{y^{2}z^{2}}, \\[1ex]
E_{16} &= \frac{y^{2}z^{2} - yz^{2} - yz + y - 1}{y^{2}z^{2}}, \\[1ex]
E_{25} &= \frac{-yz^{3} + yz - z + 1}{y^{2}z^{2}}, \\[1ex]
E_{26} &= \frac{y^{2}z - yz^{2} + yz + y - 1}{y^{2}z^{2}}, \\[1ex]
E_{35} &= \frac{-y^{2}z^{3} + y^{2}z^{2} - yz^{3} + yz - z + 1}{y^{2}z^{2}}, \\[1ex]
E_{36} &= \frac{-y^{3}z^{3} + y^{3}z^{2} - y^{2}z^{2} + y^{2}z - yz^{2} + yz + y - 1}{y^{2}z^{2}}, \\[1ex]
E_{45} &= \frac{-z^{2} + z + 1}{z}, \\[1ex]
E_{46} &= \frac{yz + y - z - 1}{z}, \\[1ex]
E_{55} &= \frac{yz - z^{2} + z + 1}{z}, \\[1ex]
E_{56} &= \frac{y^{2}z + y - z - 1}{z}, \\[1ex]
E_{65} &= \frac{y^{2}z^{2} + yz - z^{2} + z + 1}{z},
\end{align*}
and
\[
E_{66} = \frac{y^{3}z^{2} - y^{2}z^{2} + y^{2}z + y - z - 1}{z}.
\]
 
All denominators are monomials, so every entry is a finite Laurent
polynomial. This is therefore a local change of boundary generators.
 
Its determinant is $det E = yz$, which is an invertible Laurent monomial. Thus $E$ is an
allowed local basis transformation.
 
The exact matrix identity is
\[
E^{\dagger}\,\Xi_{(ST)^3}\,E =
\begin{pmatrix}
0  & 1 & 0  & 0 & 0 & 0 \\
-1 & 0 & 0  & 0 & 0 & 0 \\
0  & 0 & 0  & 1 & 0 & 0 \\
0  & 0 & -1 & 0 & 0 & 0 \\
0  & 0 & 0  & 0 & z^{-1} - z & -1 + z^{-1} + y + yz^{-1} \\
0  & 0 & 0  & 0 & -z + 1 - zy^{-1} - y^{-1} & y - y^{-1}
\end{pmatrix}.
\]
 
The upper-left and middle $2 \times 2$ blocks are
\[
\lambda_{1} =
\begin{pmatrix}
0 & 1 \\
-1 & 0
\end{pmatrix},
\]
and the lower-right block is exactly
\[
\Xi_{(ST)} =
\begin{pmatrix}
z^{-1} - z & -1 + z^{-1} + y + yz^{-1} \\
-z + 1 - zy^{-1} - y^{-1} & y - y^{-1}
\end{pmatrix}.
\]
 
Therefore, without any suppressed dimensional step,
\[
\boxed{\,E_{\mathrm{full}}^{\dagger}\,\Xi_{(ST)^3}\,E_{\mathrm{full}}
  = \lambda_{1} \oplus \lambda_{1} \oplus \Xi_{(ST)}.\,}
\]

\subsection{Ground state degeneracy comparison}

Consider the Hamiltonian constructed from (1) taking the stabilizers generated by $\{X_e\}$ and performing on these stabilizers $ST$ to get $\{B_p\}$ and (2) taking the stabilizers generated by $\{X_e\}$ and performing $(ST)^3$ to get $\{\tilde{B}_p\}$. These stabilizers are illustrated in Fig.~\ref{fig:Bpcompare}. If each $B_p$ were unitarily related to a corresponding $\tilde{B}_p$, then a Hamiltonian with ground state space satisfying $\{B_p=1\}$ vs $\{\tilde{B}_p=1\}$ would have the same ground state degeneracy (this may not be equal to 1 until after we add the vertex terms).

\begin{figure}[h]
   \centering
   \includegraphics[width=.8\columnwidth]{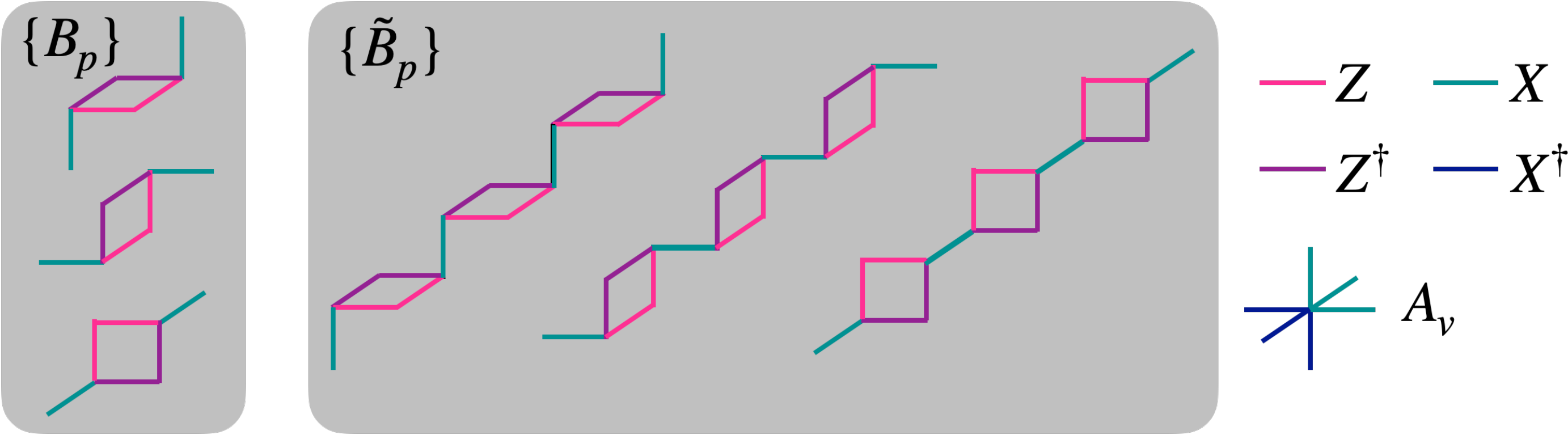} 
   \caption{Two sets of stabilizers obtained from $ST$ and $(ST)^3$, respectively.}
   \label{fig:Bpcompare}
   \end{figure}

For a periodic lattice with $L^3$ vertices, there are $3L^3$ edges and hence $3L^3$ qutrits. If there are $r$ independent $\mathbb{Z}_3$ stabilizers, then the common $+1$ eigenspace has dimension $3^{3L^3-4}$. Omitting the $A_v$ term, we get
\begin{align}
    \log_3\mathrm{GSD}_{ST}&=\left\{
    \begin{array}{cl}
         0&   L\text{ odd}\\
          L^2& L\text{ even}
    \end{array}\right.
    \cr 
    \log_3\mathrm{GSD}_{(ST)^3}&=\left\{
   \begin{array}{cl}
        0&   L\text{ odd}\\
          L^2& L=2\text{ mod }4\\
          3L^2 & L=0\text{ mod }4
   \end{array} \right. ~.
\end{align}
The GSD comes from the fact that taking the product of $B_p$ around a cube gives $\prod_{p\in c}B_p^{s(p)}=A_vA_{v+(1,1,1)}=1$ ($s(p)=\pm1$ depending on the orientation of $p$, and $A_{v+(1,1,1)}$ is shifted from $A_v$ by one unit in $\hat{x},\hat{y}$ and $\hat{z}$) and without the additional $A_v$ stabilizers, $A_vA_{v+(1,1,1)}=1$ can be satisfied in three different ways for $\mathbb{Z}_3$ qutrits. On a lattice with $L^3$ vertices, we get $A_{v}=A_{v+(1,1,1)}=A_{v+(2,2,2)}=\cdots$ coming from taking products of $B_p$ terms, so we get a total ground state degeneracy of $L^3$ vertices divided by $L$ vertices in each orbit $=L^2$ value for $\log_3$GSD. In other words, GSD=$3^{L^2}$). A similar calculation gives the results above. Clearly, since the GSD mismatch on $L=0$ mod 4, the stabilizers $\{B_p\}$ cannot be unitarily related to $\{\tilde{B}_p\}$. On the other hand, including the $A_v$ terms gives a unique ground state on any closed manifold, allowing for the possibility that the two unique ground states admit parent Hamiltonians related by a FDQC.
\bibliographystyle{utphys}
\bibliography{biblio}

\end{document}